\documentclass{birkjour}

\usepackage{hyperref}
\usepackage{amsfonts,amssymb,amsmath,exscale,bbm}
\usepackage{empheq}
\usepackage{footnpag} % Footnotes start with 1 at each page
\usepackage[all,knot]{xy}  % xy diagrams
\usepackage{color}
\usepackage{graphicx}
\usepackage{epsfig}
\usepackage{subfig}
\usepackage{enumerate}

\usepackage{bbold}

\usepackage{amsmath}
\usepackage{amsfonts}
\usepackage{graphicx}
\usepackage{psfrag}
\usepackage{array}
\usepackage{bbm}
\usepackage{hyperref}
\usepackage{amsthm}
\theoremstyle{definition}
\newtheorem{definition}{Definition}

\newtheorem{proposition}{Proposition}

\newcommand{\cA}{{\mathcal A}}
\newcommand{\cB}{{\mathcal B}}

\newcommand{\cG}{{\mathcal G}}
\newcommand{\cH}{{\mathcal H}}
\newcommand{\cI}{{\mathcal I}}

\newcommand{\cM}{{\mathcal M}}

\newcommand{\cO}{{\mathcal O}}

\newcommand{\cS}{{\mathcal S}}
\newcommand{\cT}{{\mathcal T}}

\newcommand{\cZ}{{\mathcal Z}}

\def\inv{{\mbox{\tiny -1}}}

\newcommand{\fa}{\mathbb{a}}
\newcommand{\fb}{\mathbb{b}}
\newcommand{\fc}{\mathbb{c}}
\newcommand{\fd}{\mathbb{d}}
\newcommand{\fe}{\mathbb{e}}
\newcommand{\ff}{\mathbb{f}}
\newcommand{\fg}{\mathbb{g}}

\newcommand\beq{\begin{equation}}
\newcommand\eeq{\end{equation}}
\newcommand{\be}{\begin{equation}}
\newcommand{\ee}{\end{equation}}
\newcommand{\bes}{\begin{eqnarray}}
\newcommand{\ees}{\end{eqnarray}}

\def\vphi{{\varphi}}
\def\vphib{\overline{{\varphi}}}

\newcommand{\one}{\mbox{$1 \hspace{-1.0mm}  {\bf l}$}}
\def\Phib{\overline{\Phi}}
\def\omegab{\overline{\omega}}
\def\tS{\widetilde{S}}
\def\tC{\widetilde{C}}
\def\tu{\tilde{u}}
\def\tt{\tilde{t}}
\def\bu{{\bf{u}}}

\newcommand{\su}{\mathfrak{su}}

\newcommand{\SU}{\mathrm{SU}}

\newcommand{\U}{\mathrm{U}}

\newcommand\acts\triangleright

\def\e{\mbox{e}}
\def\extd{\mathrm {d}}

\def\nn{{\nonumber}}

\newcommand\restr[2]{{% we make the whole thing an ordinary symbol
  \left.\kern-\nulldelimiterspace % automatically resize the bar with \right
  #1 % the function
  \vphantom{\big|} % pretend it's a little taller at normal size
  \right|_{#2} % this is the delimiter
  }}

\newtheorem{theo}{Theorem}%[section]
\newtheorem{lemma}[theo]{Lemma}

\begin{document}

% title page and abstract %

\title[Discrete Renormalization Group for $\SU(2)$ TGFT]{Discrete Renormalization Group for $\SU(2)$ Tensorial Group Field Theory}

\author{Sylvain Carrozza}\email{sylvain.carrozza@cpt.univ-mrs.fr}\address{Centre de Physique Th\'eorique\\ 
CNRS UMR7332, Universit\'e d'Aix-Marseille, 13288 Marseille cedex 9, France}

\begin{abstract}
\bigskip
This article provides a Wilsonian description of the perturbatively renormalizable Tensorial Group Field Theory introduced in arXiv:1303.6772 [hep-th] (Commun. Math. Phys. 330, 581-637). It is a rank-3 model based on the gauge group $\SU(2)$, and as such is expected to be related to Euclidean quantum gravity in three dimensions. By means of a power-counting argument, we introduce a notion of dimensionality of the free parameters defining the action. General flow equations for the dimensionless bare coupling constants can then be derived, in terms of a discretely varying cut-off, and in which all the so-called melonic Feynman diagrams contribute. Linearizing around the Gaussian fixed point allows to recover the splitting between relevant, irrelevant, and marginal coupling constants. Pushing the perturbative expansion to second order for the marginal parameters, we are able to determine their behaviour in the vicinity of the Gaussian fixed point. Along the way, several technical tools are reviewed, including a discussion of combinatorial factors and of the Laplace approximation, which reduces the evaluation of the amplitudes in the UV limit to that of Gaussian integrals.  
\end{abstract}

% main body %

\maketitle

%%%%%%%%%%%%%%%%%%% Intro %%%%%%%%%%%%%%%%%%%%%%%%%%%%%%%%%%%%%%%%%%%%%%%%%%%%%%%%%
\section{Introduction}

Group Field Theory (GFT) \cite{freidel_gft, daniele_rev2006, daniele_rev2011, Krajewski_rev} is an approach to Quantum Gravity which lies at the crossroad of Tensor Models \cite{Ambjorn_tensors, gross, Sasakura:1990fs, razvan_jimmy_rev, tt2}, Loop Quantum Gravity (LQG) \cite{ashtekar_book, rovelli_book, thiemann_book, gambini_pullin_book, bojowald_book}, and Spin Foam Models \cite{perez_review2012, Baez:1997zt, oriti2001spacetime, Alexandrov:2011ab}. This research program aims at addressing questions which are notoriously difficult in loop quantum gravity, such as the construction of the dynamical sector of the theory and of its continuum limit \cite{daniele_hydro}, by means of quantum and statistical field theory techniques. A GFT is nothing but a field theory defined on a (compact) group manifold, with specific non-local interactions. The latter are chosen in such a way that the Feynman expansion generates cell complexes of a given dimension, interpreted as discrete space-time histories. In particular, the amplitudes of any spin foam model can be generated by a suitably constructed GFT, which is what triggered interest in the GFT formalism \cite{dPFKR, GFT_rovelli_reisenberg}. Hence GFTs can be viewed as a natural way of completing the definition of spin foam models, which by themselves do not associate unambiguous amplitudes to boundary states. Note that alternatively to the summing strategy implemented in GFT, one can instead look for a definition of the refinement limit of spin foams \cite{bianca_cyl, bianca_review, bahr2012}. A more direct route from canonical LQG to GFT has recently been proposed \cite{daniele_2nd}: from this perspective, the GFT formalism defines a second-quantized version of LQG, and hence should be especially relevant to the analysis of the many-body sector of the theory (see \cite{gfc_letter, gfc_full} for cosmological applications of these ideas). 

Thanks to recent breakthroughs in the field of tensor models, triggered by the pioneering work of Gurau \cite{razvan_colors}, who found a generalization of the $1/N$ expansion of matrix models \cite{RazvanN, RazvanVincentN, razvan_complete}, standard field theory techniques are currently being developed and generalized to more and more complicated GFTs. The common ingredient to all the models studied so far is \emph{tensor invariance} \cite{r_vir, universality, uncoloring, tt2, tt3, tt4}: it provides a generalized notion of locality for GFTs, which as we have just mentioned are non-local in the ordinary sense. They then differ by: the rank of the tensor fields, identified with the space-time dimension; the space in which the indices of the tensor fields live; and the propagator. Theories based on tensors with indices in the integers\footnote{Equivalently, such models can also be viewed as GFTs on the group $\U(1)$.} and trivial propagator are referred to as Tensor Models. A wealth of results about their $1/N$ expansions has now been accumulated \cite{critical, dr_dually, v_revisiting, v_new, jr_branched}, from single to multiple scalings in $N$ \cite{wjd_double, stephane_double, bgr_double}, and up to the rigorous non-perturbative level \cite{razvan_beyond, delepouve_borel}. Similar models with non-trivial propagators are referred to as Tensorial Field Theories \cite{tensor_4d, josephsamary, joseph_d2, joseph_etera}: they are indeed genuine field theories, for which full-fledge renormalization methods take over the large $N$ expansion. Finally, Tensorial Group Field Theories (TGFT) are Tensorial Field Theories based on (compact) Lie groups\footnote{In view of the correspondence between simple GFTs and tensor models resulting from harmonic analysis, we have in mind a theory in which the group structure plays a significant enough role.}. The only models of this type available in the literature so far are the so-called TGFTs with \emph{gauge invariance condition} \cite{cor_u1, cor_su2, samary_vignes, thesis}. This paper will focus on one such example, a rank-$3$ renormalizable model based on the gauge group $\SU(2)$. In addition to being technically relevant to spin foam models and LQG in general, it is expected to be related to Euclidean quantum gravity in three dimensions, though the full correspondence is unclear at present: while tensor invariance seems to provide a perfectly viable discretization prescription, the quantum gravity interpretation (if any) of the Laplace-type propagator has not been investigated in details.  

The aim of this article is two-folds. The first objective is to pave the way towards general renormalization group techniques \`a la Wilson, which will allow to better understand the theory spaces of TGFTs, their flows and fixed points. Secondly, we want to determine the properties of the Gaussian fixed point of the specific model we are considering. Indeed, the examples worked out so far \cite{josephaf, samary_beta} suggest that asymptotic freedom might be a reasonably generic property of Tensorial Field Theories. The occurrence of asymptotic freedom in such models is surprising at first, because they are not gauge theories. However, the non-locality of the interactions is responsible for wave-function renormalization terms which are absent from ordinary scalar field theories, and which typically dominate over the coupling constants renormalization terms. Hence the $\beta$-functions can be negative despite the absence of non-Abelian gauge symmetries.

In section \ref{sec:model}, we introduce the model and the main results of \cite{cor_su2} which are relevant to the present publication. In addition, we provide a detailed analysis of the symmetry factors appearing in the Feynman expansion, which significantly simplifies the calculations reported on in the later sections. In section \ref{sec:rgflow} we introduce a discrete version of Wilson's renormalization group. It is based on the introduction of dimensionless coupling parameters, and differs in this sense from the analyses performed in previous works \cite{josephaf, samary_beta, cor_su2}. A new notion of reducibility is also proposed, which is to some extent the correct generalization of $1$-particle reducibility from ordinary local field theories to TGFTs. In section \ref{sec:gaussian}, we linearize the flow equations in the vicinity of the Gaussian fixed point. This allows to classify the coupling constants in terms of their relevance, and recover the fact that this model is renormalizable up to order six interactions. We then explicitly compute the eigendirections associated to this linear system, and deduce the functional relationship between the marginal coupling constants and the other renormalizable constants in the asymptotic UV region. In section \ref{sec:as}, the flow equations for the marginal coupling constants are pushed to second order in perturbation theory. Although the calculations underlying the results of this section are rather lengthy and technical, they are as far as we know the first of their kind in a non-Abelian model and are therefore included in full details. Finally, the qualitative properties of the flow equations in the vicinity of the Gaussian fixed point are investigated in section \ref{sec:portrait}. Relying on a continuous version of the discrete flow, we will be able to completely settle the question of asymptotic freedom when $u_{6,1}$ and $u_{6,2}$ have same signs, and in particular when they are strictly positive.

% 
% This is the technically challenging section of the paper, sometimes involving lengthy computations of counter-terms up to second order in the coupling constants. This detailed investigation not only allows us to determine the $\beta$-functions (or more precisely their counter-parts in the discrete version of the renormalization group) of the model, and therefore investigate the question of asymptotic freedom, but also provides the first illustration of such a computation in a non-Abelian model.  

%%%%%%%%%%%%%%%%%%% Background %%%%%%%%%%%%%%%%%%%%%%%%%%%%%%%%%%%%%%%%%%%%%%%%%%%%
\section{The model and its divergences}\label{sec:model}

In this section, we introduce the TGFT model studied in \cite{cor_su2}, and summarize some of the key steps in the proof of its renormalizability. This paper was to a large extent based on a Bogolioubov recursion relation, which defined the renormalized series. The analysis was greatly simplified, thanks to a refined notion of graph connectedness, called \textit{face-connectedness}. As already noted in \cite{cor_su2}, this structure is not appropriate in the effective Wilsonian language, where connected graphs in the ordinary sense\footnote{This ordinary notion was called \textit{vertex-connectedness} in \cite{cor_su2}.} must be summed over. Since the purpose of this paper is to investigate further the renormalization group flow equations of this model, we will only outline the definition of the effective series, in which the usual graph-theoretic notion of connectedness is at play.  

\subsection{Definition and Feynman expansion}

We are interested in a group field theory for a field $\vphi \in L^{2} (\SU(2)^{\times 3})$ and its complex conjugate $\vphib$. The free theory is defined by a Gaussian measure $\extd \mu_C$, with covariance $C$, that is to say:
\beq
\int \extd \mu_C (\vphi , \vphib) \, \vphi (g_1, g_2 , g_3) \vphib ( g_1' , g_2' , g_3') = C(g_1 , g_2 , g_3 ; g_1' , g_2' , g_3')\,.
\eeq
This covariance can be perturbed by non-Gaussian interaction terms, encapsulated in an action $S(\vphi , \vphib)$, so as to define the (Euclidean) interacting partition function:
\beq
\cZ \equiv \int \extd \mu_C (\vphi , \vphib) \, \exp( - S (\vphi , \vphib) )\,.
\eeq

\

In tensorial GFTs, $S$ is assumed to be a weighted sum of \textit{connected tensor invariants}. By analogy with space-time based quantum field theories, this prescription is referred to as a \textit{locality principle}, and is used as such for renormalization purposes. Connected tensor invariants in dimension $d$ are in one-to-one correspondence with $d$-\textit{colored graphs} (also called $d$-\textit{bubbles}), which are bipartite edge-colored closed graphs with fixed valency $d$ at each node. In our $3$-dimensional context, a $3$-colored graph is a connected graph with two types of nodes (black or white), and edges labeled by integers $\ell \in \{ 1, 2 , 3\}$ (the colors), in such a way that: a) any edge connects a white node to a black one; b) at any node, exactly three edges with distinct colors meet. Simple examples are provided in Figure \ref{ex_coloredgraphs}. 
\begin{figure}[h]
\begin{center}
\includegraphics[scale=0.5]{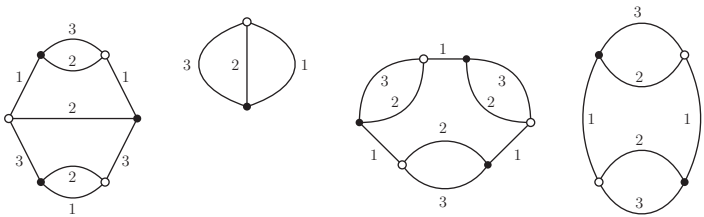}
\caption{Examples of $3$-colored graphs.}
\label{ex_coloredgraphs}
\end{center}
\end{figure}
The unique invariant $I_b (\vphi , \vphib)$ associated to a given $3$-bubble $b$ is constructed as follows: a) each white (resp. black) node represents a field $\vphi$ (resp. a conjugate field $\vphib$); b) a color-$\ell$ edge between the nodes $n$ and $\overline{n}$ indicates a convolution with respect to the $\ell^{{\rm{th}}}$ variables of the fields located at $n$ and $\overline{n}$ respectively. 

\noindent{\bf Example.} The colored graph on the right side of Figure \ref{ex_coloredgraphs} represents the following invariant integral:
\beq
\int [\extd g ]^6 \, \vphi(g_1 , g_2 , g_3 ) \vphib(g_4 , g_2 , g_3 ) \vphi(g_4 , g_5 , g_6 ) \vphib(g_1 , g_5 , g_6 )\,. 
\eeq

\noindent With these definitions, the interaction part of the action $S$ can be written as
\beq\label{action1}
S(\vphi , \vphib) = \sum_{b \in \cB} \frac{t_{b}}{k(b)} I_b (\vphi , \vphib)\,,
\eeq
where $\cB$ is the set of all bubbles, and $t_b \in \mathbb{R}$ is the coupling constant associated to $b$. In order to simplify the counting of Feynman graphs, we divided each coupling constant $t_b$ by a combinatorial factor $k(b)$, defined as the \emph{number of automorphisms} of the bubble graph $b$. 
\begin{definition}
Let $b$ be a colored graph. An \emph{automorphism} of $b$ is a permutation $\sigma$ of its nodes\footnote{This definition coincides with the more general concept of graph automorphism, even if a graph automorphism is usually thought of as a couple $(\sigma_n , \sigma_e)$ of permutations, respectively of the nodes and of the edges. When imposing compatibility with the colored structure, $\sigma_e$ becomes redundant, hence our definition in terms of a single permutation.}, such that: 
\begin{enumerate}[(i)]
\item $\sigma$ conserves the nature of the nodes; 
\item if $n$ and $\overline{n}$ are connected by an edge of color $\ell$, then so do $\sigma(n)$ and $\sigma(\overline{n})$.
\end{enumerate}
\end{definition}
At this stage, it is natural to assume that the colors have no physical role other than imposing combinatorial restrictions on the interactions, and we will therefore require $S$ to be invariant under color permutations. This can be formalized as follows. The group $S_3$ of permutations of the color set $\{1,2,3\}$ acts naturally on $\cB$: for any $b \in \cB$, $\sigma.b$ is the bubble obtained from $b$ by permutation of the color labels as dictated by $\sigma$. The invariance of $S$ under color permutation is simply the statement that:
\beq
\forall b \in \cB\,, \; \forall \sigma \in S_3 \,, \; t_{\sigma.b} = t_b\,.
\eeq

\

The second ingredient of the model is the covariance $C$. In \cite{cor_u1, cor_su2}, it was motivated from two basic requirements: first, it should impose the so-called \textit{gauge invariance condition} of spin foam models; second, it should have a rich enough spectrum so as to provide an abstract notion of scale. The gauge invariance condition is an invariance of the field\footnote{Note that there is no gauge symmetry involved at the field theory level, and that the nomenclature arises from the lattice gauge theory interpretation of the \emph{amplitudes}.} under an arbitrary simultaneous translation of its variables:
\beq\label{gauge_invariance}
\forall h \in \SU(2), \qquad \vphi( g_1 h , g_2 h , g_3 h ) = \vphi (g_1 , g_2 , g_3 )\,.
\eeq
In the GFT context we can however work with generic fields, and impose (\ref{gauge_invariance}) through the covariance. The latter should therefore contain the projector $P$ on the space of gauge invariant fields, defined by the kernel:   
\beq
P( g_1 , g_2 , g_3 ; g_1' , g_2' , g_3' ) = \int \extd h \, \prod_{i = 1}^{3} \delta( g_i h g_i'^{\inv})\,.
\eeq
In order to get a non-trivial spectrum, we combine it with the operator 
\beq
\widetilde{C} \equiv \left( m^2 - \sum_\ell \Delta_\ell \right)^{-1} \,,
\eeq
where $\Delta_\ell$ denotes the Laplace operator on $\SU(2)$ acting on the color-$\ell$ variables. At this stage, this should be seen as a conservative natural choice, partially motivated by a study of the $2$-point radiative corrections of the Boulatov-Ooguri models \cite{Valentin_Joseph}, and formal analogies with the Osterwalder-Schrader axioms \cite{tt2, tt3}. The covariance $C$ is then defined as\footnote{It is easily seen that $\widetilde{C}$ and $P$ commute.}
\beq
C = \widetilde{C} P = P \widetilde{C} \,,
\eeq  
whose kernel can be written in the Schwinger representation:
\beq\label{def_propa}
C ( g_1 , g_2 , g_3 ; g_1' , g_2' , g_3' ) = \int \extd \alpha \, \e^{- m^2 \alpha} \int \extd h \, \prod_{i=1}^{3} K_{\alpha} ( g_i h g_i'^{\inv})\,, 
\eeq
where $K_\alpha$ is the heat kernel on $\SU(2)$ at time $\alpha$. This covariance is ill-defined at coinciding points $g_i = g_i'$, and this can be regulated by cutting-off the contribution of the integral over $\alpha$ in the neighborhood of $0$. The divergences in this cut-off, generic in the perturbation expansion of the theory, can be analyzed in details. 

\

The unnormalized $N$-point functions $\cZ\cS_N$ can be expanded perturbatively in terms of Feynman amplitudes, which are indexed by generalized $4$-colored graphs. The new type of lines, of color $0$ and represented as dashed lines, is introduced to represent the propagators; they connect the elementary bubble vertices generated by $S$. These lines can also appear as external legs, contrary to the internal bubble lines, and this is the only respect in which the Feynman graphs differ from generic $4$-colored graphs. Given a Feynman graph $\cG$, we define $N(\cG)$ as the number of external legs, and $n_b (\cG)$ as the number of bubble vertices of type $b$. Then 
\beq
\cZ \cS_N = \sum_{\cG \vert N(\cG) = N} \frac{1}{k(\cG)} \left( \prod_{b \in \cB} (- t_b )^{n_b (\cG)} \right)\, \cA_\cG \,,
\eeq
where $k(\cG)$ is a symmetry factor associated to $\cG$, and $\cA_\cG$ its amplitude. A simple counting shows that $k(\cG)$ is nothing but the number of automorphisms of $\cG$, which generalizes the definition of $k(b)$ for bubbles (hence the notation). 
\begin{definition}
Let $\cG$ be a Feynman graph. An \emph{automorphism} of $\cG$ is a permutation $\sigma$ of its nodes, such that: 
\begin{enumerate}[(i)]
\item $\sigma$ conserves the nature of the nodes; 
\item if $n$ and $\overline{n}$ are connected by an edge of color $\ell$, then so do $\sigma(n)$ and $\sigma(\overline{n})$;
\item if a node $n$ (resp. $\overline{n}$) is connected to an external leg $l$, then $\sigma(n) = n$ (resp. $\sigma(\overline{n}) = \overline{n}$).
\end{enumerate}
\end{definition}
\begin{proposition}
The symmetry factor $k(\cG)$ associated to an arbitrary Feynman graph $\cG$ is the order of its group of automorphisms. 
\end{proposition}
\begin{proof}
Consider a Feynman graph $\cG$, with $N$ labeled external legs, and $n_b$ bubbles of type $b$ for any $b \in \cB$. Define then the group
\beq
G = \prod_{b \in \cB \vert n_b \neq 0} {\mathrm{Aut}}(b)^{n_b} \times S_{n_b}  \,,
\eeq
where $S_{n}$ is the permutation group of $n$ elements, and ${\mathrm{Aut}}(b)$ is the group of automorphisms of $b$. Its order is
\beq
\vert G \vert = \prod_{b \in \cB} k(b)^{n_b} n_b !  
\eeq
One can fix an arbitrary labeling of all the nodes of all the vertices appearing in $\cG$. Any Wick contraction with the same vertices appearing in the Feynman expansion is then represented by a set of $N + \sum_b n_b N_b$ pairs of labels. Call $X$ the set of all possible such pairings. $G$ acts naturally on $X$, by permutation of identical bubbles and by automorphism on each individual bubble. The Feynman graph $\cG$ can be identified with the orbit $G.x$ of some $x \in X$, and the stabilizer $G_x$ to the group of automorphisms ${\mathrm{Aut}}(\cG)$. Taking the factors coming from the perturbative expansion of the exponentials and from the definition of the action into account, we therefore obtain:
\bes
\frac{1}{k(\cG)} &=& \vert G.x \vert \times\prod_{b \in \cB} \left(\frac{1}{k(b)}\right)^{n_b} \frac{1}{n_b !} 
 = \frac{\vert G \vert}{\vert G_x \vert} \times\prod_{b \in \cB} \left(\frac{1}{k(b)}\right)^{n_b} \frac{1}{n_b !} \nn \\ 
&=& \frac{1}{\vert G_x \vert} = \frac{1}{\vert {\mathrm{Aut}}(\cG) \vert}\,.
\ees
\end{proof}

It is immediate to remark that, because of the colored structure of the Feynman graphs, the symmetry factor of any connected and non-vacuum graph (i.e. with at least one external leg) $\cG$ is simply $k(\cG) = 1$. This implies in particular that the connected and normalized Schwinger functions $\cS^{(c)}_N$ expand as: 
\beq
\cS^{(c)}_N = \sum_{\cG \; \rm{connected}\vert N(\cG) = N} \left( \prod_{b \in \cB} (- t_b )^{n_b (\cG)} \right)\, \cA_\cG \,.
\eeq

\noindent{\bf{Examples.}} The graphs $\cG_1$, $\cG_2$ and $\cG_3$ represented in Figure \ref{ex_sym} have different symmetry factors. $\cG_1$ is a connected and non-vacuum graph with external legs, hence $k(\cG_1) = 1$. On the other hand, its amputated version $\cG_2$ admits one non-trivial automorphism, and therefore $k(\cG_2) =2$. Finally, the vacuum graph $\cG_3$ has a symmetry factor $k(\cG_3) = 3$.  
\begin{figure}[h]
  \centering
  \subfloat[$\cG_1$]{\label{symg1}\includegraphics[scale=0.5]{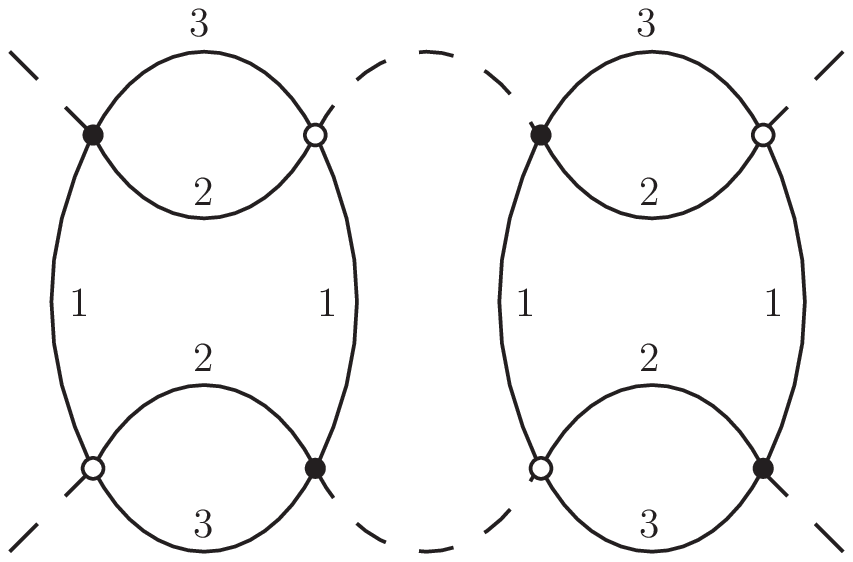}}     
  \subfloat[$\cG_2$]{\label{symg2}\includegraphics[scale=0.5]{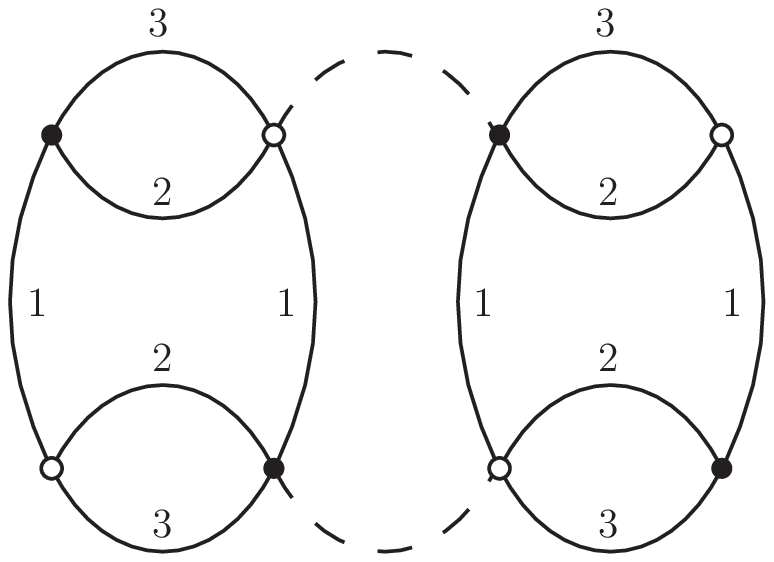}}     
  \subfloat[$\cG_3$]{\label{symg3}\includegraphics[scale=0.5]{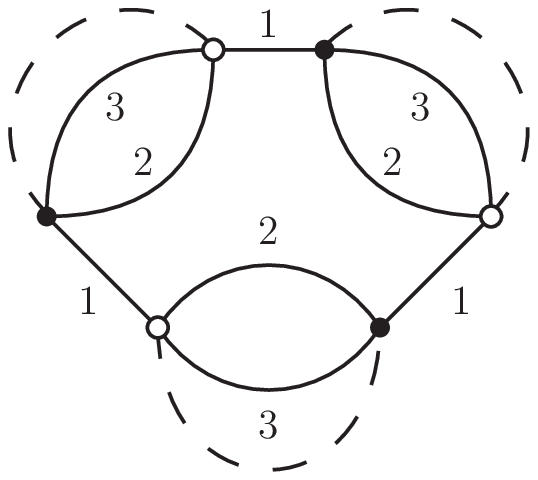}}
  \caption{Three graphs with different symmetry factors: $k(\cG_1) = 1$, $k(\cG_2) = 2$ and $k(\cG_3) = 3$.}\label{ex_sym}
\end{figure}

In order to write the explicit expression of $\cA_\cG$, we need to introduce further graph-theoretic definitions and notations. We note $L(\cG)$ the set of color-$0$ \textit{internal} edges in $\cG$, and $N(\cG)$ the set of external legs, which we will sometimes simply call \textit{lines} and \textit{legs}. We can furthermore partition $N(\cG) = N_{\circ} (\cG) \cup N_{\bullet} (\cG)$, where $N_{\circ} (\cG)$ (resp. $N_{\bullet} (\cG)$) is the set of legs hooked to white (resp. black) nodes. A \textit{face of color} $\ell$ is a non-empty subset $f \subset L(\cG)$ which, upon addition of color-$\ell$ edges only, can be completed into a maximally connected subset of color-$\ell$ edges and (internal, color-$0$) lines. Such a maximally connected subset may form a \textit{closed} loop, we say that $f$ is \textit{internal} (or closed) in this case, and $f$ is \textit{external} (or open) otherwise. The set of internal faces and external faces of $\cG$ are respectively noted $F(\cG)$, and $F_{ext}(\cG)$. We will also need to keep track of direct identifications of boundary variables through single colored lines. We call these \emph{empty external faces} because they play a similar role as external faces; their set is denoted $F_{ext}^{\emptyset}(\cG)$. When no confusion 
arises, we will use the same symbol to denote the cardinality of one of the sets defined so far and the set itself. We can fix an arbitrary orientation of the lines $l \in L(\cG)$ and faces $f \in F(\cG) \cup F_{ext} (\cG)$, and encode their adjacency relations into a matrix $\epsilon_{lf}= \pm 1$ or 0. However, in order to make the expressions more explicit, it is convenient to fix the orientations so that: a) $l \in L(\cG)$ is positively oriented from the white to the black end nodes it connects; b) $\epsilon_{lf} = 1$ if $l \in f$, and $0$ otherwise\footnote{The fact that the faces can always be oriented in such a way is a particularity of tensorial GFTs.}. This canonical orientation allows to define the \textit{source} and \textit{target} of an external face. The function $s: F_{ext} (\cG) \cup F_{ext}^{\emptyset} (\cG) \rightarrow N_{\bullet}(\cG) \times \{ 1, 2, 3\}$ maps $f$ to $(e,\ell)$, where $\ell$ is the color of $f$ and $e$ is the external leg connected to its source. We define $t: F_{ext} (\cG) \cup F_{ext}^{\emptyset} (\cG) \rightarrow N_{\circ}(\cG) \times 
\{ 1, 2, 3\}$ in a similar way.   
The bare amplitude of a graph $\cG$ is a function of $3 N(\cG)$ external group elements $\{ g^{ext}_{(e,\ell)}\}$, which formally writes: 
\begin{eqnarray}\label{ampl_ab}
\cA_\cG (g^{ext}_{(e,\ell)}) &=& \left[ \prod_{l \in L(\cG)} \int \extd \alpha_{l} \, e^{- m^2 \alpha_l} \int \extd h_l \right] 
\left( \prod_{f \in F (\cG)} K_{\alpha(f)}\left( \overrightarrow{\prod_{l \in f}} h_l \right) \right) \nn\\
&& \int [\extd g_{(e,l)}] \, \prod_{e \in N_{\bullet} (\cG)} C(g^{ext}_{(e,\ell)}; g_{(e,\ell)}) \, \prod_{e \in N_{\circ} (\cG)} C(g_{(e,\ell)}; g^{ext}_{(e,\ell)}) \nn\\
&& \left( \prod_{f \in F_{ext}(\cG)} K_{\alpha(f)} \left( g_{s(f)}
\left[\overrightarrow{\prod_{e \in f}} {h_e} \right] g_{t(f)}^{\inv} \right) \right) \nn\\
&& \left( \prod_{f \in F_{ext}^{\emptyset} (\cG)} \delta\left( g_{s(f)} g_{t(f)}^{\inv} \right) \right) \,.
\end{eqnarray}
We clearly separated the contributions of the internal faces (first line), from the external propagators (second line\footnote{We could simplify this expression by means of the symmetry $C(g_i ; g_i') = C(g_i' ; g_i)$, but we think that the present expression is better suited for what will come next.}), and how these are connected to the holonomy variables in the bulk through the external faces (third and fourth line).

\subsection{Divergences and power-counting}

In \cite{cor_su2}, the scale ladder provided by the parameter $\alpha$ was used as a basis for renormalization theory. Divergences from high scales (i.e. from the region $\alpha \approx 0$) generate counter-terms at lower scales (i.e. bigger $\alpha$ parameters), and the theory is renormalizable if those can be reabsorbed into a finite number of coupling constants. The divergences can be most easily understood in the multiscale representation of the Feynman amplitudes. This consists in slicing the range of the Schwinger parameter $\alpha$, according to a geometric progression. To this effect, one fixes an arbitrary constant $M > 1$ and define the propagator $C_i$ at scale $i$ by:
\beq 
C_i ( g_1 , g_2 , g_3 ; g_1' , g_2' , g_3' ) = \int_{M^{-2i}}^{{M^{-2(i-1)}}} \extd \alpha \, \e^{- m^2 \alpha} \int \extd h \, \prod_{\ell=1}^{3} K_{\alpha} ( g_\ell h g_\ell'^{\inv})\,,
\eeq
if $i \neq 0$ and
\beq
C_0 ( g_1 , g_2 , g_3 ; g_1' , g_2' , g_3' ) = \int_{1}^{+ \infty} \extd \alpha \, \e^{- m^2 \alpha} \int \extd h \, \prod_{\ell=1}^{3} K_{\alpha} ( g_\ell h g_\ell'^{\inv})\,.
\eeq
This induces a decomposition of the amplitude of a graph $\cG$ as a sum indexed by internal \textit{scale attributions} $\mu \equiv \{ i_l \vert \; l \in L(\cG) \}$:
\beq\label{ampl_multi}
\cA_\cG = \sum_{\mu} \cA_{\cG , \mu}\,.
\eeq
The amplitude at scale $\mu$ $\cA_{\cG,\mu}$ is simply obtained from (\ref{ampl_ab}) by restricting the $\alpha_l$ integrals to the slices $i_l$. Note that the external legs are left untouched, and by convention they are attributed the scale label $i = - 1$. The sum over $\mu$ is regulated thanks to the introduction of a \textit{cut-off} $\rho$ on the scale labels $i$, which can be removed only after renormalization. The main interest of the decomposition (\ref{ampl_multi}) is that it allows to compute rigorous bounds on $\vert \cA_\cG \vert$ by a systematic optimization procedure, which yields simple bounds on $\vert \cA_{\cG , \mu} \vert$ as functions of $\mu$.  

Given a couple $(\cG , \mu)$, one can construct a set of \textit{high subgraphs}, defined as the maximally connected subgraphs with internal scales strictly smaller than the external scales. To this effect, let us define $\cG_i$ as the subgraph made of the lines of $\cG$ with scales higher or equal to $i$. Its connected components\footnote{As already mentioned before, in this paper we use the ordinary graph-theoretic notion of connectedness, which is referred to as vertex-connectedness in \cite{cor_su2}.} can be labeled $\cG_i^{(1)}, \ldots , \cG_i^{(k(i))}$, where $k(i)$ is the number of connected components. The $\cG_i^{(k)}$'s are exactly the high subgraphs at scale $i$: they are connected; their internal lines have scales higher or equal to $i$; and their external legs have scales strictly smaller than $i$. An important property of the high subgraphs is that they form an \textit{inclusion tree}. That is to say that two high subgraphs $\cH_1 \subset \cG$ and $\cH_2 \subset \cG$ are either line-disjoint, or one 
is included into the other; and furthermore all the high subgraphs are by definition included in $\cG$, which is itself a high subgraph (at scale $i = 0$). These high subgraphs are ultimately responsible for the nested structure of divergences, and when successively integrated out, make the coupling constants run with respect to $i$. More precisely, only the \textit{divergent} high subgraphs have to be taken into account in the renormalization group equations. They are determined by a precise power-counting theorem, and we refer the reader to \cite{cor_u1, cor_su2} for details. For the purpose of this paper, we only need to know that the \textit{divergent} high subgraphs $\cH$ are characterized by the inequality $\omega(\cH) \geq 0$, where $\omega$ is the \textit{superficial degree of divergence}, defined as \cite{lin_GFT, cor_u1, cor_su2}:
\beq
\omega(\cH) = - 2 L(\cH) + 3 (F(\cH) - R(\cH))\,, 
\eeq  
and $R(\cH)$ is the rank of the $L(\cH) \times F(\cH)$ incidence matrix $\epsilon_{lf}$ of $\cH$. Alternatively, as proven in \cite{cor_su2}, the divergence degree can be conveniently re-expressed in terms of the numbers of $2k$-valent vertices $n_{2k}$, the number of external legs $N$ and an integer $\rho \leq 0$:
\beq
\omega = 3 - \frac{N}{2}  + 3 \rho + \sum_{k \in \mathbb{N}\setminus\{0\}} (k - 3) n_{2k}\,.  
\eeq 
Note that in \cite{cor_su2} the bare action was assumed to stop at $6$-valent bubble interactions, in which case the last sum reduces to $- 2 n_2 - n_4$. The condition $\rho = 0$ characterizes the class of \textit{melonic graphs}\footnote{In this paper, we always assume $N \neq 0$. When $N = 0$, $\rho$ may also be $1$ and $\rho(\cH) = 0$ does not imply that $\cH$ is melonic.}, which contains the divergent graphs as a subset when $n_{2k} = 0$ for $k \geq 4$ (see Table \ref{div}). 
\begin{table}[h]
\centering
\begin{tabular}{| l | c | c | c || r |}
    \hline
    $N$ & $n_2$ & $n_4$ & $\rho$ & $\omega$  \\ \hline\hline
 6 & 0 & 0 & 0 & 0 \\ \hline
 4 & 0 & 0 & 0 & 1 \\ 
 4 & 0 & 1 & 0 & 0 \\ \hline
 2 & 0 & 0 & 0 & 2 \\ 
 2 & 0 & 1 & 0 & 1 \\
 2 & 0 & 2 & 0 & 0 \\
 2 & 1 & 0 & 0 & 0 \\ 
    \hline
  \end{tabular}
\caption{Non-vacuum divergent graphs when $n_{2k} = 0$ for any $k \geq 4$.}
\label{div}
\end{table} 

\

Let us conclude this section by recalling how melonic graphs are defined.
\begin{definition}
 Let $\cG$ be a graph. For any integer $k$ such that $1 \leq k \leq 4$, a \emph{$k$-dipole} is a line of $\cG$ linking two nodes $n$ and
$\overline{n}$ which are connected by exactly $k - 1$ additional colored lines.
\end{definition}

\begin{definition}
Let $\cG$ be a graph. The \emph{contraction of a $k$-dipole} $d_k$ is an operation consisting in:
\begin{enumerate}[(i)]
 \item deleting the two nodes $n$ and $\overline{n}$ linked by $d_k$, together with the $k$ lines that connect them;
 \item reconnecting the resulting $d - k + 1$ pairs of open legs according to their colors.
\end{enumerate}
We call $\cG / d_k$ the resulting graph. 
\end{definition}

See Figure \ref{dipoles} for examples of $k$-dipoles and their contractions.

\begin{figure}[ht]
\begin{center}
\includegraphics[scale=0.5]{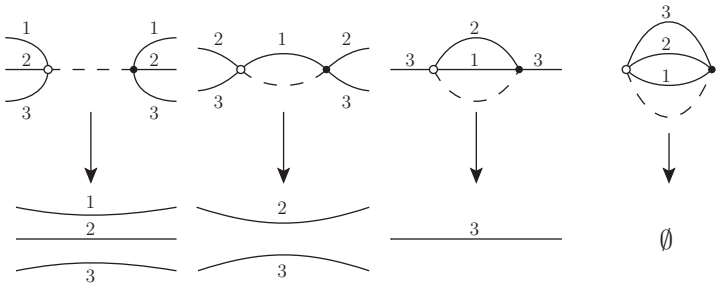}
\caption{From left to right: $1$-, $2$-, $3$- and $4$-dipoles, together with their contractions.}
\label{dipoles}
\end{center}
\end{figure}

\begin{definition}
We call \emph{contraction of a subgraph} $\cH \subset \cG$ the successive contractions of all the lines of $\cH$. The resulting graph is independent of the order in which the lines of $\cH$ are contracted,
and is noted $\cG / \cH$.
\end{definition}

Given a graph $\cG$, and some lines $l_1 , \ldots , l_k \in L(\cG)$, we will denote by $\{ l_1, \ldots , l_k\}$ the minimal subgraph of $\cG$ containing the lines $l_1, \ldots, l_k$. 

\begin{definition}
A \emph{melopole} is a single-vertex graph $\cG$ such that there is at least one ordering of its $k$ lines as $l_1, \cdots , l_k$ such that $ \{l_1 ,  \dots , l_{i} \} / \{l_1 , \dots , l_{i-1} \} $ is a $3$-dipole for $1 \le i \le k$. See Figure \ref{ex_melop}.
\end{definition}

\begin{figure}[ht]
\begin{center}
\includegraphics[scale=0.5]{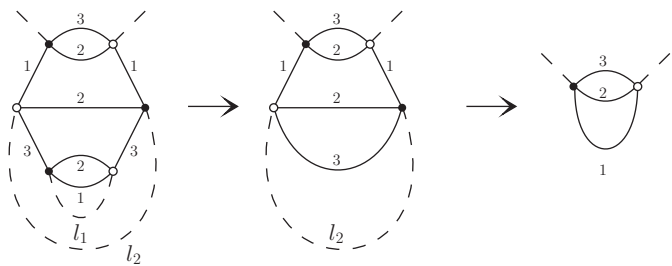}
\caption{A melopole with two lines. $\{ l_1 \}$ and $\{l_1 , l_2 \} / \{l_1 \}$ are $3$-dipoles, as illustrated by the successive contractions of $l_1$ and $l_2$.}
\label{ex_melop}
\end{center}
\end{figure}

\begin{definition}
A \emph{melonic graph} 
is a connected\footnote{This definition slightly differs from that of \cite{cor_su2}, in the sense that melonic graphs were defined to be face-connected there, which is a stronger condition.} graph $\cG$ containing at least one maximal tree $\cT$ such that $\cH / \cT$ is a melopole.
\end{definition}

A simple example of melonic graph is provided in Figure \ref{ex_melonic}.

\begin{figure}[ht]
\begin{center}
\includegraphics[scale=0.4]{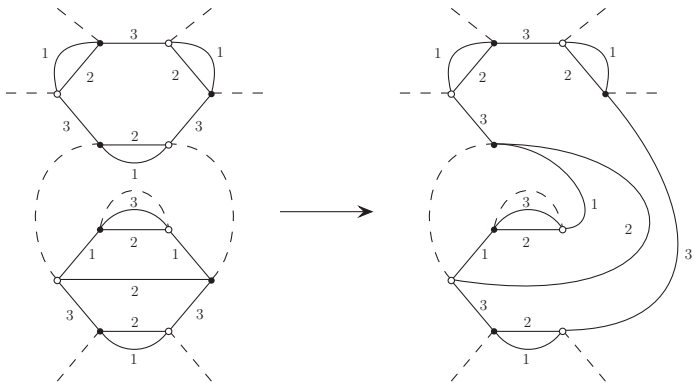}
\caption{A melonic graph which reduces to a melopole after contraction of a single tree line.}
\label{ex_melonic}
\end{center}
\end{figure}

\subsection{Contraction operators}

In the multiscale representation of the amplitudes, the divergences can be extracted by the action of so-called \emph{contraction operators}. In order to define them, let us introduce generalized amplitudes
\begin{eqnarray}\label{ampl_t}
\cA_\cG (g^{ext}_{(e,\ell)}; t) &\equiv& \left[ \prod_{l \in L(\cG)} \int \extd \alpha_{l} \, e^{- m^2 \alpha_l} \int \extd h_l \right] 
\left( \prod_{f \in F (\cG)} K_{\alpha(f)}\left( \overrightarrow{\prod_{l \in f}} h_l \right) \right) \nn\\
&& \int [\extd g_{(e,l)}] \, \prod_{e \in N_{\bullet} (\cG)} C(g^{ext}_{(e,\ell)}; g_{(e,\ell)}(t)) \, \prod_{e \in N_{\circ} (\cG)} C(g_{(e,\ell)}; g^{ext}_{(e,\ell)}) \nn\\
&& \left( \prod_{f \in F_{ext}(\cG)} K_{\alpha(f)} \left( g_{s(f)}
\left[\overrightarrow{\prod_{e \in f}} {h_e} \right] g_{t(f)}^{\inv} \right) \right) \nn \\
&& \left( \prod_{f \in F_{ext}^{\emptyset} (\cG)} \delta\left( g_{s(f)} g_{t(f)}^{\inv} \right) \right) \,.
\end{eqnarray}
where $t \in [0,1]$ and\footnote{For any $g \in \SU(2)$, $X_g$ denotes the Lie algebra element with the smallest norm such that $\exp(X_g) = g$.}
\bes
\forall f \in F_{ext}(\cG) \,, \qquad g_{s(f)}(t) &\equiv& g_{t(f)} \exp\left(t X_{g_{t(f)}^{\inv} g_{s(f)}}\right) \,,\\
\forall f \in F_{ext}^{\emptyset}(\cG) \,, \qquad g_{s(f)}(t) &\equiv& g_{s(f)} \,.
\ees
We can then express any amplitude $\cA_{\cG, \mu}$ as a Taylor expansion with respect to the parameter $t$:
\bes\label{taylor}
\cA_{\cG , \mu} = \cA_{\cG , \mu} ( \cdot ;  1) &=& \cA_{\cG , \mu} (\cdot ; 0) + \sum_{k = 1}^{\omega(\cG)} \frac{1}{k !} \cA_{\cG , \mu}^{(k)} ( \cdot ; 0) \nn \\
&&+ \int_{0}^{1} \extd t \, \frac{(1 - t)^{\omega(\cG)}}{\omega(\cG ) !} \cA_{\cG , \mu}^{(\omega(\cG ) + 1)} (\cdot ; t)\,.
\ees 
The interest of such an expression is that the remainder can be shown to be power-counting convergent \cite{cor_su2}, and can therefore be dispensed with. Moreover, one shows that
\beq
\cA_{\cG , \mu} (g^{ext}_{(e,\ell)}; 0)  \propto \cA_{\cG / \{ l_1 , \ldots , l_k \}, \mu} (g^{ext}_{(e,\ell)}; 0)\,, 
\eeq
where $l_1, \ldots , l_k$ are the lines of $\cG$. We can therefore implicitly define the contraction operator $\tau$ by the equation:
\beq
\cA_{\cG , \mu} (g^{ext}_{(e,\ell)}; 0)  = [\tau \cA_{\cG , \mu} ] \cA_{\cG / \{ l_1 , \ldots , l_k \}, \mu} (g^{ext}_{(e,\ell)}; 0)\,.
\eeq
More explicitly, the heat-kernels associated to the external faces can be integrated out, yielding:
\beq
[\tau \cA_{\cG , \mu} ] = \left[ \prod_{l \in L(\cG)} \int \extd \alpha_{l} \, e^{- m^2 \alpha_l} \int \extd h_l \right] 
\left( \prod_{f \in F (\cG)} K_{\alpha(f)}\left( \overrightarrow{\prod_{l \in f}} h_l \right) \right)\,.
\eeq
Finally, if $\overline{\cG}$ is obtain from $\cG$ by amputation of its external legs, we define:
\beq
[\tau \cA_{\overline{\cG} , \mu} ] \equiv [\tau \cA_{\cG , \mu} ]\,.
\eeq

Let us now consider higher order terms. To begin with, the first order always identically vanishes:
\beq
\cA_{\cG , \mu}^{(1)} (\cdot ; 0) = 0\,.
\eeq
The model studied in the present paper generates quadratically divergent graphs at most, which means we will not need to go beyond second order. Moreover, all the quadratically divergent graphs have $N=2$ external legs and are melonic. The amplitudes we need to expand to second order have therefore the following structure:
\bes\label{factorization_w}
\cA_{\cG , \mu} (g_\ell^{ext}, \overline{g}_\ell^{ext} ; t ) &=& \left[ \prod_{l \in L(\cG)} \int \extd \alpha_{l} \, e^{- m^2 \alpha_l} \int \extd h_l \right] \nn \\ 
&& \left( \prod_{f \in F (\cG)} K_{\alpha(f)}\left( \overrightarrow{\prod_{l \in f}} h_l \right) \right) \nn\\
&& \int [\extd g_{\ell}] [\extd \overline{g}_{\ell}] \, C(g^{ext}_{\ell}; g_{\ell}(t)) \,  C(\overline{g}_{\ell}; \overline{g}^{ext}_{\ell}) \nn\\
&& \left( \prod_{f \in F_{ext}(\cG)} K_{\alpha(f)} \left( g_{s(f)}
\left[\overrightarrow{\prod_{e \in f}} {h_e} \right] g_{t(f)}^{\inv} \right) \right) \nn \\
&& \left( \prod_{f \in F_{ext}^{\emptyset} (\cG)} \delta\left( g_{s(f)} g_{t(f)}^{\inv} \right) \right)
\ees
where $F_{ext}(\cG) \cup F_{ext}^{\emptyset}(\cG)$ has three elements. Suppose for instance that \linebreak $F_{ext} (\cG) = \{ f_1 \}$, with $f_1$ of color $1$. Then on can prove that \cite{cor_su2}:
\beq
\frac{1}{2} \cA_{\cG , \mu}^{(2)} (g_\ell^{ext}, \overline{g}_\ell^{ext} ; 0 ) = \tau^{(2)} \cA_{\cG , \mu}  \times
\int [\extd g_\ell^{ext}]^3  C(g_\ell^{ext};g_\ell)
\Delta_{g_1}
C(g_\ell;\overline{g}_\ell^{ext})\,,
\eeq
where we have defined
\bes
\tau^{(2)} \cA_{\cG , \mu} &\equiv& \frac{1}{6} \left[ \prod_{l \in L(\cG)} \int \extd \alpha_{l} \, e^{- m^2 \alpha_l} \int \extd h_l \right] 
\left( \prod_{f \in F (\cG)} K_{\alpha(f)}\left( \overrightarrow{\prod_{l \in f}} h_l \right) \right) \nn \\
&& \int \extd g \vert X_g \vert^2\, K_{\alpha(f_1)} \left( g
\overrightarrow{\prod_{e \in f_1}} {h_e} \right) \,.
\ees
By use of the Leibniz rule, together with the fact that terms containing first derivatives vanish, the last two equations can be generalized to
\beq
\frac{1}{2} \cA_{\cG , \mu}^{(2)} (g_\ell^{ext}, \overline{g}_\ell^{ext} ; 0 ) = 
\int [\extd g_\ell^{ext}]^3  C(g_\ell^{ext};g_\ell)
 \left( \sum_\ell [\tau^{(2)}_\ell \cA_{\cG , \mu}] \Delta_{g_\ell} \right)
C(g_\ell;\overline{g}_\ell^{ext})\,,
\eeq
where
\bes
\tau^{(2)}_\ell \cA_{\cG , \mu} &=& \frac{1}{6} \left[ \prod_{l \in L(\cG)} \int \extd \alpha_{l} \, e^{- m^2 \alpha_l} \int \extd h_l \right] 
\left( \prod_{f \in F (\cG)} K_{\alpha(f)}\left( \overrightarrow{\prod_{l \in f}} h_l \right) \right) \nn \\
&&\int \extd g \vert X_g \vert^2\, K_{\alpha(f_\ell)} \left( g
\overrightarrow{\prod_{e \in f_\ell }} {h_e} \right)
\ees
if there exists a $f_\ell \in F_{ext} (\cG)$ of color $\ell$, and
\beq
\tau^{(2)}_\ell \cA_{\cG , \mu} = 0
\eeq
otherwise.
Again, this definition makes no reference to external legs, and can therefore be generalized to amputated graphs: if $\overline{\cG}$ is obtained by amputation of the external legs of $\cG$, then
\beq
\tau^{(2)}_\ell \cA_{\overline{\cG} , \mu} \equiv \tau^{(2)}_\ell \cA_{\cG , \mu}\,.
\eeq

%%%%%%%%%%%%%%%%%%% Renormalization group flow %%%%%%%%%%%%%%%%%%%%%%%%%%%%%%%%%%%%
\section{Renormalization group flow}\label{sec:rgflow}

In this section, we define general flow equations for Wilson's effective action. They involve: a) a step by step integration of UV slices which retains only the tensor invariant contributions of high divergent graphs, as already outlined above; and b) a large scale approximation ($i \to + \infty$) of the coefficients entering the equations. Because we are working with a field theory defined on a compact group rather than Euclidean space, finite size effects make the renormalization group non-autonomous\footnote{I thank Daniele Oriti for discussions on this point.} i.e. the flow equations have an explicit dependence in the cut-off. However, such effects can be neglected in the deep UV regime, where only the local structure of $\SU(2)$ is probed.

\subsection{Wilson's effective action and dimensionless coupling constants}

In ordinary quantum field theories, Wilson's effective action is best described in terms of \textit{dimensionless} coupling constants. In GFT, and more generally in quantum gravity, such a notion is \textit{a priori} empty since all the fields and coupling constants are already strictly speaking dimensionless. Dimensionful quantities should only appear in the spectra of quantum observables such as the length, area or volume operators. For example, in canonical loop quantum gravity, the area of a surface punctured by a single spin-network link with spin $j$ is $8 \pi \gamma \sqrt{j (j+1)} \ell_P^2$ where $\gamma$ is the (dimensionless) Immirzi parameter and $\ell_P$ is the Planck length. In this respect, the spins themselves can be attributed a canonical dimension. In our context, let us attribute a unit canonical dimension 
\beq
[j] \equiv 1
\eeq
to spin variables. 
Since the propagator decays quadratically in the spins, one immediately infers the dimension of the mass: $[m] = 1$. As for the canonical dimensions of all the other coupling constants, they can be deduced from the power-counting. For example, Table \ref{div} shows that four-valent coupling constants will receive linearly divergent contributions of order $M^i$ at scale $i$. Hence they are to be thought of as coupling constants with unit canonical dimension. More generally, we are lead in this way to define
\beq\boxed{
[ t_b ] = 3 - \frac{N_b}{2} \equiv d_b \,, }
\eeq  
for an arbitrary coupling constant $t_b$, where $N_b$ is the valency of the bubble $b$. Note that such a definition of canonical dimension has also recently been introduced in matrix and tensor models, for similar purposes \cite{astrid_tim}. We define the effective action at scale $M^{i}$ as a sum over all possible bubbles
\bes
S_i (\vphi , \vphib) &=& \sum_b t_{b,i} \frac{ I_b (\vphi , \vphib) }{k(b)}\\
&=& \sum_b  u_{b,i} M^{d_b i} \frac{I_b (\vphi , \vphib)}{k(b)} \,, \label{S_dl}
\ees
where $t_{b,i}$ (resp. $u_{b,i}$) are the dimensionful (resp. dimensionless) coupling constants. We furthermore again assume color permutation invariance, and set up $t_{2, i} = 0$. The latter is consistent provided that we allow the mass parameter of the covariance to vary with the scale. We will denote by $C_{i, m} = P \overline{C}_{i,m}$ the covariance in the slice $i$ with mass $m$, and by $C^{i}_{m} = P \overline{C}^{i}_{m} = \underset{k \leq i}{\sum} C_{k, m}$ the full covariance with cut-off $i$. The dimensionless mass coupling at scale $i$ is 
\beq
u_{2,i} \equiv \frac{m_i^2}{M^{2i}}\,.
\eeq 
In the following, we will use perturbative expansions with respect to the dimensionless coupling constants. The degree of divergence
\beq
\omega = 3 - \frac{N}{2} + \sum_{k \in \mathbb{N}} (3 - k) n_{2k}  + 3 \rho\,,
\eeq
tells us how Feynman amplitudes diverge in an expansion with respect to the dimensionful parameters $t_{b,i}$. Taking into account the additional scalings appearing in equation (\ref{S_dl}), we immediately infer a modified degree of divergence
\beq\label{degree2}
\overline{\omega} = 3 - \frac{N}{2} + 3 \rho
\eeq 
for the new expansion in the $u_{b,i}$'s. From now on, (a non-vacuum) graph $\cG$ will be said to be divergent whenever $\overline{\omega}(\cG) \geq 0$, that is whenever $N \leq 3$ and $\rho = 0$. The new classification of divergent graphs is summarized in Table \ref{div2}. 

\begin{table}[h]
\centering
\begin{tabular}{| l | c || r |}
    \hline
    $N$ & $\rho$ & $\omegab$  \\ \hline\hline
 6 & 0 & 0 \\ \hline 
 4 & 0 & 1 \\ \hline
 2 & 0 & 2 \\ 
    \hline
  \end{tabular}
\caption{Non-vacuum divergent graphs in the expansion with respect to dimensionless coupling constants.}
\label{div2}
\end{table}

\subsection{Discrete renormalization group flow}

In order to determine the effective action $S_{i-1}$ and the mass coupling $u_{2,i-1}$ from $S_i$ and $u_{2,i}$, we proceed in two steps. The first step consists in integrating out fluctuations at scale $i$ to deduce the effective action $\tilde{S}_{i-1}$ before wave-function renormalization:
\beq\label{effective_a}
  K_{i-1} \exp\left(-\tS_{i-1}(\Phi, \Phib) + R_{i-1}(\Phi, \Phib)\right) \equiv \int \extd \mu_{C_{i, m_i}} (\vphi , \vphib) \, \exp\left(- S_i (\Phi + \vphi , \Phib + \vphib)\right)\,,
\eeq
where the rest term $R_{i-1} (\Phi, \Phib) = \cO ( M^{-i} )$ is a sum of contributions which are suppressed at large $i$, and $K_{i-1}$ is a possibly large constant due to vacuum divergences. $R_{i-1}$ contains in particular the Feynman graphs with $\omegab \leq - 1$, and the convergent Taylor remainders associated to the non-tensorial parts of the non-vacuum divergent graphs. $\tS_{i-1}$ may be written as 
\beq
\tS_{i-1} = CT_{\vphi, i-1} S_\vphi + CT_{m, i-1} S_2 + \sum_{b\vert N_b \neq 2} \tu_{b,i-1} M^{ d_b (i-1)} \frac{I_b}{k(b)}  \,,
%+ \sum_{b \in \cB} \tu_{b,i-1} M^{d_b (i-1)} \frac{I_b}{k(b)}
\eeq
in terms of intermediate dimensionless coupling constants $\tu_{b,i}$. The additional wave-function and mass terms
\bes
S_{\vphi} (\vphi , \vphib) &=& \int [\extd g ]^3 \, \vphi(g_1 , g_2 , g_3 ) \left( - \sum_{l = 1}^{3} \Delta_\ell \right) \vphib(g_1 , g_2 , g_3 )\,,\\
S_{2} (\vphi , \vphib) &=& \int [\extd g ]^3 \, \vphi(g_1 , g_2 , g_3 ) \vphib(g_1 , g_2 , g_3 )\,,
\ees
respectively parameterized by a dimensionless constant $CT_{\vphi, i-1}$ and a dimension $2$ constant $CT_{m, i-1} = \tu_{2, i - 1} M^{2(i -1)}$, are generated by the $2$-valent divergent graphs. They need to be reabsorbed into the covariance $C^{i - 1}_{m_{i-1}}$ via a field renormalization, which is the second step of the procedure. To this effect, let us define the operator 
\beq
M_{i-1} = - CT_{m , i -1} + CT_{ \vphi , i - 1} \sum_\ell \Delta_\ell \,,
\eeq
and call $\tC$ the covariance of the measure:
\beq
\extd \mu_{C^{i-1}_{m_{i}}} (\Phi , \overline{\Phi} ) \exp\left( \int [\extd g_\ell] [\extd g_\ell'] \Phi(g_1 , g_2 , g_3) \, M_{i-1} (g_\ell ; g_\ell' ) \, \overline{\Phi}(g_1' , g_2' , g_3') \right) \,.
\eeq
It can be computed by summing over connected $2$-point functions, in the following way:
\bes
\tC &=& C^{i-1}_{m_{i}} + C^{i-1}_{m_{i}} M_{i-1} C^{i-1}_{m_{i}} + C^{i-1}_{m_{i}} M_{i-1} C^{i-1}_{m_{i}} M_{i-1} C^{i-1}_{m_{i}} + \ldots \nn \\
&=& P \left( \overline{C}^{i-1}_{m_i} + \overline{C}^{i-1}_{m_i} M_{i-1} \overline{C}^{i-1}_{m_i} + \overline{C}^{i-1}_{m_i} M_{i-1} \overline{C}^{i-1}_{m_i} M_{i-1} \overline{C}^{i-1}_{m_i} + \ldots\right) \nn \\
&=& P \frac{\overline{C}^{i-1}_{m_i}}{1 - \overline{C}^{i-1}_{m_i} M_{i-1}} \,.
\ees
The expression of $\overline{C}^{i-1}_{m_i}$ is:
\bes
\overline{C}^{i-1}_{m_i} &=& \int_{M^{- 2 (i-1)}}^{+\infty} \extd \alpha \, \exp\left(- \alpha ( m_{i}^2 - \sum_\ell \Delta_\ell )\right)\nn \\
&=& \frac{ \exp\left( - M^{- 2 (i-1)} ( m_{i}^2 - \sum_\ell \Delta_\ell )\right)}{m_{i}^2 - \sum_\ell \Delta_\ell} \,, 
%&\approx& \frac{\exp\left(- M^{-2 (i-1)} m_{i}^2\right)}{m_{i}^2 - \sum_\ell \Delta_\ell}\,,
\ees
the second line making the smooth cut-off on the spins explicit. Thanks to this decay we can discard the higher powers of the mass and Laplace operators generated by the exponential terms in the denominator of $\tC$. Using moreover the approximation $m_i^2 \sim m_{i-1}^2$ in the numerator, we are lead to:
\bes
\tC &\approx& \frac{P}{Z_{i-1}} \frac{\exp\left( - M^{- 2 (i-1)} ( m_{i-1}^2 - \sum_\ell \Delta_\ell )\right)}{m^2_{i-1} - \sum_\ell \Delta_\ell} = \frac{1}{Z_{i - 1}} C^{i-1}_{m_{i-1}} \,,\\
Z_{i-1} &\equiv& 1 + CT_{\vphi , i-1} \,, \label{zi} \\
m^2_{i-1} &\equiv& \frac{m_{i}^2 +  CT_{m , i-1}}{1 +  CT_{\vphi , i-1}} \,. \label{mass_phys}
\ees
We have thus determined the mass at scale $i-1$. The other coupling constants are obtained after the field redefinition
\beq
\Phi \to \frac{\Phi}{\sqrt{Z_{i-1}}} \,.
\eeq
The powers of $Z_{i-1}$ subsequently appearing in the interaction part of the action must be reabsorbed into the coupling constants:
\beq\label{rescale_couplings}
u_{b , i - 1} \equiv \frac{\tu_{b , i}}{{Z_{i-1}}^{N_b / 2}} 
\eeq
entering the effective action $S_{i-1}$. Together with the covariance $C^{i-1}_{m_{i-1}}$, they parameterize the theory once the cut-off has been lowered to $(i-1)$.

\subsection{Reducible graphs}

In space-time based quantum field theories, momentum conservation allows to discard $1$-particle irreducible graphs, because their contributions are strongly suppressed\footnote{They are even identically zero if one works with a sharp slicing of the momenta.} when the difference between the scale of the probes $i_0$ and that of the internal lines $i$ grows very large. Due to the combinatorial non-locality introduced by the interactions, momentum conservation has slightly stronger implications in TGFT, which we wish to elaborate upon here.  

\

The harmonic decomposition of the field $\vphi$ is
\beq
\vphi(g_1 , g_2 , g_3 ) = \sum_{\{ j_\ell , a_\ell , b_\ell \}} \vphi^{\{ j_\ell\}}_{\{ a_\ell , b_\ell \}} \prod_\ell \sqrt{2 j_\ell + 1} D^{j_\ell}_{a_\ell , b_\ell} (g_\ell)\,,
\eeq 
where $D^j$ is the Wigner matrix associated to the irreducible representation $j$. The full propagator $\ref{def_propa}$ is diagonal with respect to the modes $\vphi^{\{ j_\ell\}}_{\{ a_\ell , b_\ell \}}$:
\beq
C^{\{ j_\ell\}; \{ j_\ell'\}}_{\{ a_\ell , b_\ell \};\{ a_\ell' , b_\ell' \}} = \frac{P^{\{ j_\ell \}}_{\{a_\ell , b_\ell\}} }{m^2 + \sum_\ell j_\ell ( j_\ell + 1)} \prod_\ell \delta_{j_\ell j_\ell'} \delta_{a_\ell a_\ell'} \delta_{b_\ell b_\ell'}\,,
\eeq
where
\beq
P^{\{ j_\ell \}}_{\{a_\ell , b_\ell \}} = \int \extd h \prod_\ell D^{j_\ell}_{a_\ell b_\ell}(h)\,.
\eeq
The kernel of a tensor invariant $I_b$ in this spin representation is simply a product of delta functions, identifying the indices $\{ j_\ell , a_\ell , b_\ell \}$ pairwise along colored edges. As a result, these indices are conserved along the faces of the Feynman graphs. This is the counterpart of momentum conservation in local quantum field theories. Note that the key difference between these two conservation rules is of combinatorial nature, and therefore it is natural to expect that the class of graphs which are suppressed due to combinatorial obstructions in TGFTs is not exactly the same as in local quantum field theories.

\

Let us introduce the following notion of reducibility in TGFT.
\begin{definition}
Let $\cG$ be a connected graph. We say that $\cG$ is \textit{reducible} if it possesses a line which does not appear in any of the internal faces:
\beq\label{irred_def}
\exists l \in L(\cG)\,, \qquad \forall f \in F(\cG)\,,\; l \notin f\,. 
\eeq
If not, $\cG$ is said to be \textit{irreducible}.
\end{definition}

\noindent{\bf Examples.} A $1$-particle reducible graph is always reducible, but the converse is not true (see Figure \ref{ex_reduce}).

\begin{figure}[h]
  \centering
  \subfloat[$\cH_1$]{\label{1pr_melonic}\includegraphics[scale=0.5]{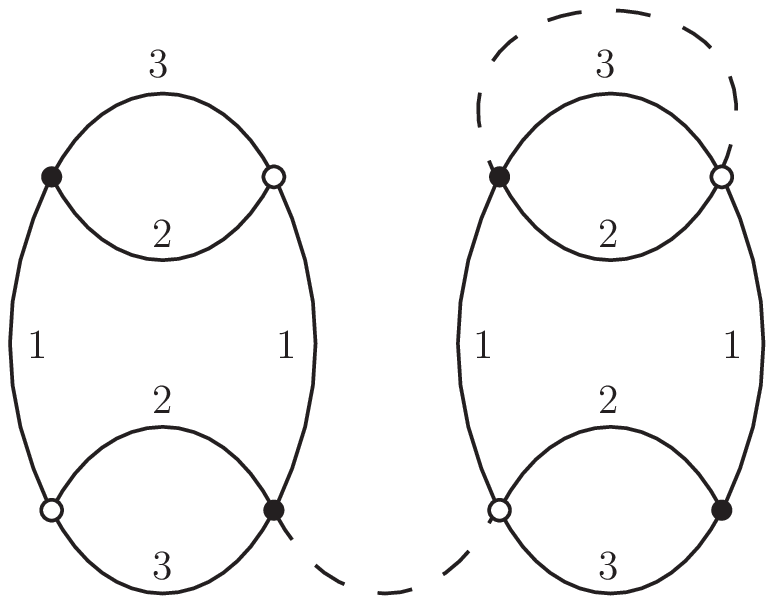}}     
  \subfloat[$\cH_2$]{\label{reduce_1pi}\includegraphics[scale=0.5]{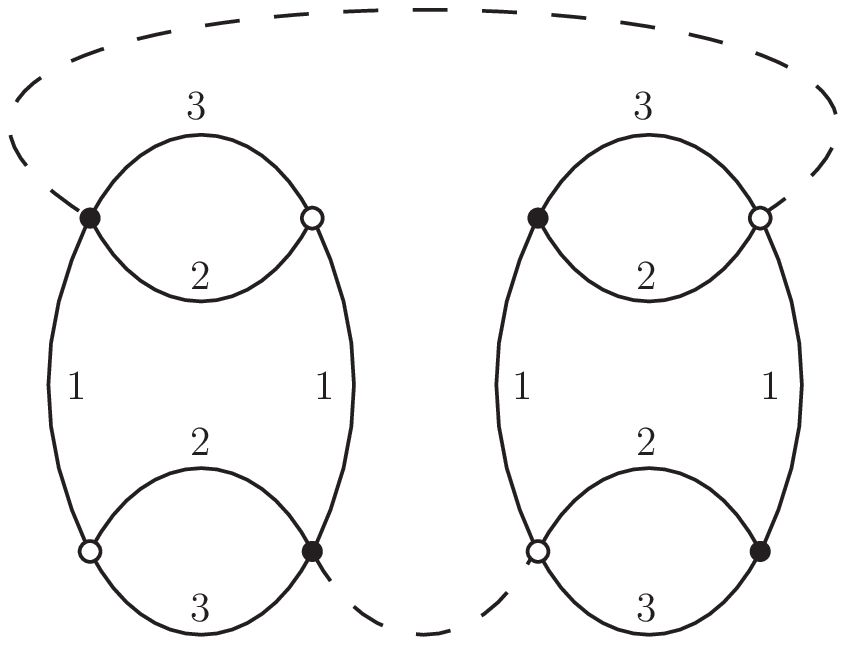}}
  \caption{$\cH_1$ is melonic, $1$-particle reducible, and reducible. $\cH_2$ is $1$-particle irreducible, but reducible.}\label{ex_reduce}
\end{figure}

In the particular model we are considering, all divergent graphs are melonic. It turns out that $1$-particle reducibility and reducibility are equivalent in this restricted context.  
\begin{proposition}
Let $\cG$ be a melonic graph. Then
\beq
\cG \; {\mathrm{irreducible}} \; \Leftrightarrow \; \cG \; 1{\mathrm{-particle}}\; {\mathrm{irreducible}} \,.
\eeq
\end{proposition}
\begin{proof}
The non-trivial implication is $\Leftarrow$. Let $\cG$ be melonic. Let us suppose that $\cG$ is reducible, and that $l \in L(\cG)$ is a reducible line\footnote{That is to say that $l$ verifies the condition of (\ref{irred_def}).}. Then there exists a maximal tree $\cT$ such that $\cG / \cT$ is a melopole. It must be that $l \in \cT$ since $l$, being contained in no internal face, cannot form a $3$-dipole in $\cG/\cT$. Then $\cG / (\cT \setminus \{l \})$ consists in two vertices connected by $l$, and is therefore $1$-particle reducible. Hence $\cG$ is itself $1$-particle reducible. 
\end{proof}

\

Let us now explain why, similarly to $1$-particle reducible graphs in ordinary field theories, the contribution of reducible graphs can be neglected in the renormalization group flow of TGFTs. 
Let $\cG$ be such a graph. Suppose moreover that the internal lines are in the slice $i$, while the boundary data have spins in the slice $i_0 \leq i$. By conservation of the spins along the external faces of $\cG$, the sum over the spins $\{j_1, j_2 , j_3\}$ associated to a reducible line $l \in L(\cG)$ is effectively zero except in the region where both
\beq
\forall i \in \{ 1, 2, 3\}\,, \qquad j_i \sim M^{i_0}
\eeq 
and
\beq
\sum_{i = 1}^3 j_i (j_i + 1) \sim M^{2i}
\eeq
hold. These two conditions cannot be satisfied at the same time when $i - i_0$ grows large, and therefore the amplitude of $\cG$ vanishes in this limit.

\subsection{Melonic flow equations}

Now that the general set up has been introduced, let us derive the flow equations. Given a bubble $b$, let us note $\cM(b)$ the set of melonic (one particle) irreducible graphs which give back $b$ after contraction of their internal lines. We will also note $\cA_{i}(u_2, \cG)$ the amplitude of a graph $\cG$ with respect to the covariance $C_{i,m}$, where $m$ is defined by $m^2 = M^{2i} u_2$, and $[\tau \cA_{i}(u_2 , \cG)]$ its amputated value. The key coefficients entering the flow equations will be expressed in terms of the finite values:
\beq
a(u_2 , \cG)  \equiv \lim_{i \to + \infty} 
\frac{[\tau \cA_{i}(u_2 , \cG)]}{M^{\omega(\cG) i}} < + \infty
\,.
\eeq

\

Let $b \in \cB$ be a bubble of arbitrary valency. Identifying the melonic term in $\cM(b)$ appearing on the right side of equation (\ref{effective_a}), we immediately find: 
\bes
\frac{- \tu_{b, i-1}}{k(b)} M^{d_b (i-1) } &=& \frac{- u_{b, i}}{k(b)} M^{d_b i } \\
&+& \sum_{\cG \in \cM(b)} \frac{1}{k(\cG)} \left( \prod_{b'} (-u_{b',i} M^{d_{b'} i})^{n_{b'}(\cG)} \right) [\tau \cA_{i}(u_{2,i} , \cG)]\,,  \nn
\ees
and therefore:
\beq
\tu_{b, i-1} M^{- d_b } = u_{b, i} - \sum_{\cG \in \cM(b)} \frac{k(b)}{k(\cG)} \left( \prod_{b'} (-u_{b',i})^{n_{b'}(\cG)} \right) \frac{[\tau \cA_{i}(u_{2,i} , \cG)]}{M^{\omega(\cG) i}}\,. 
\eeq

When $i$ is large enough, we can finally approximate the $i$-dependent terms by their limit values, yielding:
\beq\boxed{
\tu_{b, i-1} M^{- d_b } = u_{b, i} - \sum_{\cG \in \cM(b)} \frac{k(b)}{k(\cG)} a(u_{2,i} , \cG) \left( \prod_{b'} (-u_{b',i})^{n_{b'}(\cG)} \right)\,.  }
\eeq

\

Let us now turn to the computation of the wave-function renormalization. By color invariance, we can focus on terms proportional to $\Delta_\ell$ at arbitrary fixed $\ell$.  Let us call $\cM_2$ the set of irreducible and connected $2$-point melonic graphs. For any graph $\cG \in \cM_2$, we can define
\beq
w_\ell (u_2 , \cG)  \equiv \lim_{i \to + \infty} 
\frac{[\tau^{(2)}_\ell \cA_{i}(u_2 , \cG)]}{M^{(\omega(\cG)-2) i}} < + \infty
\,.
\eeq
Then $Z_{i-1}$ can formally be written:
\beq\boxed{
Z_{i-1} = 1 + \sum_{\cG \in \cM_{2}} \frac{1}{k(\cG)} w_\ell (u_2 , \cG) \left( \prod_{b} (-u_{b,i})^{n_{b}(\cG)} \right) \,, }
\eeq     
for any fixed $\ell \in \{1 , 2 , 3\}$.

\section{Relevant and irrelevant directions around the Gaussian fixed point}\label{sec:gaussian}

In this section, we study the behavior of the model around the Gaussian fixed point:
\beq 
\forall b \in \cB\,, \qquad u_b = 0\,.
\eeq

\subsection{Linearized flow equations}

The wave function renormalization does not contribute at linear order around the Gaussian fixed point. One can also set $u_2 = 0$ in all the coefficients $a(u_2 , \cG)$ at this order.
\begin{definition}
Given two bubbles $b, b' \in \cB$, we say that $b$ is smaller than $b'$
($b \leq b'$) if and only if there exists a single-vertex melonic graph on $b'$ which contracts to $b$. If in addition $b \neq b'$, we say that $b$ is strictly smaller than $b'$ ($b < b'$). This defines a partial order relation on $\cB$.
\end{definition}
Given two bubbles $b,b' \in \cB$ such that $b < b'$, let us designate by $\cM(b,b')$ the set of melonic graphs with a single bubble $b'$ and which contract to $b$. We then define the quantity
\beq
\lambda(b,b') \equiv \sum_{\cG \in \cM(b,b')} \frac{k(b)}{k(\cG)} a(0, \cG)\,,
\eeq
which is strictly positive. 
Let us also set $\lambda(b, b) = 1$ for any $b \in \cB$, and $\lambda(b,b') = 0$ for any $b,b'$ such that the condition $b \leq b'$ does not hold.
The linearized flow equations around the Gaussian fixed point are:
\bes\label{lin_flow}
\forall b \in \cB\,, \qquad u_{b, i-1} %&=& M^{d_b} \left[ u_{b,i} + \sum_{b' > b} \lambda(b,b') u_{b',i} \right] + \cO(u^2) \\
&=& M^{d_b} \sum_{b' \geq b} \lambda(b,b') u_{b',i} + \cO(u^2)
\ees

\subsection{Irrelevant, marginal and relevant directions}

We will say that a coupling constant $t_b$ and its dimensionless counterpart $u_b$ are \emph{renormalizable} if $[t_b ] \geq 0$, and are \emph{non-renormalizable} otherwise. Just like in ordinary field theories, the infrared physics is largely independent of the values of the non-renormalizable coupling constants, because they correspond to stable directions of the Gaussian fixed point. On the contrary, a renormalizable coupling constant $u_b$ is: \emph{unstable} against perturbations around the Gaussian fixed point if $[ u_ b ] > 0$; \emph{marginal} if $[u_b] = 0$, in which case stability requires further discussion. 

\

In our discrete setting, this is a consequence of the form of the linearized recursive relation (\ref{lin_flow}). The linear operator mapping the vector $(u_{b,i})_{b\in \cB}$ to $(u_{b,i-1})_{b\in \cB}$ is triangular superior with respect to the order relation $\leq$ on $\cB$. Its diagonal is $(M^{d_b})_{b \in \cB}$. We can find eigendirections $\sigma_b \equiv (\sigma_b (b'))_{b' \in \cB}$ such that $\sigma_b (b) = 1$ and $\sigma_b (b') = 0$ whenever $b' \leq b$ does not hold.
These eigenvectors verify
\beq\label{first_order}
\forall b \in \cB\,, \qquad \sigma_{b , i-1} \approx M^{d_b} \sigma_{b , i}\,. 
\eeq
Unstable directions correspond to $M^{d_b} > 1$ (i.e. $d_b > 0$), and stable ones to $M^{d_b} < 1$ (i.e. $d_b < 0$). When $d_b = 0$, the evolution of linear perturbations in the direction $\sigma_b$ is trivial, and therefore the stability properties of this direction is determined by higher order contributions to the full flow equations.

All the directions $\sigma_b$ such that $N_b \geq 8$ are stable, and hence \textit{irrelevant} in the renormalization sense. Setting perturbations in these directions to $0$ is equivalent to imposing $u_b = 0$ for any $b$ of valency higher or equal to $8$, which we assume in the rest of this paper. 

\

Finally there is one type of $6$-valent interactions which can be also discarded from the outset. Such bubbles, represented in Figure \ref{sing_bubble} cannot generate irreducible melonic graphs. These is fortunate because they represent topologically singular elementary cells (their boundary is a $2$-torus), and therefore they are difficult to interpret geometrically. The fact that we are free to set their coupling constant to $0$ is particularly interesting.

\begin{figure}[h]
\begin{center}
\includegraphics[scale=0.5]{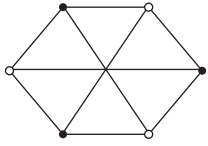}
\caption{Structure of the topologically singular $6$-valent bubbles.}
\label{sing_bubble}
\end{center}
\end{figure}

\

All in all, we have reached the conclusion that the effective action can be well-approximated by three coupling constants: one associated to $4$-valent bubbles, and two associated to $6$-valent interactions. Explicitly, we have:
\beq\label{drawing_sym}
S_i = \frac{t_{4,i}}{2} S_{4} + \frac{t_{6,1,i}}{3} S_{6,1} + t_{6,2,i} S_{6,2}\,,
\eeq
where
\bes\label{color_sym}
%S_{2} (\vphi , \vphib) &=& \int [\extd g ]^3 \, \vphi(g_1 , g_2 , g_3 ) \vphib(g_1 , g_2 , g_3 )\,,\\
S_{4} (\vphi , \vphib) &=& \int [\extd g ]^6 \, \vphi(g_1 , g_2 , g_3 ) \vphib(g_1 , g_2 , g_4 ) \vphi(g_5 , g_6 , g_3 ) \vphib(g_5 , g_6 , g_4 ) \nn \\
&&+ \; {\rm color} \; {\rm permutations}  \,,  \\
S_{6,1} (\vphi , \vphib) &=&  \int [\extd g ]^9 \, \vphi(g_1 , g_2 , g_7 ) \vphib(g_1 , g_2 , g_9 ) \vphi(g_3 , g_4 , g_9 ) \nn \\
&&\vphib(g_3 , g_4 , g_8 ) \vphi(g_5 , g_6 , g_8 ) \vphib(g_5 , g_6 , g_7 ) \\
&& + \; {\rm color} \; {\rm permutations}  \,, \nn \\
S_{6, 2} (\vphi , \vphib) &=& \int [\extd g ]^9 \, \vphi(g_1 , g_2 , g_3 ) \vphib(g_1 , g_2 , g_4 ) \vphi(g_8 , g_9 , g_4 ) \nn \\
&&\vphib(g_7 , g_9 , g_3 ) \vphi(g_7 , g_5 , g_6 ) \vphib(g_8 , g_5 , g_6 ) \nn \\
&& + \; {\rm color} \; {\rm permutations} \,.   
%S_{\vphi} (\vphi , \vphib) &=& \int [\extd g ]^3 \, \vphi(g_1 , g_2 , g_3 ) \left( - \sum_{l = 1}^{3} \Delta_\ell \right) \vphib(g_1 , g_2 , g_3 )\,.
\ees
Each term in $S_i$ is represented by a drawing from the top line of Figure \ref{int} (where colors are left implicit). In $S_4$, $S_{6,1}$ or $S_{6,2}$, there are exactly three distinct bubbles contributing, which correspond to the three possible colorings of the graphs $(4)$, $(6,1)$ or $(6,2)$ respectively. Such bubbles have therefore identical coupling constants $t_b$ by color invariance.  
In addition there are four possible types of counter-terms one generates when lowering the scale from $i$ to $i-1$. The first is mass, associated to the intermediate parameter $CT_{m,i-1}$ and reabsorbed into $u_{2,i-1}$.
The three others are wave-function counter-terms, depending on the color on which the Laplace operator is inserted (which is graphically represented by a cross), associated to a unique parameter $CT_{\vphi , i-1}$ or equivalently $Z_{i-1}$. 

\begin{figure}[h]
\begin{center}
\includegraphics[scale=0.5]{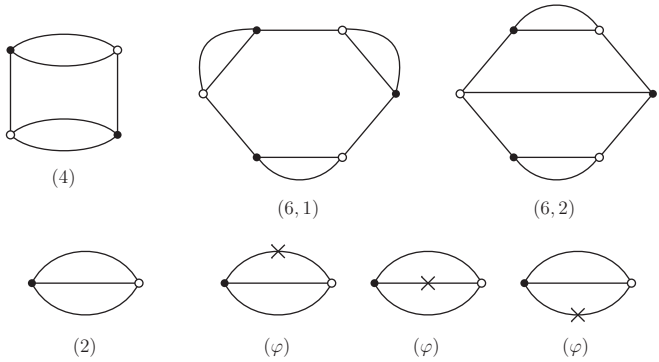}
\caption{Interactions and counter-terms.}
\label{int}
\end{center}
\end{figure}

\subsection{Computation of the relevant and marginal eigendirections}

Let us now investigate how the color invariant parameters $u_2$, $u_4$, $u_{6,1}$ and $u_{6,2}$ vary with the scale $i$.  We will also denote the eigendirections of this system $\sigma_2$, $\sigma_4$, $\sigma_{6,1}$ and $\sigma_{6,2}$. In order to determine them, we first need to compute the coefficients appearing in the recursive equation:
{\footnotesize{
\beq\label{gauss_lin}
\left( \begin{array}{c}
u_{2,i-1} \\
u_{4,i-1} \\
u_{6,1,i-1} \\ 
u_{6,2,i-1} \end{array} \right) = 
\left( \begin{array}{cccc} 
M^2 & M^2 \lambda(2,4) & M^2 \lambda(2, (6,1)) & M^2 \lambda(2, (6,2)) \\
0 & M & M \lambda(4, (6,1)) & M \lambda(4, (6,2)) \\
0 & 0 & 1 & 0 \\
0 & 0 & 0 & 1
\end{array}\right)
\left( \begin{array}{c}
u_{2,i} \\
u_{4,i} \\
u_{6,1,i} \\ 
u_{6,2,i} \end{array} \right)
\eeq}}

\begin{figure}[h]
	\begin{center}
  \includegraphics[scale=0.6]{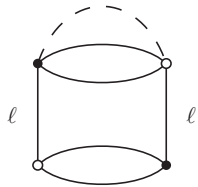}   	       
  \caption{The graphs $G_{1}^{\ell}$ contributing to $\lambda( 2, 4 )$.}\label{g1l}
  \end{center}
\end{figure}

%%%%%%%%%%%%%%%%%%%%%%%%%% Direct way %%%%%%%%%%%%%%%%%%%%%%%%%%%%%%%%%%%%%%%%%%%%%%%%%
%The unique type of graph contributing to $\lambda( 2, 4 )$ is represented in Figure \ref{g1l}. There is one such graph per color $\ell$, which we label $G_{1}^{\ell}$. Each such graph occurs for exactly $2$ Wick pairings. One therefore finds
%\beq
%\lambda(2, 4) = k(2) \sum_{\ell = 1}^{3} 2 \frac{a(0 , G_{1}^{\ell})}{k(4)}
%= 3 \, a(0 , G_{1}^{\ell}) 
%\eeq
%%%%%%%%%%%%%%%%%%%%%%%%% With k(G) %%%%%%%%%%%%%%%%%%%%%%%%%%%%%%%%%%%%%%%%%%%%%%%%%%%%%%%
The unique type of graph contributing to $\lambda( 2, 4 )$ is represented in Figure \ref{g1l}. There is one such graph per color $\ell$, which we label $G_{1}^{\ell}$. Each such graph has a symmetry factor $k(G_{1}^{\ell}) = 1$. One therefore finds
\beq
\lambda(2, 4) = 3 k(2) \, a(0 , G_{1}^{\ell})
= 3 \, a(0 , G_{1}^{\ell}) \,.
\eeq  
The value of $a(0 , G_{1}^{\ell})$, which we denote $\fa$, can be straightforwardly computed:
\bes
\fa &=& \frac{1}{M^i} \int_{M^{-2 i}}^{M^{-2 (i-1) }} \extd \alpha \int \extd h \, [K_\alpha (h)]^2 = \frac{1}{M^i} \int_{M^{-2 i}}^{M^{-2 (i-1) }} \extd \alpha  \, K_{2\alpha} (\one)\\
&=& \frac{\sqrt{4 \pi}}{M^i} \int_{M^{-2 i}}^{M^{-2 (i-1) }} \extd \alpha \, (2 \alpha)^{-3/2}
= \sqrt{4 \pi} \int_{1}^{M^{2}} \extd \alpha \, (2 \alpha)^{-3/2} \\
&=& \sqrt{2 \pi} \left( 1 - \frac{1}{M} \right)\,.
\ees
Hence we have shown that:
\beq
\lambda(2, 4) = 3 \fa = 3 \sqrt{2\pi} \left( 1 - \frac{1}{M} \right)\,.
\eeq

\begin{figure}[h]
	\begin{center}
  \includegraphics[scale=0.6]{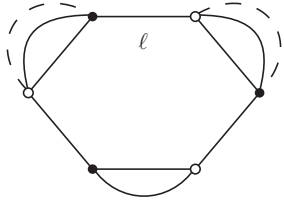}   	       
  \caption{The graphs $G_{2}^{\ell}$ contributing to $\lambda( 2, (6,1) )$.}\label{g2l}
  \end{center}
\end{figure}

We call $G_{2}^{\ell}$ the graphs contributing to $\lambda(2, (6,1))$, as represented in Figure \ref{g2l}. 
%%%%%%%%%%%%%%%%%%%%%%%%%% Direct way %%%%%%%%%%%%%%%%%%%%%%%%%%%%%%%%%%%%%%%%%%%%%%%%%
%There are three such graphs, each of them encapsulating three Wick pairings. 
%Remark moreover that
%\beq
%a(0, G_{2}^{\ell}) = \fa^2 \,,
%\eeq
%and therefore
%\beq
%\lambda(2, (6,1)) = \frac{k(2)}{k(6,1)} \times 9 \fa^2\,. 
%\eeq
%This yields
%\beq
%\lambda(2, (6,1)) = 3 \fa^2 = 6 \pi \left( 1 - \frac{1}{M} \right)^2
%\eeq
%%%%%%%%%%%%%%%%%%%%%%%%%%%%%%%%%%%%%%%%%%%%%%%%%%%%%%%%%%%%%%%%%%%%%%%%%%%%%%%%%%%%%%%%%%
There are three such graphs, each of them having a symmetry factor $k=1$. 
Remark moreover that
\beq
a(0, G_{2}^{\ell}) = \fa^2 \,,
\eeq
and therefore
\beq
\lambda(2, (6,1)) = 3 \fa^2 = 6 \pi \left( 1 - \frac{1}{M} \right)^2\,.
\eeq

\begin{figure}[h]
  \centering
  \subfloat[$G_{3}^\ell$]{\label{g62l}\includegraphics[scale=0.6]{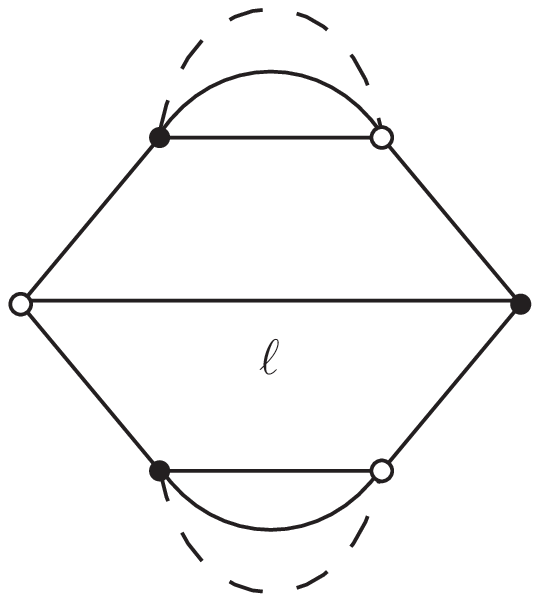}}     
  \subfloat[$G_{4}^{\ell \ell'}$]{\label{g62ll}\includegraphics[scale=0.6]{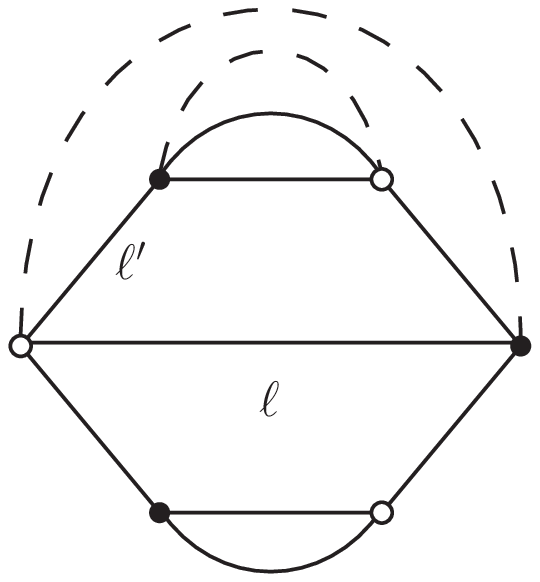}}
  \caption{Graphs contributing to $\lambda(2, (6,2))$.}\label{mass62}
\end{figure}

Two types of graphs contribute to $\lambda(2, (6,2))$, they are represented in Figure \ref{mass62}. The value of a $G_{3}^{\ell}$ is again $\fa^{2}$, and its symmetry factor \linebreak $k(G_3^\ell) = 1$. 
%%%%%%%%%%%%%%%%%%%%%%%%%%%%%%%%%%%%%% Direct way %%%%%%%%%%%%%%%%%%%%%%%%%%%%%%%%%%%%
%They have no degeneracy, and since $k(2) = k(6,2) = 1$, they contribute with a term $3 \fa^{2}$ to $\lambda(2, (6,2))$. 
%%%%%%%%%%%%%%%%%%%%%%%%%%%%%%%%%%%%%%%%%%%%%%%%%%%%%%%%%%%%%%%%%%%%%%%%%%%%%%%%%%%%%%%%%%%%%
These graphs therefore contribute with a term $3 \fa^{2}$ to $\lambda(2, (6,2))$. 
%%%%%%%%%%%%%%%%%%%%%%%%%%%%%%%%%%%%%%%%%%%%%%%%%%%%%%%%%%%%%%%%%%%%%%%%%%%%%%%%
The evaluation of a graph $G_{4}^{\ell \ell'}$ is
\bes
\fb &=& \lim_{i \to + \infty} \frac{1}{M^i} \int_{M^{- 2 i}}^{M^{- 2 (i-1)}} \extd \alpha_1 \extd \alpha_2 \int \extd h_1 \extd h_2 \, \nn \\
&&[K_{\alpha_1} (h_1)]^2 K_{\alpha_1 + \alpha_2} (h_1 h_2) K_{\alpha_2} (h_2) \\
&=& \lim_{i \to + \infty} \frac{1}{M^i} \int_{M^{- 2 i}}^{M^{- 2 (i-1)}} \extd \alpha_1 \extd \alpha_2 \int \extd h_1 \,[K_{\alpha_1} (h_1)]^2 K_{\alpha_1 + 2 \alpha_2} (h_1 )\,.
\ees
In order to evaluate this limit, we resort to a Laplace approximation (see the Appendix), which turns the integral over $g \in \SU(2)$ into a Gaussian integral over a vector $X \in \mathbb{R}^3 \sim \su(2)$. This yields
\beq
\fb = \frac{\sqrt{4 \pi}^3}{16 \pi^2} \int_{1}^{M^2} \extd \alpha_1 \extd \alpha_2 \int_{\mathbb{R}^3} \extd X \frac{\exp\left( - \frac{X^2}{2 \alpha_1}\right)}{{\alpha_1}^3} \frac{\exp\left( - \frac{X^2}{4 (\alpha_1 + 2 \alpha_2)}\right)}{(\alpha_1 + 2 \alpha_2)^{3/2}}\,,
\eeq
and the Gaussian integral can be explicitly performed to give
\bes
\fb &=& 4 \pi \int_{1}^{M^2} \frac{\extd \alpha_1 \extd \alpha_2}{{\alpha_1}^{3/2} (3 \alpha_1 + 4 \alpha_2 )^{3/2}}  \\
&=& \pi \left( \sqrt{7} - \sqrt{3 + \frac{4}{M^2}} + \frac{\sqrt{7}}{M^2} - \frac{\sqrt{3 + 4 M^2}}{M^2} \right) \,. \nn
\ees
%%%%%%%%%%%%%%%%%%%%%%%%%%%%%%%%%%%%%%%%%%%% direct way %%%%%%%%%%%%%%%%%%%%%%%%%%%
%Since there are three graphs of the type $G_{4}^{\ell \ell'}$, with no degeneracy and no additional combinatorial wheight, the net contribution to $\lambda(2 , (6,2))$ is $3 \fb$. 
%%%%%%%%%%%%%%%%%%%%%%%%%%%%%%%%%%%%%%%%%%%%%%%%%%%%%%%%%%%%%%%%%%%%%%%%%%%%%%%%%%%%%%%%%%%%%
Since there are three graphs of the type $G_{4}^{\ell \ell'}$, with trivial symmetry factors, the net contribution to $\lambda(2 , (6,2))$ is $3 \fb$. 
We have thus obtained:
\beq
\lambda(2 , (6,2)) = 3 \fa^2 + 3 \fb\,. 
\eeq

\begin{figure}[h]
	\begin{center}
  \includegraphics[scale=0.6]{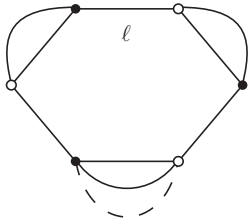}   	       
  \caption{The graphs $G_{5}^{\ell}$ contributing to $\lambda( 4, (6,1) )$.}\label{g5l}
  \end{center}
\end{figure}

The graphs to be taken into account for the computation of $\lambda( 4, (6,1) )$ are those of Figure \ref{g5l}. Their value is $\fa$. 
%%%%%%%%%%%%%%%%%%%%%%%%%%%%%%%%%%%%%%% Direct way %%%%%%%%%%%%%%%%%%%%%%%%%%%%%%%%%%%%%%%
%Each graph $G_{5}^{\ell}$ contributes to the flow of a single $4$-valent bubble, and brings a combinatorial wheight of $3$. One therefore finds
%\beq
%\lambda( 4, (6,1) ) = \frac{k(4)}{k(6,1)} \times 3 \fa = \frac{2}{3}3 \fa
%\eeq   
%and hence
%\beq
%\lambda( 4, (6,1) ) = 2 \fa  = 2 \sqrt{2\pi} \left( 1 - \frac{1}{M}\right)
%\eeq
%%%%%%%%%%%%%%%%%%%%%%%%%%%%%%%%%%%%%%%%%%%%%%%%%%%%%%%%%%%%%%%%%%%%%%%%%%%%%%%%%%%%%%%%%%%
Two graphs $G_{5}^{\ell}$ and $G_{5}^{\ell'}$ with $\ell \neq \ell'$ contribute to the flows of distinct $4$-valent bubbles. The symmetry factor being $k(G_{5}^{\ell}) = 1$, one therefore finds
\beq
\lambda( 4, (6,1) ) = k(4) \fa = 2 \fa = 2 \sqrt{2\pi} \left( 1 - \frac{1}{M}\right)\,.
\eeq

\begin{figure}[h]
	\begin{center}
  \includegraphics[scale=0.6]{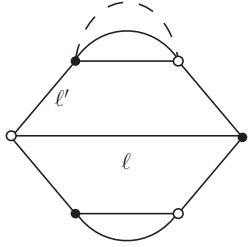}   	       
  \caption{The graphs $G_{6}^{\ell \ell'}$ contributing to $\lambda( 4, (6,2) )$.}\label{g6ll}
  \end{center}
\end{figure}

Similarly, $\lambda(4, (6,2))$ is determined by the graphs $G_{6}^{\ell \ell'}$ (see Figure \ref{g6ll}), which have value $\fa$. 
Two such graphs (mapped into one another by the exchange $\ell \leftrightarrow \ell'$) correspond to a single $4$-valent effective bubble, and \linebreak $k(G_{6}^{\ell \ell'}) = 1$. This shows that
\beq
\lambda(4, (6,2)) = 4 \fa = 4 \sqrt{2\pi} \left( 1 - \frac{1}{M} \right)
\eeq 

\

The eigendirections associated to the dynamical system (\ref{gauss_lin}) can be readily computed. One finds the following (unnormalized) vectors:
\beq
\sigma_2 = \left( \begin{array}{c}
1 \\
0 \\
0 \\ 
0 \end{array} \right) \,, \qquad
\sigma_4 = \left( \begin{array}{c}
- \frac{3 M \fa}{M - 1} \\
1 \\
0 \\ 
0 \end{array} \right) = \left( \begin{array}{c}
- 3 \sqrt{2\pi} \\
1 \\
0 \\ 
0 \end{array} \right) \,, 
\eeq
\bes
\sigma_{6,1} &=& \left( \begin{array}{c}
 \frac{3 M^2}{(M - 1)^2} \fa^2 \\
- \frac{2 M \fa}{M - 1} \\
1 \\ 
0 \end{array} \right) = \left( \begin{array}{c}
6 \pi \\
- 2 \sqrt{2\pi} \\
1 \\ 
0 \end{array} \right) \,, \\
\sigma_{6,2} &=& \left( \begin{array}{c}
 3 M^2 \frac{ (3 \fa^2 - \fb) M + \fa^2 + \fb }{(M^2 - 1)(M-1)} \\
- \frac{4 M \fa}{M - 1} \\
1 \\ 
0 \end{array} \right) \approx \left( \begin{array}{c}
(18 - 3 \sqrt{7} + 3 \sqrt{3}) \pi \\
- 4 \sqrt{2\pi} \\
0 \\ 
1 \end{array} \right) \,.
\ees
In the last equation, we used the approximation $M \gg 1$ to get rid of the $M$ dependence of the first entry. $\sigma_2$ and $\sigma_4$ correspond to unstable and hence relevant directions, while $\sigma_{6,1}$ and $\sigma_{6,2}$ are marginal. The stability of the latter is determined by higher order contributions to the flow equations, which we analyze in the next section. The unstable directions in the $(u_2, u_4)$ plane are represented in Figure \ref{u2u4}.

\

In the deep UV, the exponential suppression of the parameters of the theory along these directions imposes a non-trivial functional dependence of the parameters $u_2$ and $u_4$ from $u_{6,1}$ and $u_{6,2}$. Indeed, assuming that the components of the vector $(u_{2}, u_{4} , u_{6,1} , u_{6,2} )^T$ in the directions $\sigma_2$ and $\sigma_4$ are small but non-vanishing at small scale $i_0$, the first-order equations (\ref{first_order}) impose that they must be respectively of order $M^{-2 (i - i_0 )}$ and $M^{- (i - i_0 )}$ at scales $i$. When $i$ grows very large, we can therefore assume that $(u_{2}, u_{4} , u_{6,1} , u_{6,2} )^T$ vanishes in the directions $\sigma_2$ and $\sigma_4$. This yields:
%\noindent\fbox{
%\begin{minipage}[c]{\linewidth}
\bes
u_{2,i} &\underset{i \to + \infty}{\approx}& 6 \pi \, u_{6,1,i} + (18 - 3 \sqrt{7} + 3 \sqrt{3}) \pi \, u_{6,2,i} \,, \label{u2_UV} \\
u_{4,i} &\underset{i \to + \infty}{\approx}& - 2 \sqrt{2 \pi} \, u_{6,1,i} - 4 \sqrt{2 \pi} \, u_{6,2,i} \,. \label{u4_UV}
\ees
%\end{minipage}
%}

\noindent Hence, the behaviour of the flow in this asymptotic region is determined by that of the two marginal coupling constants.

\begin{figure}[h]
\begin{center}
\includegraphics[scale=0.5]{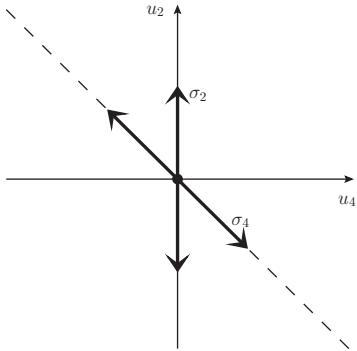}
\caption{Relevant directions in the $(u_2, u_4)$ plane.}
\label{u2u4}
\end{center}
\end{figure}

\section{Second order contributions to the flow of marginal coupling constants}\label{sec:as}

We now focus on second order contributions to the flows of $u_{6,1}$ and $u_{6,2}$. For clarity of the presentation, we successively compute the coupling constants counter-terms and the wave-function renormalization. This partial results are then combined to determine the flow equations.

\subsection{Coupling constants renormalization}

\begin{figure}[h]
  \centering
  \subfloat[$H_{1+}^\ell$]{\label{h1lp}\includegraphics[scale=0.5]{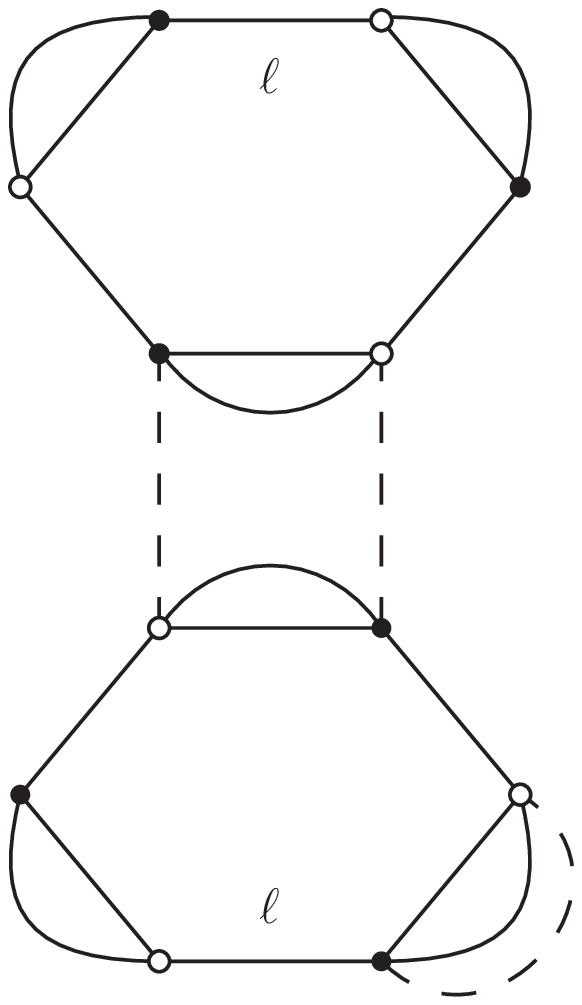}}
  \subfloat[$H_{1-}^\ell$]{\label{h1lm}\includegraphics[scale=0.5]{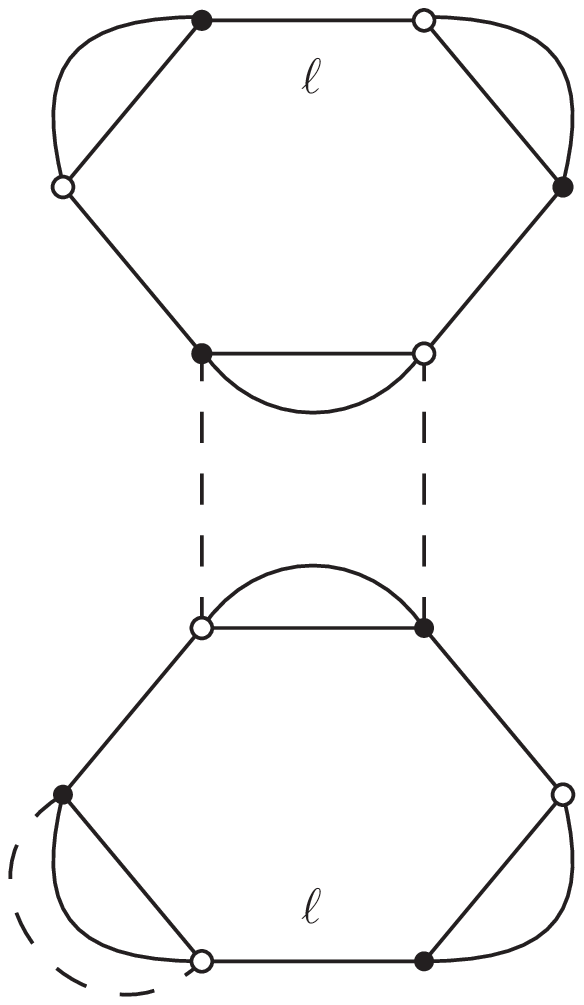}}     
  \subfloat[$H_{2}^{\ell \ell'}$]{\label{h2ll}\includegraphics[scale=0.5]{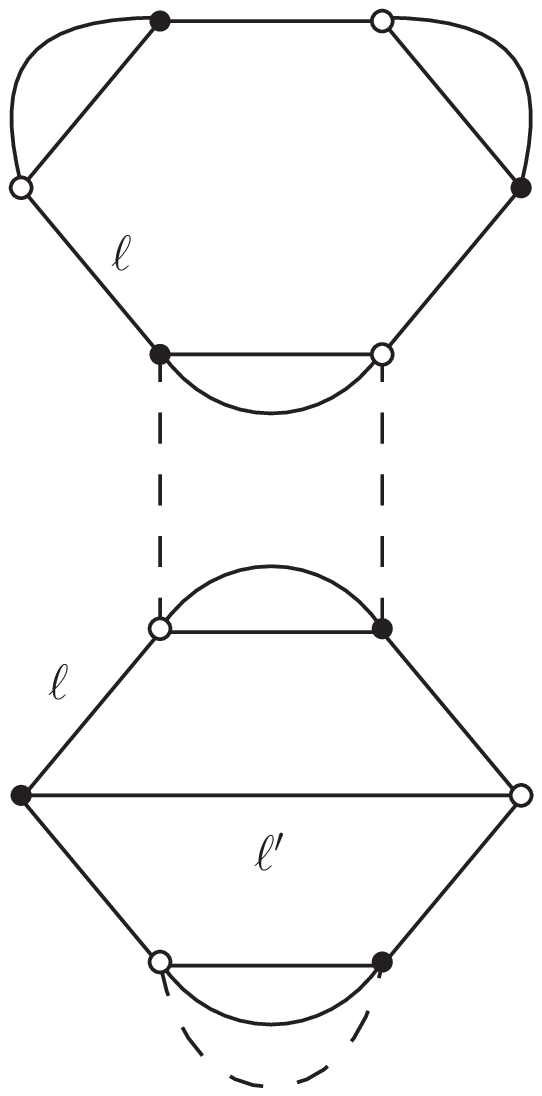}} \\
  \subfloat[$H_{3}^{\ell \ell'}$]{\label{h3ll}\includegraphics[scale=0.5]{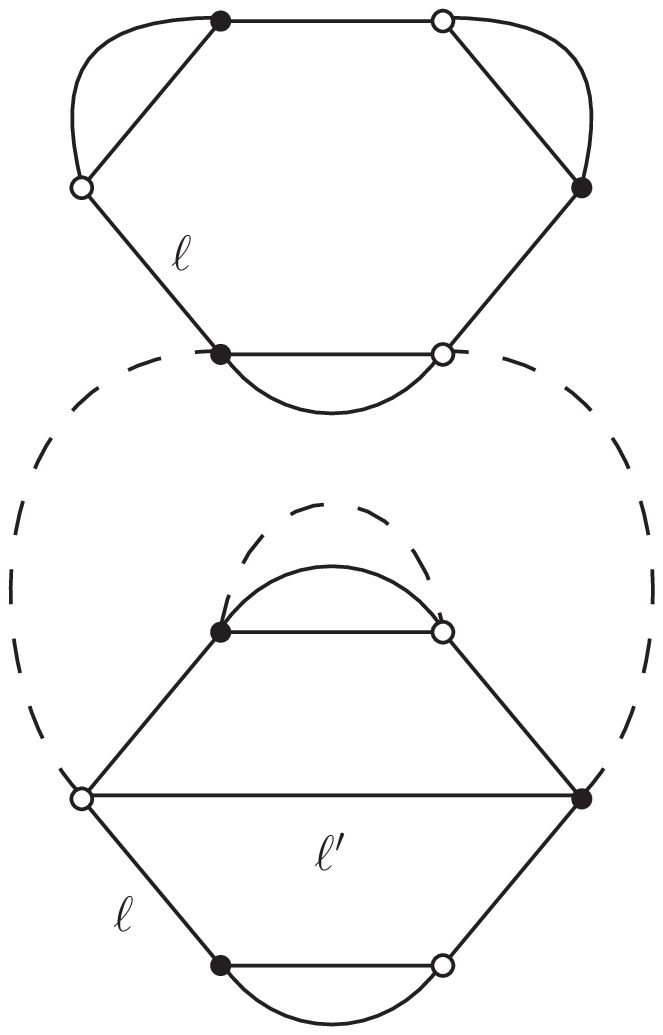}}
  \subfloat[$H_{4}^{\ell}$]{\label{h4l}\includegraphics[scale=0.5]{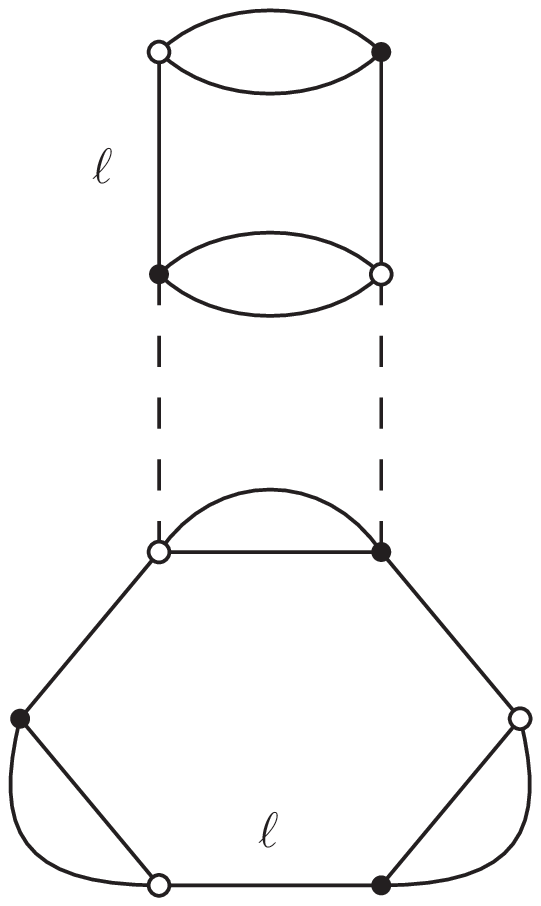}}
  \caption{Second order contributions to $\tu_{6,1}$.}\label{61_2}
\end{figure}
%
%\begin{figure}[h]
%  \centering
%  \begin{subfigure}{0.2\textwidth}\includegraphics[scale=0.35]{figures/h1lp}}\caption{$H_{1+}^\ell$}\label{hlp}
%  \end{subfigure}
%  \begin{subfigure}{0.2\textwidth}\includegraphics[scale=0.35]{figures/h1lm}}\caption{$H_{1-}^\ell$}\label{h1lm}
%  \end{subfigure}
%  \begin{subfigure}{0.2\textwidth}\includegraphics[scale=0.35]{figures/h2ll}}\caption{$H_{2}^{\ell \ell'}$}\label{h2ll}
%  \end{subfigure}
%  \begin{subfigure}{0.2\textwidth}\includegraphics[scale=0.35]{figures/h3ll}}\caption{$H_{3}^{\ell \ell'}$}\label{h3ll}
%  \end{subfigure}
%  \begin{subfigure}{0.2\textwidth}\includegraphics[scale=0.35]{figures/h4l}}\caption{$H_{4}^\ell$}\label{h4l}
%  \end{subfigure}
%  \caption{Second order contributions to $\tu_{6,1}$.}\label{61_2}
%\end{figure}

As far as $\tu_{6,1}$ is concerned, one needs to evaluate the graphs represented in Figure \ref{61_2}. The contracted amplitudes of the graphs $H_{1\pm}^\ell$ factorizes over its two face-connected components. One has already been encountered, and has value $\fa$. The second factor is
\bes
\fc &\equiv& \lim_{i \to + \infty} M^i \int_{M^{-2i}}^{M^{-2(i-1)}} \extd \alpha_1 \extd \alpha_2 \int \extd h_1 \extd h_2\, [K_{\alpha_1 + \alpha_2}(h_1 h_2)]^2\\
&=& \lim_{i \to + \infty} M^i \int_{M^{-2i}}^{M^{-2(i-1)}} \extd \alpha_1 \extd \alpha_2 \int \extd h_1 \,K_{2(\alpha_1 + \alpha_2)}(\one)\,,
\ees
and can be explicitly computed:
\bes
\fc &=& \sqrt{\frac{\pi}{2}} \int_{1}^{M^{2}} \extd \alpha_1 \int_{1}^{M^{2}} \extd \alpha_2 \, \frac{1}{(\alpha_1 + \alpha_2)^{3/2}} \nn \\
&=& 4 \sqrt{\pi} \left( \sqrt{2} \sqrt{M^2 + 1} - M - 1 \right)\,.
\ees
%%%%%%%%%%%%%%%%%%%%%%%%%%%%%%%%%%%%%%%%%%% By hand %%%%%%%%%%%%%%%%%%%%%%%%%%
%At fixed $\ell$, $2 \times 3 \times 3 \times 2 = 36$ Wick contractions yield a graph $H_{1\pm}^\ell$: indeed, we are free to exchange the two identical bubbles, to choose among three ways of connecting each bubble to the other one, and finally there is an additional factor $2$ due to the choice of sign $\pm$ labeling the graphs. The net contribution of these graphs to $\tu_{6,1,i-1}$ is therefore:
%\beq
%- k(6,1) \times 36 \, \fa \, \fc \times \frac{1}{2!} \left( \frac{u_{6,1,i}}{k(6,1)} \right)^2
%= - 6 \, \fa \, \fc \, {u_{6,1,i}}^2 \,. 
%\eeq 
%%%%%%%%%%%%%%%%%%%%%%%%%%%%%%%%%%%%%%%%%%%%%%%%%%%%%%%%%%%%%%%%%%%%%%%%%%%%%%%%%%%%%%%%%%%%%%%
The symmetry factors are trivial: $k(H_{1\pm}^\ell) =1$. At fixed $\ell$, there are two graphs contributing, corresponding to the $\pm$ labeling of the graphs. The net contribution to $\tu_{6,1,i-1}$ is therefore:
\beq
- k(6,1) \times 2 \, \fa \, \fc {u_{6,1,i}}^2
= - 6 \, \fa \, \fc \, {u_{6,1,i}}^2 \,. 
\eeq

Each graph $H_{2}^{\ell \ell'}$ also evaluates to $\fa \times \fc$. 
%%%%%%%%%%%%%%%%%%%%%%%%%%%%%%%%%%%% by hand %%%%%%%%%%%%%%%%%%%%%%%%%%%%%%%%%%%%%%%
%At fixed $\ell$, there are two possible choices for the color $\ell'$, and three Wick %contractions yielding 
%the graph $H_{2}^{\ell \ell'}$. 
%Therefore the total contribution of these graphs to $\tu_{6,1,i-1}$ is:
%\beq
%- k(6,1) \times 6 \, \fa \, \fc \times \frac{u_{6,1,i} u_{6,2,i}}{k(6,1) k(6,2)} 
%= - 6 \, \fa \, \fc \, u_{6,1,i} u_{6,2,i}\,. 
%\eeq  
%%%%%%%%%%%%%%%%%%%%%%%%%%%%%%%%%%%%%%%%%%%%%%%%%%%%%%%%%%%%%%%%%%%%%%%%
At fixed $\ell$, there are two possible choices for the color $\ell'$, and the symmetry factors are again trivial.
Therefore the total contribution of these graphs to $\tu_{6,1,i-1}$ is:
\beq
- k(6,1) \times 2 \, \fa \, \fc \, u_{6,1,i} u_{6,2,i}
= - 6 \, \fa \, \fc \, u_{6,1,i} u_{6,2,i}\,. 
\eeq  

Let us now turn to the graphs $H_{3}^{\ell \ell'}$. Their contracted amplitudes evaluate to:
\bes
\fd &\equiv& \lim_{i \to + \infty} \int_{M^{-2i}}^{M^{-2(i-1)}} \extd \alpha_1 \extd \alpha_2 \extd \alpha_3 \int \extd h_1 \extd h_2 \extd h_3 \, \\
&& [K_{\alpha_1}(h_1 )]^2 K_{\alpha_1 + \alpha_2 + \alpha_3} (h_1 h_2 h_3) K_{\alpha_2 + \alpha_3 } (h_2 h_3) \nn\\
&=& \lim_{i \to + \infty} \int_{M^{-2i}}^{M^{-2(i-1)}} \extd \alpha_1 \extd \alpha_2 \extd \alpha_3 \int \extd h_1 \, \\
&& [K_{\alpha_1}(h_1 )]^2 K_{\alpha_1 + 2(\alpha_2 + \alpha_3)} (h_1) \nn \\
&=& \frac{\sqrt{4 \pi}^{3}}{16 \pi^2} \int_{1}^{M^2} \extd \alpha_1 \extd \alpha_2 \extd \alpha_3 \int_{\mathbb{R}^3} \extd X \, \frac{\e^{- \frac{X^2}{2 \alpha_1}}}{{\alpha_1}^{3}} \frac{\e^{- \frac{X^2}{4 [\alpha_1 + 2 (\alpha_2 + \alpha_3)]}}}{[\alpha_1 + 2 (\alpha_2 + \alpha_3)]^{3/2}} \\
&=& 4 \pi \int_{1}^{M^2} \frac{\extd \alpha_1 \extd \alpha_2 \extd \alpha_3}{{\alpha_1}^{3/2} [3 \alpha_1 + 4 (\alpha_2 + \alpha_3) ]^{3/2}}\,.
\ees
%%%%%%%%%%%%%%%%%%%%%%%%%%%%%%%%%%%%%% By hand %%%%%%%%%%%%%%%%%%%%%%%%%%%%%%%%%%%
%At fixed $\ell$ there are two possible choices for $\ell'$ and three Wick contractions producing the graph $H_{3}^{\ell \ell'}$, therefore these graphs contribute to $\tu_{6,1,i-1}$ with a term:
%\beq
%- k(6,1) \times 6 \, \fd \times \frac{u_{6,1,i} u_{6,2,i}}{k(6,1) k(6,2)} 
%= - 6 \, \fd \, u_{6,1,i} u_{6,2,i}\,. 
%\eeq  
%%%%%%%%%%%%%%%%%%%%%%%%%%%%%%%%%%%%%%%%%%%%%%%%%%%%%%%%%%%%%%%%%%%%%%%%%%%%%%%%%%%%%%
One can easily see that $k(H_{3}^{\ell \ell'}) = 1$. Since at fixed $\ell$ there are two possible choices for $\ell'$ which generate the same effective bubble, we conclude that these graphs produce a term
\beq
- k(6,1) \times 2 \, \fd \, u_{6,1,i} u_{6,2,i}
= - 6 \, \fd \, u_{6,1,i} u_{6,2,i} 
\eeq 
in the expression of $\tu_{6,1,i-1}$.

%%%%%%%%%%%%%%%%%%%%%%%%%%%%%%% By hand %%%%%%%%%%%%%%%%%%%%%%%%%%
%Finally, the contracted amplitude of a graph $H_{4}^{\ell}$ is $\fc$.
%Each of them is generated by $2 \times 3 = 6$ Wick contractions, thus yielding the following overall contribution to $\tu_{6,1,i-1}$:
%\beq
%- k(6,1) \times 6 \, \fc \times \frac{u_{6,1,i} u_{4,i}}{k(6,1) k(4)} 
%= - 3 \, \fc \, u_{6,1,i} u_{4,i}\,. 
%\eeq 
%%%%%%%%%%%%%%%%%%%%%%%%%%%%%%%%%%%%%%%%%%%%%%%%%%%%%%%%%%%%%%%%%%%%%%%%%%%%%
Finally, the contracted amplitude of a graph $H_{4}^{\ell}$ is $\fc$, and $k(H_{4}^{\ell}) = 1$.
Each of them generates a distinct effective bubble. Therefore this graphs result in a term:
\beq
- k(6,1) \times \fc \,  u_{6,1,i} u_{4,i} 
= - 3 \, \fc \, u_{6,1,i} u_{4,i}\,. 
\eeq

In summary, we have just shown that:
\beq\label{result_61}
\boxed{
\tu_{6,1,i-1} = u_{6,1,i} - 6 \, \fa \, \fc \, {u_{6,1,i}}^2 - 6 \left( \fa \, \fc + \fd \right) u_{6,1,i} u_{6,2,i}  - 3 \, \fc \, u_{6,1,i} u_{4,i} + \cO(u^3)}
\eeq

\

\begin{figure}[h]
  \centering
  \subfloat[$I_{1 +}^{\ell \ell'}$]{\label{i1llp}\includegraphics[scale=0.5]{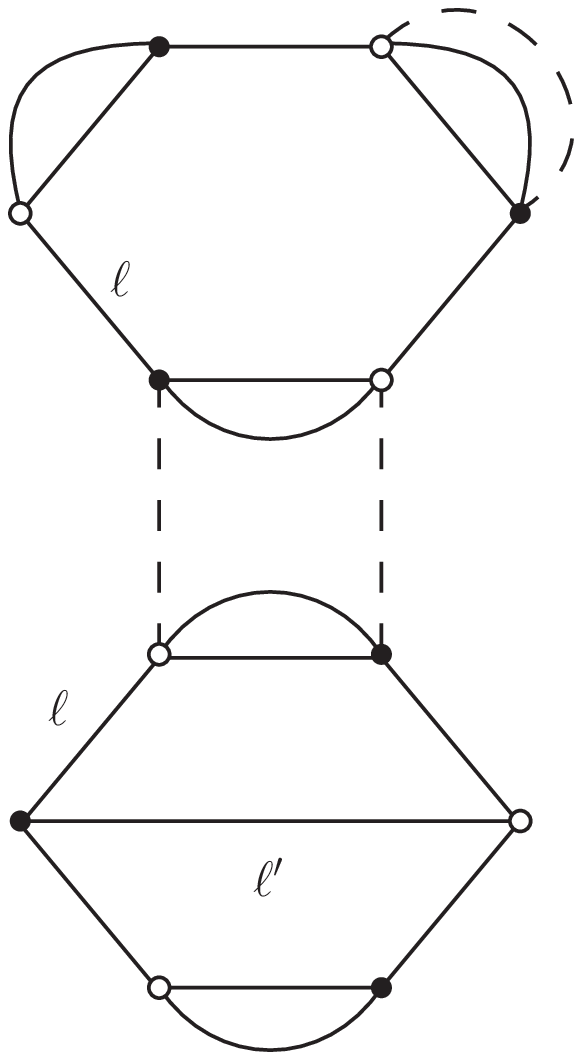}}
  \subfloat[$I_{1 -}^{\ell \ell'}$]{\label{i1llm}\includegraphics[scale=0.5]{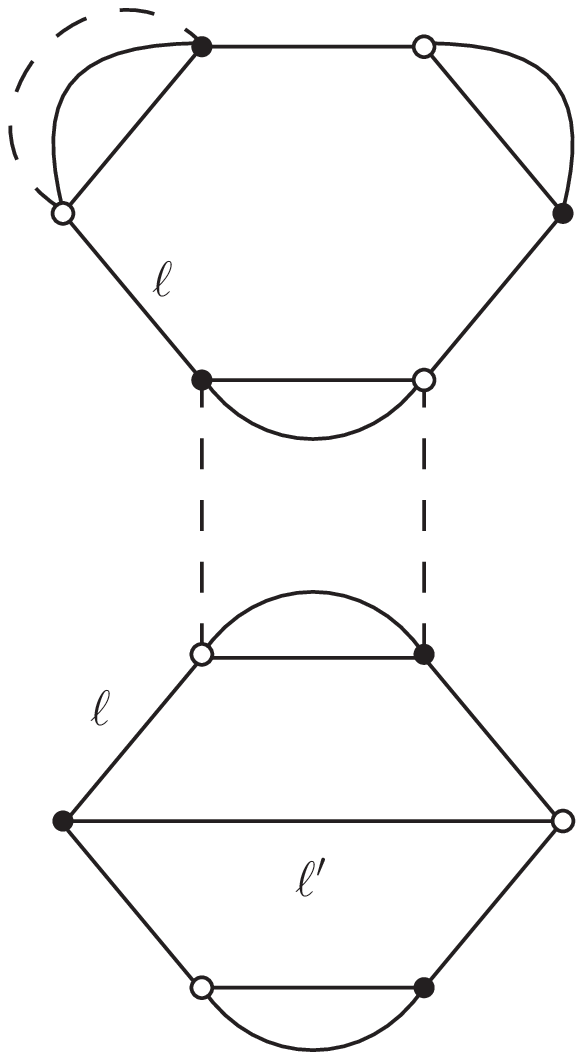}}
  \subfloat[$I_{2}^{\ell \ell' \ell'' }$]{\label{i2lll}\includegraphics[scale=0.5]{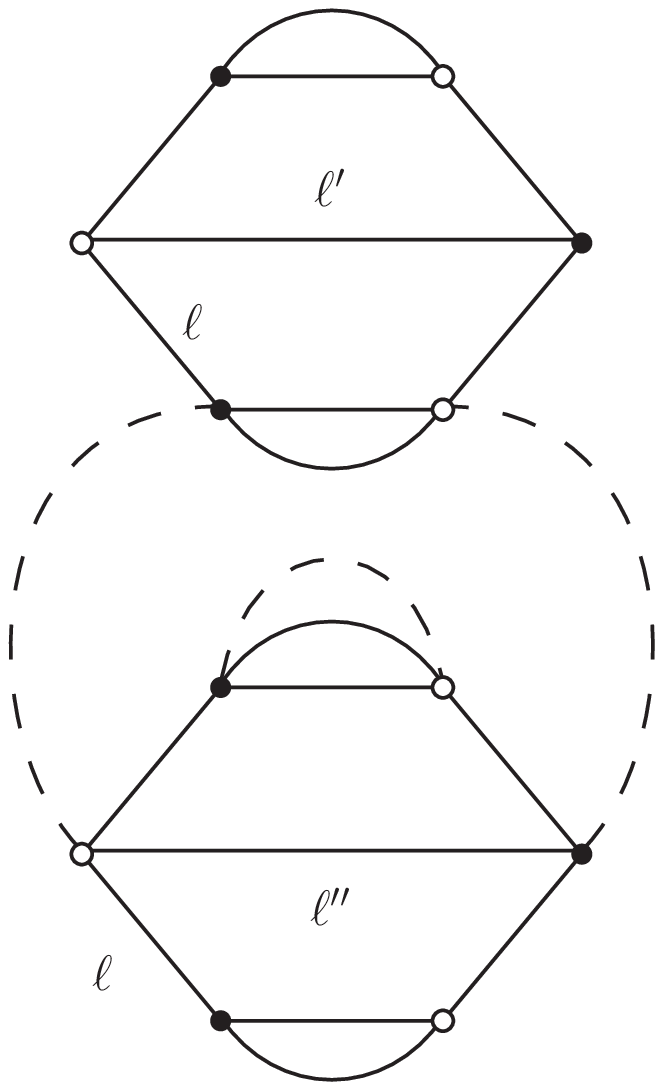}} \\     
  \subfloat[$I_{3}^{\ell \ell'}$]{\label{i3ll}\includegraphics[scale=0.5]{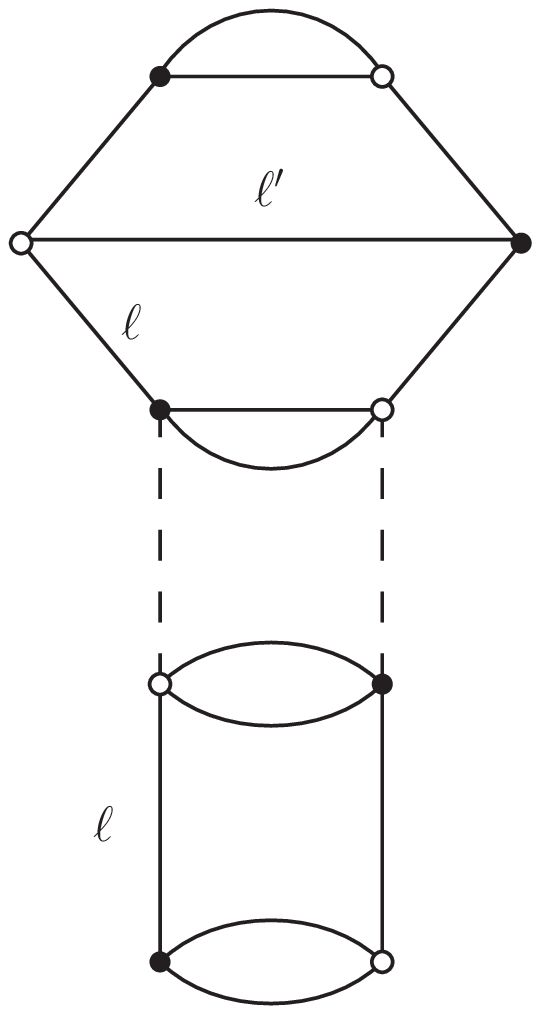}}
  \subfloat[$I_{4}^{\ell \ell' \ell'' }$]{\label{i4lll}\includegraphics[scale=0.5]{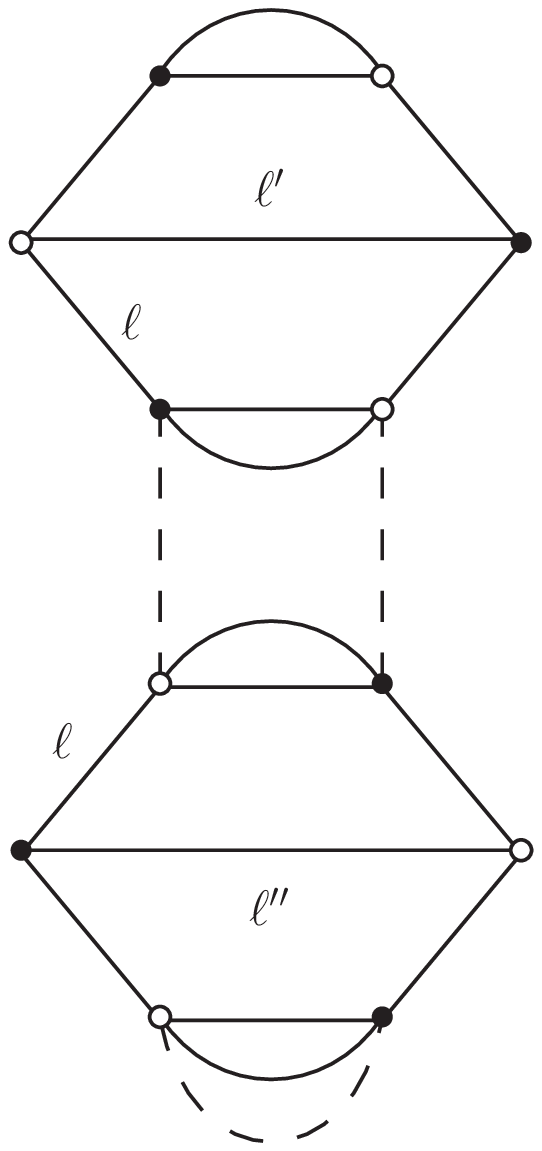}}
  \caption{Second order contributions to $\tu_{6,2}$.}\label{62_2}
\end{figure}

We proceed in a similar way to determine $\tu_{6,2,i-1}$. The four types of graphs contributing are represented in Figure \ref{62_2}. 

The graphs $I_{1\pm}^{\ell \ell'}$ evaluate to $\fa \times \fc$. 
%%%%%%%%%%%%%%%%%%%%%%%%%%%%%%% By hand %%%%%%%%%%%%%%%%%%%%%%%%%%%%%%%%%%%%%%%%%%%%%%%
%At fixed $\ell'$, we are free to choose between $2$ values of $\ell$, and $3$ Wick contractions yield the same $I_{1\pm}^{\ell \ell'}$. Taking the sign $\pm$ into account, one therefore obtains a total combinatorial factor of $12$. Therefore these graph are responsible for a term 
%\beq
%- k(6,2) \times 6 \, \fa \, \fc \times \frac{u_{6,1,i} u_{6,2,i}}{k(6,1) k(6,2)} = - 4 \, \fa \, \fc \, u_{6,1,i} u_{6,2,i}
%\eeq
%in the computation of $\tu_{6,2,i-1}$.
%%%%%%%%%%%%%%%%%%%%%%%%%%%%%%%%%%%%%%%%%%%%%%%%%%%%%%%%%%%%%%%%%%%%%%%%%%%%%%%%%%%%%%%% 
At fixed $\ell'$, we are free to choose between two values of $\ell$ and the sign $\pm$, therefore four distinct graphs generate the same effective bubble. Since the symmetry factors are trivial, these graphs are responsible for a term 
\beq
- k(6,2) \times 4 \, \fa \, \fc \, u_{6,1,i} u_{6,2,i} = - 4 \, \fa \, \fc \, u_{6,1,i} u_{6,2,i}
\eeq
in the computation of $\tu_{6,2,i-1}$.

The amplitude of any graph $I_2^{\ell \ell' \ell''}$ is $\fd$. 
%%%%%%%%%%%%%%%%%%%%%%%%%%%%%%%%%  By hand %%%%%%%%%%%%%%%%%%%%%%%%%%%%%%%%%%%%%%%%%%%%%%%%%%%
%In order to determine the correct combinatorial factors, let us distinguish two cases: $\ell' = \ell''$ and $\ell' \neq \ell'$. When $\ell' = \ell''$, the two vertices are identical, and therefore we can permute them without changing the graph. At fixed $\ell'$, there are moreover $2$ possible choices for $\ell$, hence the term generated by this subcase is
%\beq
%- k(6,2) \times 4 \, \fd \times \frac{1}{2!} \left( \frac{u_{6,2, i}}{k(6,2)} \right)^2 = - 2 \, \fd \, {u_{6,2, i}}^2\,.
%\eeq
%When $\ell \neq \ell'$, the two vertices are different and the only freedom is in the choice of $\ell$. Hence this generates a term:
%\beq
%- k(6,2) \times 2 \, \fd \times  \left( \frac{u_{6,2, i}}{k(6,2)} \right)^2 = - 2 \, \fd \, {u_{6,2, i}}^2\,.
%\eeq
%All in all, the graphs $I_2^{\ell \ell' \ell''}$ contribute with a term $- 4 \, \fd \, {u_{6,2, i}}^2$.
%%%%%%%%%%%%%%%%%%%%%%%%%%%%%%%%%%%%%%%%%%%%%%%%%%%%%%%%%%%%%%%%%%%%%%%%%%%%%%%%%%%%%%%%%%%%%
The symmetry factors are trivial, and the value of $\ell''$ uniquely determines the effective bubble generated by the graph. Once $\ell''$ is fixed one has two possibilities for $\ell'$, and two for $\ell$. Hence these graphs produce a term 
\beq
- k(6,2) \times 4 \, \fd \, {u_{6,2, i}}^2 = - 4 \, \fd \, {u_{6,2, i}}^2\,.
\eeq

The graphs $I_{3}^{\ell \ell'}$ generate a term proportional to $\fc$. 
%%%%%%%%%%%%%%%%%%%%%%%%%%%%%%%%%%%%%% By hand %%%%%%%%%%%%%%%%%%%%%%%%%%%%%%%%%%%%%%%%%%%
%At fixed $\ell$, there is a freedom in the choice of $\ell'$ and in the way of connecting the $4$-valent bubble to the $6$-valent one, giving a total combinatorial factor of $4$. This results in the following contribution to $\tu_{6, 2, i-1}$:
%\beq
%- k(6,2) \times 4 \, \fc \times \frac{u_{6,2,i} u_{4,i}}{k(6,2) k(4)} = - 2 \, \fc \, u_{6,2,i} u_{4,i}\,.  
%\eeq
%%%%%%%%%%%%%%%%%%%%%%%%%%%%%%%%%%%%%%%%%%%%%%%%%%%%%%%%%%%%%%%%%%%%%%%%%%%%%%%%%%%%%%%%%%%%%%%%
At fixed $\ell'$, there is a freedom in the choice of $\ell$. Since $k(I_{3}^{\ell \ell'}) = 1$ this results in the following contribution to $\tu_{6, 2, i-1}$:
\beq
- k(6,2) \times 2 \, \fc \times \, u_{6,2,i} u_{4,i} = - 2 \, \fc \, u_{6,2,i} u_{4,i}\,.  
\eeq

Finally, the graphs $I_{4}^{\ell \ell' \ell''}$ generate amplitudes evaluating to $\fa \times \fc$. At fixed $\ell'$, one is free to choose between two values for $\ell'$ and two values for $\ell''$. The symmetry factor is $k(I_{4}^{\ell \ell' \ell''})=1$, hence the total contribution of these graphs is:
\beq
- k(6,2) \times 4 \, \fa \, \fc \times \, {u_{6,2,i}}^2 = - 4 \, \fa \, \fc \, {u_{6,2,i}}^2\,. 
\eeq

This concludes the computation of $\tu_{6, 2, i-1}$:
\beq\label{result_62}
\boxed{
\tu_{6, 2, i-1} = u_{6,2,i} - 4 \left( \fa \, \fc +  \fd \right) {u_{6,2,i}}^2 - 4 \, \fa \, \fc \, u_{6,1,i} u_{6,2,i}  - 2 \, \fc \, u_{6,2,i} u_{4,i} + \cO(u^3) }
\eeq

\subsection{Wave-function renormalization}

\begin{figure}[h]
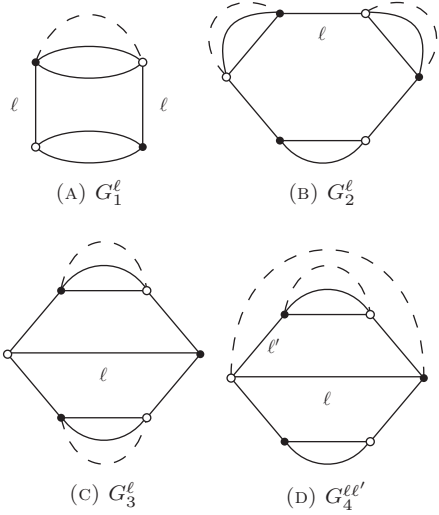

  \centering
  \subfloat[$G_{1}^\ell$]{\label{g1lbis}\includegraphics[scale=0.5]{figures/g1l}}
  \subfloat[$G_{2}^\ell$]{\label{g2lbis}\includegraphics[scale=0.5]{figures/g2l}} \\
  \subfloat[$G_{3}^\ell$]{\label{g3lbis}\includegraphics[scale=0.5]{figures/g3l}}     
  \subfloat[$G_{4}^{\ell \ell'}$]{\label{g4ll}\includegraphics[scale=0.5]{figures/g4ll}}
  \caption{Graphs contributing to $Z_{i-1}$.}\label{wf}
\end{figure}

The wave-function renormalization is determined, at this order, by the graphs $G_{1}^\ell$, $G_{2}^\ell$, $G_{3}^\ell$ and $G_{4}^{\ell \ell'}$, more precisely by their Taylor expansion at order two around their local part. 

Each graph $G_{1}^\ell$ generates a term:
\bes
\fe &\equiv& w_{\ell}(0, G_{1}^\ell) = \lim_{i \to + \infty} M^i \times \frac{1}{6} \int \extd g \vert X_g \vert^2  \int_{M^{-2i}}^{M^{-2(i-1)}} \extd \alpha \, \\
&& \hfill [K_{\alpha}(h)]^2 K_{\alpha} (hg) \\
&=& \frac{1}{6} M^i \sqrt{4 \pi}^3 \left(\frac{1}{16 \pi^2}\right)^2 \int_{M^{-2i}}^{M^{-2(i-1)}} \extd \alpha \nn \\
&& \int_{\mathbb{R}^3} \extd X \, \vert X \vert^2 \int_{\mathbb{R}^3} \extd H  \, \frac{\e^{- \frac{H^2}{2 \alpha}}}{{\alpha}^{3}} \frac{\e^{- \frac{(H + X)^2}{4 \alpha}}}{{\alpha}^{3/2}}\,.
\ees
In order to evaluate the Gaussian integral with respect to $H$, we re-arrange the quadratic form appearing in the exponentials:
\beq
\frac{H^2}{2 \alpha} + \frac{\left( H + X \right)^2}{4 \alpha} = \frac{3}{4 \alpha} \left( H + \frac{X}{3} \right)^2 + \frac{X^2}{6 \alpha}\,, 
\eeq
so that:
\beq
\fe = \frac{1}{6} M^i \sqrt{4 \pi}^3 \left(\frac{1}{16 \pi^2}\right)^2 \left( \frac{4 \pi}{3} \right)^{3/2} \int_{M^{-2i}}^{M^{-2(i-1)}} \extd \alpha \int_{\mathbb{R}^3} \extd X \, \vert X \vert^2 \, \frac{\e^{- \frac{X^2}{6 \alpha}}}{\alpha^3}\,.
\eeq
To evaluate the integral over $X$, we make use of the general formula:
\bes
\int_{\mathbb{R}^d} \extd X \, X^2 \,\e^{- a X^2} &=& -\frac{\extd }{\extd a} \int_{\mathbb{R}^d} \extd X \, \e^{- a X^2} = -\frac{\extd }{\extd a} \left( \frac{\pi}{a} \right)^{d/2} \nn \\
&=& \frac{d}{2 a} \left( \frac{\pi}{a}\right)^{d/2}\,.
\ees
This finally yields:
\beq
\fe = \frac{3}{2} \sqrt{\frac{\pi}{2}} \int_{1}^{M^2} \frac{\extd \alpha}{\sqrt{\alpha}} = 3 \sqrt{ \frac{\pi}{2} } \left( M - 1 \right)\,. 
\eeq
%%%%%%%%%%%%%%%%%%%%%%%%%%%%%%%%%%%%%%%%%%% By hand %%%%%%%%%%%%%%%%%%%%%%%%%%%%%%%%%%%%%%
%The combinatorial factor associated to each such graph is $2$, therefore all these graphs generate a term
%\beq
%- 2 \, \fe \times \frac{u_{4,i}}{k(4)} = - \fe \, u_{4,i} 
%\eeq
%in the expression of $Z_{i-1}$.
%%%%%%%%%%%%%%%%%%%%%%%%%%%%%%%%%%%%%%%%%%%%%%%%%%%%%%%%%%%%%%%%%%%%%%%%%%%%%%%%%%%%%%%%%%%
The symmetry factor associated to each such graph is $1$, therefore they collectively generate a term
\beq
- \fe \, u_{4,i} 
\eeq
in the expression of $Z_{i-1}$.

The contribution of a graph $G_{2}^\ell$ can be evaluated in a similar way. It is
\bes
\ff &\equiv& w_{\ell}(0, G_{2}^\ell) = \lim_{i \to + \infty}   \frac{1}{6} \int \extd g \vert X_g \vert^2  \int_{M^{-2i}}^{M^{-2(i-1)}} \extd \alpha_1 \extd \alpha_2 \, \nn \\
&& \hfill [K_{\alpha_1}(h_1)]^2 [K_{\alpha_2}(h_2)]^2 K_{\alpha_1 + \alpha_2} (h_1 h_2 g) \\
&=& \frac{1}{6}  \sqrt{4 \pi}^5 \left(\frac{1}{16 \pi^2}\right)^3 \int_{M^{-2i}}^{M^{-2(i-1)}} \extd \alpha_1 \extd \alpha_2 \, \nn \\ 
&&\int_{\mathbb{R}^3} \extd X \, \vert X \vert^2 \int_{\mathbb{R}^3} \extd H_1  \, \int_{\mathbb{R}^3} \extd H_2  \, \frac{\e^{-Q(X, H_1, H_2) }}{{\alpha_1}^{3} {\alpha_2}^{3} (\alpha_1 + \alpha_2)^{3/2}} \,,
\ees
where the quadratic form $Q$ is defined by:
\beq
Q(X , H_1 , H_2 ) = \frac{{H_1}^2}{2 \alpha_1} + \frac{{H_2}^2}{2 \alpha_2} + \frac{\left( H_1 + H_2 + X \right)^2}{4 (\alpha_1 + \alpha_2)}\,.
\eeq
One can re-arrange $Q$ in such a way that one can easily integrate $H_1$ and $H_2$ successively, for instance:
\bes
Q(X , H_1 , H_2 ) &=& \frac{3 \alpha_1 + \alpha_2}{4 \alpha_1 (\alpha_1 + \alpha_2)} \left( H_1 + \frac{\alpha_1}{3 \alpha_1 + 2 \alpha_2} (H_2 + X) \right)^2 \nn \\
&& + \frac{3}{2} \frac{\alpha_1 + \alpha_2}{(3 \alpha_1 + 2 \alpha_2) \alpha_2} \left( H_2 + \frac{1}{3} \frac{\alpha_2}{\alpha_1 + \alpha_2} X \right)^2  \\
&& + \frac{1}{6(\alpha_1 +  \alpha_2 )} X^2\,. \nn
\ees
Hence:
\bes
\ff &=& \frac{1}{6}  (4 \pi)^{- 7/2} \int_{M^{-2i}}^{M^{-2(i-1)}} \extd \alpha_1 \extd \alpha_2 \, \int_{\mathbb{R}^3} \extd X \, \vert X \vert^2 \, \frac{\e^{-\frac{X^2}{6 (\alpha_1 + \alpha_2)}}}{{\alpha_1}^{3} {\alpha_2}^{3} (\alpha_1 + \alpha_2)^{3/2}} \nn \\
&& \times \left( \frac{4 \pi \alpha_1 (\alpha_1 + \alpha_2 ) }{ 3 \alpha_1 + 2 \alpha_2} \right)^{3/2} 
\left( \frac{ 2 \pi (3 \alpha_1 + 2 \alpha_2) \alpha_2}{3(\alpha_1 + \alpha_2)} \right)^{3/2}  \\
&=& \frac{3 \pi}{4} \int_{1}^{M^{2}} \extd \alpha_1 \extd \alpha_2 \, \frac{\alpha_1 + \alpha_2 }{ {\alpha_1}^{3/2} {\alpha_2}^{3/2}} 
= \frac{3 \pi}{2} \int_{1}^{M^{2}} \extd \alpha_1 \extd \alpha_2 \, \frac{1}{{\alpha_1}^{1/2} {\alpha_2}^{3/2}} \\
&=& 6 \pi \left( M - 1 \right) \left( 1 - \frac{1}{M}\right) = 6 \pi \left( M - 2 + \frac{1}{M}\right)\,.
\ees
%%%%%%%%%%%%%%%%%%%%%%%%%%%%%%%% By hand %%%%%%%%%%%%%%%%%%%%%%%%%%%%%%%%%%%%%%%%%%%%
%The combinatorial factor associated to the graphs $G_{2}^\ell$ is $3$, yielding a total contribution to $Z_{i-1}$ of:
%\beq
%- 3 \, \ff \times \frac{u_{6,1,i}}{k(6,1)} = - \ff \, u_{6,1,i}\,. 
%\eeq
%%%%%%%%%%%%%%%%%%%%%%%%%%%%%%%%%%%%%%%%%%%%%%%%%%%%%%%%%%%%%%%%%%%%%%%%%%%%%%%%%%%%%%
Since $k(G_{2}^\ell)=1$, the contribution of this graphs to $Z_{i-1}$ is:
\beq
- \ff \, u_{6,1,i}\,. 
\eeq

Let us now consider a graph $G_{3}^\ell$, and call $\ell'$, $\ell"$ the two remaining colors. By expanding the amplitude of the graph to Taylor order $2$ around its tensor invariant part, one easily sees that a kernel of the form
\beq
\fa \, \fe \, u_{6,2,i} \Delta_{\ell'} + \fa \, \fe \,  u_{6,2,i} \Delta_{\ell''}
\eeq
is generated. By varying $\ell$ one therefore generates $6 \, \fa \, \fe \, u_{6,2,i}$ Laplace operators, which split equitably between the three internal colors. Therefore the total contribution of these graphs to $Z_{i-1}$ is:
\beq
- 2 \fa \, \fe \, u_{6,2,i}\,.
\eeq

Finally, we can examine the $G_4^{\ell \ell'}$ contributions. Each of these amplitudes produces a term:
\bes
\fg &\equiv& \lim_{i \to + \infty}   \frac{1}{6} \int_{M^{-2i}}^{M^{-2(i-1)}} \extd \alpha_1 \extd \alpha_2  \int \extd g \vert X_g \vert^2 \, \extd h_1 \extd h_2 \, \nn \\
&& \hfill [K_{\alpha_1}(h_1)]^2 K_{\alpha_1 + \alpha_2} (h_1 h_2) K_{\alpha_2}(h_2) K_{\alpha_2}(h_2 g) \\
&=& \frac{1}{6}  \sqrt{4 \pi}^5 \left(\frac{1}{16 \pi^2}\right)^3 \int_{M^{-2i}}^{M^{-2(i-1)}} \extd \alpha_1 \extd \alpha_2 \, \int_{\mathbb{R}^3} \extd X \, \vert X \vert^2 \nn \\
&& \int_{\mathbb{R}^3} \extd H_1  \, \int_{\mathbb{R}^3} \extd H_2  \, \frac{\e^{-R(X, H_1, H_2) }}{{\alpha_1}^{3} {\alpha_2}^{3} (\alpha_1 + \alpha_2)^{3/2}} \,,
\ees
where 
\beq
R(X, H_1, H_2) = \frac{{H_1}^2}{2 \alpha_1} + \frac{\left( H_1 + H_2 \right)^2}{4 (\alpha_1 + \alpha_2)} + \frac{{H_2}^2}{4 \alpha_2} + \frac{\left( H_2 + X \right)^2}{4 \alpha_2} \,.
\eeq
One can then rewrite $R$ as
\bes
R(X , H_1 , H_2 ) &=& \frac{3 \alpha_1 + 2 \alpha_2}{4 \alpha_1 (\alpha_1 + \alpha_2)} \left( H_1 + \frac{\alpha_1}{3 \alpha_1 + 2 \alpha_2} H_2 \right)^2 \nn \\
&& + \frac{3}{2} \frac{\alpha_1 + \alpha_2}{(3 \alpha_1 + 2 \alpha_2) \alpha_2} \left( H_2 + \frac{1}{6} \frac{3 \alpha_1 + 2 \alpha_2}{\alpha_1 + \alpha_2} X \right)^2 \\
&& + \frac{1}{24} \frac{3 \alpha_1 + 4 \alpha_2}{(\alpha_1 + \alpha_2) \alpha_2} X^2\,, \nn
\ees
so that
\bes
\fg &=& \frac{1}{6}  (4 \pi)^{- 7/2}  \int_{1}^{M^{2}} \extd \alpha_1 \extd \alpha_2 \, \int_{\mathbb{R}^3} \extd X \, \vert X \vert^2 \, \frac{\e^{-\frac{3 \alpha_1 + 4 \alpha_2}{24 (\alpha_1 + \alpha_2) \alpha_2} X^2 }}{{\alpha_1}^{3} {\alpha_2}^{3} (\alpha_1 + \alpha_2)^{3/2}}  \\
&& \left( \frac{4 \pi \alpha_1 (\alpha_1 + \alpha_2)}{3 \alpha_1 + 2 \alpha_2} \right)^{3/2} \left( \frac{2 \pi (3\alpha_1 + 2\alpha_2) \alpha_2 }{3(\alpha_1 + \alpha_2)} \right)^{3/2} \nn \\
&=&  \frac{1}{2} \left(\frac{1}{3}\right)^{5/2} (4 \pi)^{-2} (2 \pi)^{3/2}  \int_{1}^{M^{2}} \extd \alpha_1 \extd \alpha_2 \, \frac{\extd \alpha_1 \extd \alpha_2}{{\alpha_1}^{3/2} {\alpha_2}^{3/2} (\alpha_1 + \alpha_2)^{3/2}} \nn \\
&& \int_{\mathbb{R}^3} \extd X \, \vert X \vert^2 \, \e^{-\frac{3 \alpha_1 + 4 \alpha_2}{24 (\alpha_1 + \alpha_2) \alpha_2} X^2 } \\
&=& 24 \pi \int_{1}^{M^{2}} \extd \alpha_1 \extd \alpha_2 \, \frac{(\alpha_1 + \alpha_2 ) \alpha_2}{{\alpha_1}^{3/2} (3 \alpha_1 + 4 \alpha_2)^{5/2}}\,.
\ees
Since $G_4^{\ell \ell'}$ and $G_4^{\ell' \ell}$ generate Laplacians acting on the same color, the combinatorial factor is simply $2$. Hence these graphs generate a term:
\beq
- 2 \, \fg \, u_{6,2,i}
\eeq
in the expression of $Z_{i-1}$.

This concludes the computation of the wave-function parameter:
\beq\label{result_wf}
\boxed{
Z_{i-1} = 1  - \fe \, u_{4,i} - \ff \, u_{6,1,i} - 2\, \left( \fa \, \fe + \fg \right) u_{6,2,i} + \cO(u^2)}
\eeq

\subsection{Resulting equations}

We can now gather the results (\ref{result_61}), (\ref{result_62}) and (\ref{result_wf}) to deduce the flow equations for $u_{6,1}$ and $u_{6,2}$ up to second order. This yields:

\noindent\fbox{
\begin{minipage}[c]{\linewidth}
\begin{align}
u_{6,1,i-1} &= u_{6,1,i} + \left( 3 \, \ff - 6 \, \fa \, \fc \right) {u_{6,1,i}}^2 + 6 \left( \fa \, \fe + \fg - \fa \, \fc - \fd \right) u_{6,1,i}  u_{6,2,i}  \nn \\
& \quad + 3 \left( \fe - \fc \right) u_{6,1,i} u_{4,i} + \cO(u^3)  \label{r_61} \\
u_{6,2,i-1} &= u_{6,2,i} + \left( 6 \left( \fa \, \fe + \fg \right) - 4 \, \fa \, \fc - 4 \, \fd \right) {u_{6,2,i}}^2   \nn \\
& \quad + \left( 3 \, \ff  - 4 \, \fa \, \fc  \right) u_{6,1,i}  u_{6,2,i} +  \left( 3 \, \fe - 2 \, \fc \right) u_{6,2,i} u_{4,i} + \cO(u^3)  \label{r_62}
\end{align}
\end{minipage}
}

\

\noindent We notice that all the wave-function contributions come with a positive sign, as opposed to all the other terms which come with a minus sign. In order to determine how the $6$-valent couplings behave in the UV region, one can first determine the signs of the resulting coefficients. To this purpose, we prove the following lemma:
\begin{lemma}
 For any $M>1$:
 \begin{enumerate}[(i)]
  \item $\fe > \fc$ ;
  \item $\ff > 2\, \fa \, \fc$.
  \end{enumerate}
 There exists $M_0 > 1$, such that for any $M> M_0$:
 \begin{enumerate}[(i)]
  \setcounter{enumi}{2}
  \item $\fg (M) > \fd (M)$ .
 \end{enumerate}
\end{lemma}
\begin{proof}
To begin with, we remark that
\bes\label{bound_m2}
M^2 + 1 &=& \left( M - 1 + 1 \right)^2 + 1 = \left( M - 1 \right)^2 + 2 \left( M - 1 \right) + 2 \\
&\leq& \left( M - 1 \right)^2 + 2 \sqrt{2} \left( M - 1 \right) + 2 = \left( M-1 + \sqrt{2}\right)^2\,. 
\ees
Hence 
\bes
\fc &=& 4 \sqrt{\pi} \left( \sqrt{2} \sqrt{M^2 + 1} - M - 1 \right) \\
%&\leq& 4 \sqrt{\pi} \left( \sqrt{2} \left( M - 1 + \sqrt{2} \right) - M - 1 \right) \\
&\leq&  4 \sqrt{\pi} \left( \sqrt{2} - 1 \right) \left( M - 1\right)\,. \label{ineq_c}
\ees

This immediately yields 
\bes
\fe - \fc &\geq& 3 \sqrt{\frac{\pi}{2}} \left( M - 1 \right) - 4 \sqrt{\pi} \left( \sqrt{2} - 1 \right) \left( M - 1\right) \\
&=& \left( 4 - \frac{5 \sqrt{2}}{2 } \right) \sqrt{\pi} \left( M - 1 \right) \\
%&>& \left( 4 - \frac{5 \times 3}{4 } \right) \sqrt{\pi} \left( M - 1 \right) = \frac{1}{4} \sqrt{\pi} \left( M - 1 \right)  \\
&>& 0\,,
\ees
where in the third line we have used that $\sqrt{2} < \frac{3}{2}$.

\

We can prove the second inequality in a similar way:
\bes
\ff - 2 \fa \fc &\geq& 6 \pi \left( M-1\right) \left(1 - \frac{1}{M} \right) \nn \\
&&- 2 \sqrt{2\pi} \left( 1 - \frac{1}{M} \right) 4 \sqrt{\pi} \left( \sqrt{2} - 1 \right) \left( M - 1\right)  \\
&=& \left( 8 \sqrt{2} - 10 \right) \pi \left( M-1\right) \left(1 - \frac{1}{M} \right) \\
&>& 0
\ees
with the last line an immediate consequence of $\sqrt{2} > \frac{5}{4}$.

\

In order to prove the third inequality it is sufficient to show that the derivative of $\fg - \fd$ with respect to $M$ converges to a strictly positive number when $M \to + \infty$. One has 
\bes
\frac{1}{\pi} \frac{\extd \fg}{\extd M} &=& 48  M \left( \int_{1}^{M^2} \extd \alpha_1 \frac{(\alpha_1 + M^2) M^2}{{\alpha_1}^{3/2} [3 \alpha_1 + 4 M^2]^{5/2}} \right.  \\
&& \left. + \int_{1}^{M^2} \extd \alpha_2  \frac{(M^2 + \alpha_2) \alpha_2}{M^3 [3 M^2 + 4 \alpha_2]^{5/2}} \right) \nn \\
%&=& 48 M^3 \int_{1}^{M^2} \extd \alpha_1  \frac{1}{{\alpha_1}^{1/2} [3 \alpha_1 + 4 M^2]^{5/2}}
%+ 48 M^5 \int_{1}^{M^2} \extd \alpha_1  \frac{1}{{\alpha_1}^{3/2} [3 \alpha_1 + 4 M^2]^{5/2}} \nn \\
%&&+ 48 \int_{1}^{M^2} \extd \alpha_2  \frac{ \alpha_2}{[3 M^2 + 4 \alpha_2]^{5/2}}
%+ \frac{48}{M^2} \int_{1}^{M^2} \extd \alpha_2  \frac{ {\alpha_2}^2}{[3 M^2 + 4 \alpha_2]^{5/2}} \\
&=& \frac{24 \left( 1 + M^2 \right)^2}{M} \frac{1}{\left( 3 + 4 M^2 \right)^{3/2}} - 24 \frac{\left( 1 + M^2 \right)^2}{M^2} \frac{1}{\left(4 + 3 M^2\right)^{3/2}} \nn \\
&\underset{M \gg 1}{\approx}& 3 - \frac{8}{3} \sqrt{3} \frac{1}{M} + \cO(\frac{1}{M^2}) \label{g_dvlt}
\ees
and
\bes
\frac{1}{\pi} \frac{\extd \fd}{\extd M} &=& 8  M \left( \int_{1}^{M^2} \extd \alpha_2 \extd \alpha_3 \frac{1}{M^3  [3 M^2 + 4 (\alpha_2 + \alpha_3)]^{3/2}} \right.  \\ 
&& \left. + 2 \int_{1}^{M^2} \extd \alpha_1 \extd \alpha_2 \frac{1}{ {\alpha_1}^{3/2}  [3 \alpha_1 + 4 (\alpha_2 + M^2)]^{3/2}} \right) \nn \\
%&=& \frac{8}{M^2} \int_{1}^{M^2} \extd \alpha_2 \extd \alpha_3 \frac{1}{[3 M^2 + 4 (\alpha_2 + \alpha_3)]^{3/2}} + 16 M \int_{1}^{M^2} \extd \alpha_1 \extd \alpha_2 \frac{1}{ {\alpha_1}^{3/2}  [3 \alpha_1 + 4 (\alpha_2 + M^2)]^{3/2}} \nn \\
&=& - \frac{2}{M^2} \sqrt{8 + 3 M^2} - \frac{2}{M} \sqrt{3 + 8 M^2}  \\
&& + \frac{4}{M^2 (1 + M^2 )} \sqrt{4 + 7 M^2} + \frac{4 M}{1 + M^2 } \sqrt{7 + 4 M^2}  \nn \\
&\underset{M \gg 1}{\approx}& 8 - 4 \sqrt{2} + \left( 4 \sqrt{7} - 2 \sqrt{3} \right) \frac{1}{M} + \cO(\frac{1}{M^2}) \,. \label{d_dvlt}
\ees
This yields in particular: 
\beq
\frac{\extd (\fg- \fd)}{\extd M} (M) \underset{M \to + \infty}{\longrightarrow} \left( 4 \sqrt{2} - 5 \right)\pi > 0\,,
\eeq
which concludes the proof.
\end{proof}
\noindent{\bf{Remark.}} The explicit expressions of $\fg$ and $\fd$ being rather involved, we did not try to prove that $\fg > \fd$ in full generality. However one can easily check numerically that this holds, as one would expect. See Figure \ref{Mpartial_g-d}. 

\begin{figure}[h]
	\begin{center}
  \includegraphics[scale=0.7]{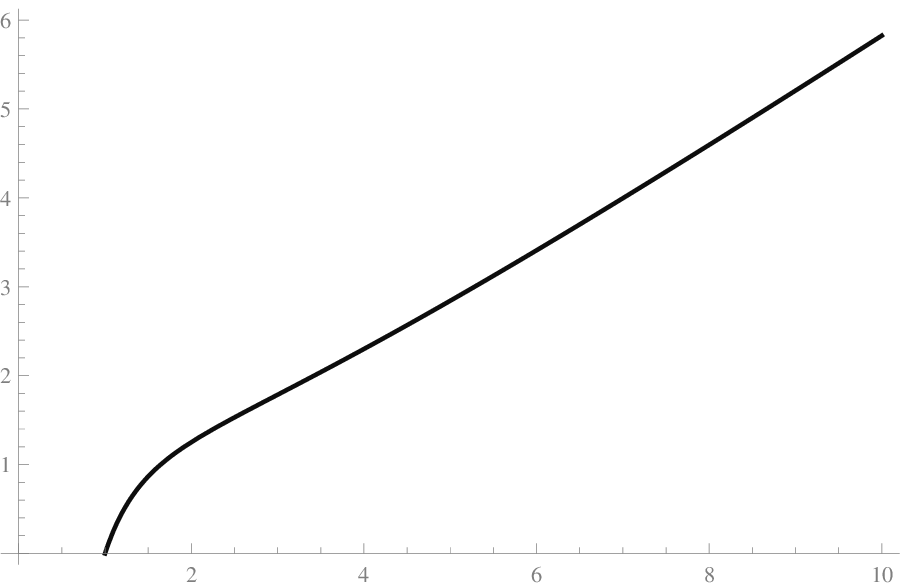}   	       
  \caption{Plot of $\frac{M}{\pi} \frac{\extd (\fg- \fd)}{\extd M}$ as a function of $M$.}\label{Mpartial_g-d}
  \end{center}
\end{figure}

\
 
These inequalities immediately imply that all the coefficients appearing in equations (\ref{r_61}) and (\ref{r_62}) are strictly positive. If one starts with small positive coupling constants at scale $i_0$, $u_{6,1}$ and $u_{6,2}$ will both grow towards the infrared. One would therefore be tempted to conclude that the model is asymptotically free. However, due to the non-linear nature of the flow equations we have derived, only a careful analysis can settle this question, which is the purpose of the next and final section.

\section{Qualitative behaviour in the vicinity of the Gaussian fixed point}\label{sec:portrait}

The main purpose of this last section is to determine under which conditions the theory flows towards a Gaussian model in the deep UV. To this effect, we will analyze the qualitative properties of the continuous flow which underlies the discrete equations derived in the preceding section.  

\subsection{Continuous flow}

The truncated system of equations we have derived in the previous sections takes the form:
\beq
\left\{
\begin{aligned}
u_{2,i-1} - u_{2,i} &= M^2 \left( \frac{M^2 - 1}{M^2} u_{2,i} + A_4 u_{4,i} + A_{6,1} u_{6,1,i} + A_{6,2} u_{6,2,i} \right) \\
u_{4,i-1} - u_{4,i} &= M \left( \frac{M - 1}{M} u_{4,i} + B_{6,1} u_{6,1,i} + B_{6,2} u_{6,2,i} \right) \\ %\label{du4} \\
u_{6,1,i-1} - u_{6,1,i} &= u_{6,1,i} \left( C_4 u_{4,i} + C_{6,1} u_{6,1,i} + D_{6,2} u_{6,2,i} \right) \\ %\label{du61} \\
u_{6,2,i-1} - u_{6,2,i} &= u_{6,2,i} \left(D_4 u_{4,i} + D_{6,1} u_{6,1,i} + D_{6,2} u_{6,2,i} \right)  %\label{du62}
\end{aligned}
\right. \nn
\eeq
Since the last three equations form a closed system, and we only wish to determine at which conditions the theory becomes Gaussian in the UV, we will ignore the mass coupling constant. We can moreover assume that the qualitative behaviour of the theory will be well-captured by a continuous version of the flow, which we now describe. 

\

Let us introduce new effective coupling constants $u_{4}(t)$, $u_{6,1}(t)$ and $u_{6,2}(t)$, where $t \in \mathbb{R}$ is to be thought of as $(- \log_M (\Lambda))$ and $\Lambda$ is a continuously varying UV cut-off\footnote{Contrary to standard notations, the UV side of the scale ladder corresponds to $t \to - \infty$.}. We then define the autonomous system of first-order differential equations:
\begin{equation}
(S) \quad \left\{
\begin{aligned}
%\dot{u}_2 &= M^2 \left( \frac{M^2 - 1}{M^2} u_{2} + A_4 u_{4} + A_{6,1} u_{6,1} + A_{6,2} u_{6,2} \right) %\equiv f_2 (u_b) 
%\\
\dot{u}_4 &= M \left( \frac{M - 1}{M} u_{4} + B_{6,1} u_{6,1} + B_{6,2} u_{6,2} \right)  %\equiv f_4 (u_b)
\\
\dot{u}_{6,1} &= u_{6,1} \left( C_4 u_{4} + C_{6,1} u_{6,1} + D_{6,2} u_{6,2} \right) %\equiv f_{6,1} (u_b)
\\
\dot{u}_{6,2} &= u_{6,2} \left(D_4 u_{4} + D_{6,1} u_{6,1} + D_{6,2} u_{6,2} \right) %\equiv f_{6,2} (u_b)
\end{aligned}
\right.
\end{equation}
From now on, the term 'Gaussian fixed point' will refer to the fixed point $u_4 = u_{6,1} = u_{6,2} = 0$ of the system $(S)$. Our goal is to determine which solutions of $(S)$ converge to the Gaussian fixed point in the limit $t \to - \infty$.

\

Let us now gather a few useful facts about $(S)$. First, the coefficients $B_{6,1}, \ldots,D_{6,2}$ are all strictly positive. Second, five different planes will play a special role, as subspaces on which one of the components of the flow cancels. These are:
\begin{align}
H_4 : &\qquad 0 = u_4 + 2 \sqrt{2\pi} u_{6,1} + 4 \sqrt{2 \pi} u_{6,2} \equiv f_4 (u_4 , u_{6,1}, u_{6,2} ) \\
H_{6,1}^{(1)} : &\qquad 0 = C_4 u_4 + C_{6,1} u_{6,1} + C_{6,2} u_{6,2} \equiv f_{6,1} (u_4 , u_{6,1}, u_{6,2} ) \\
H_{6,1}^{(2)} : &\qquad 0 = u_{6,1}  \\
H_{6,2}^{(1)} : &\qquad 0 = D_4 u_4 + D_{6,1} u_{6,1} + D_{6,2} u_{6,2} \equiv f_{6,2} (u_4 , u_{6,1}, u_{6,2} )\\
H_{6,2}^{(2)} : &\qquad 0 = u_{6,2} 
\end{align}
In particular, we immediately notice that on $H_{6,1}^{(2)}$ (resp. $H_{6,2}^{(2)}$), both $u_{6,1}= 0$ (resp. $u_{6,2} = 0$) and $\dot{u}_{6,1}= 0$ (resp. $\dot{u}_{6,2} = 0$). With the help of Cauchy--Lipschitz theorem, this immediately implies the following lemma:
\begin{lemma}
The sign of $u_{6,1}$ (resp. $u_{6,2}$) is invariant under the flow of $(S)$. 
\end{lemma}
Thus we will conveniently organize our analysis according to the signs of $u_{6,1}$ and $u_{6,2}$, and as a warming exercise construct the phase portrait of the flow in the invariant subspace $H_{6,2}^{(2)}$. The following set of inequalities will prove useful in this respect.
\begin{lemma}\label{positivebeta}
 For any $M>1$:
 \begin{enumerate}[(i)]
  \item $\beta_{11}(M) \equiv 2\sqrt{2 \pi} C_4 (M) - C_{61} (M) > 0$ ;
  \item $\beta_{21}(M) \equiv 2\sqrt{2 \pi} D_4 (M) - D_{61} (M) > 0$.
  \end{enumerate}
 There exists $M_0 > 1$, such that for any $M> M_0$:
 \begin{enumerate}[(i)]
  \setcounter{enumi}{2}
  \item $\beta_{12}(M) \equiv 4\sqrt{2 \pi} C_4 (M) - C_{62} (M) > 0$ ;
  \item $\beta_{22}(M) \equiv 4\sqrt{2 \pi} D_4 (M) - D_{62} (M) > 0$.
 \end{enumerate}
\end{lemma}
\begin{proof}
One first has that:
\bes
\beta_{11} &=& -3 \, \ff + 6 \sqrt{2 \pi}\, \fe - \frac{6 \sqrt{2 \pi}}{M}\, \fc \\ 
%&=& \frac{18 \pi}{M} \left( M - 1 \right) - \frac{6 \sqrt{2} \pi}{M} \left( \sqrt{2} \sqrt{M^2 + 1} - M - 1\right)\,,
\ees
which by use of formula (\ref{ineq_c}) implies
\beq
\beta_{11} \geq \frac{6 (1 + \sqrt{2}) \pi}{M} \left( M - 1 \right) > 0 \,.  
\eeq 
The sign of $\beta_{21}$ is immediately obtained from that of $\beta_{11}$: 
\beq
\beta_{21} = -3 \, \ff + 6 \sqrt{2 \pi}\, \fe - \frac{4 \sqrt{2 \pi}}{M}\, \fc \geq \beta_{11} > 0\,.
\eeq
The last two inequalities are again more involved, due to the complicated expressions for $\fg$ and $\fd$. Using the explicit expression of $\fa$, $\fe$, and applying once more formula (\ref{ineq_c}), one finds:
\bes
\beta_{12} &\geq& \left( 24 \sqrt{2} - 30 \right) \pi \left( 1 + \frac{1}{M}\right) \left( M -1 \right) - 6 \left( \fg - \fd \right) \equiv \chi_{12} (M) \,, \\
\beta_{22} &\geq& \left( 16 \sqrt{2} - 14 \right) \pi \left( M -1 \right) - 6 \fg + 4 \fd \equiv \chi_{22} (M) \,.
\ees
The functions $\chi_{12}$ and $\chi_{22}$ evaluate to $0$ when $M=1$, so to prove that they are positive, it is enough to show that their first derivatives are positive for any $M>1$. These are given by
\bes
\frac{\extd \chi_{12}}{\extd M} (M) &=& \left( 24 \sqrt{2} - 30 \right) \pi \left( 1 + \frac{1}{M^2} \right) - 6 \frac{\extd [\fg - \fd]}{\extd M} (M) \,,\\
\frac{\extd \chi_{22}}{\extd M} (M) &=& \left( 16 \sqrt{2} - 14 \right) \pi \left( 1 + \frac{1}{M^2} \right) -  \frac{\extd [6\fg - 4 \fd]}{\extd M} (M) \,.
\ees
By use of the asymptotic expansions of $\fg$ and $\fd$, see equations (\ref{g_dvlt}) and (\ref{d_dvlt}), one obtains:
\bes
\frac{\extd \chi_{12}}{\extd M} (M) &\underset{M \to + \infty}{\approx}&  \left(4 \sqrt{3} + 24 \sqrt{7}\right) \pi \frac{1}{M} + \cO(\frac{1}{M^2})\,, \\
\frac{\extd \chi_{22}}{\extd M} (M) &\underset{M \to + \infty}{\approx}&  \left(\frac{40}{3} \sqrt{3} + 16 \sqrt{7}\right) \pi \frac{1}{M} + \cO(\frac{1}{M^2}) \,.
\ees
In particular, one concludes that both $\chi_{12} (M)$ and $\chi_{22} (M)$ diverge to $+ \infty$ when $M \to + \infty$, which concludes the proof.
\end{proof}
\noindent{\bf Remark.} Again, even if we have not proven it in full generality, one can convince oneself that the functions $\chi_{12}$ and $\chi_{22}$ are strictly positive on $\left] 1, + \infty \right[$ (see Figure \ref{chi_plots}), which guarantees that the coefficients $\beta_{12}$ and $\beta_{22}$ are strictly positive for any value of $M>1$.
We also provide numerical plots of the $\beta$-coefficients in Figure \ref{beta_plots}, which confirm our claims.

\begin{figure}[h]
  \centering
  \subfloat[$\frac{\extd \chi_{12}}{\extd M}$]{\label{chi_12}\includegraphics[scale=0.6]{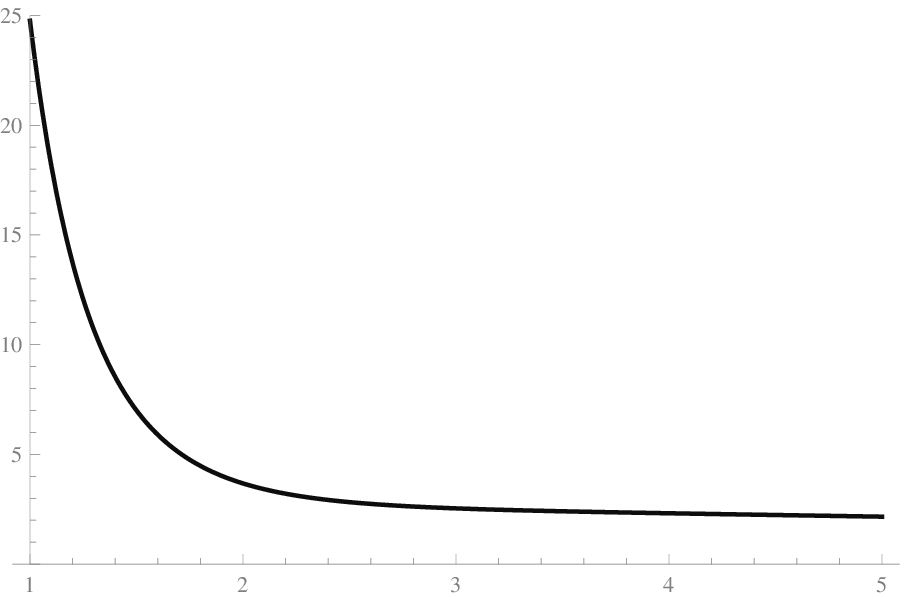}}
  \subfloat[$\frac{\extd \chi_{22}}{\extd M}$]{\label{chi_22}\includegraphics[scale=0.6]{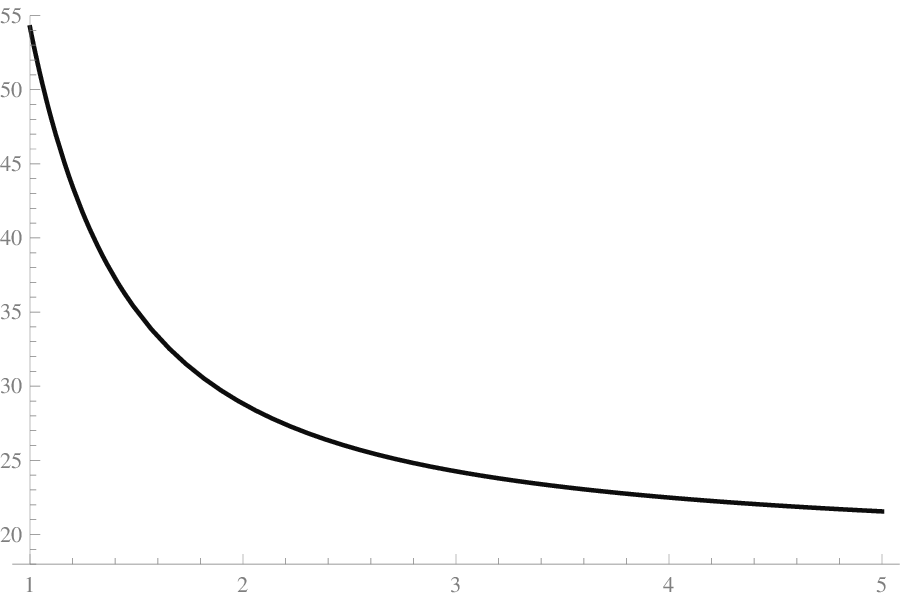}}     
  \caption{Plots of $\chi_{12}'$ and $\chi_{22}'$ as functions of $M$: both are manifestly positive.}\label{chi_plots}
\end{figure}

\begin{figure}[h]
  \centering
  \subfloat[$\beta_{11}$]{\label{beta_11}\includegraphics[scale=0.6]{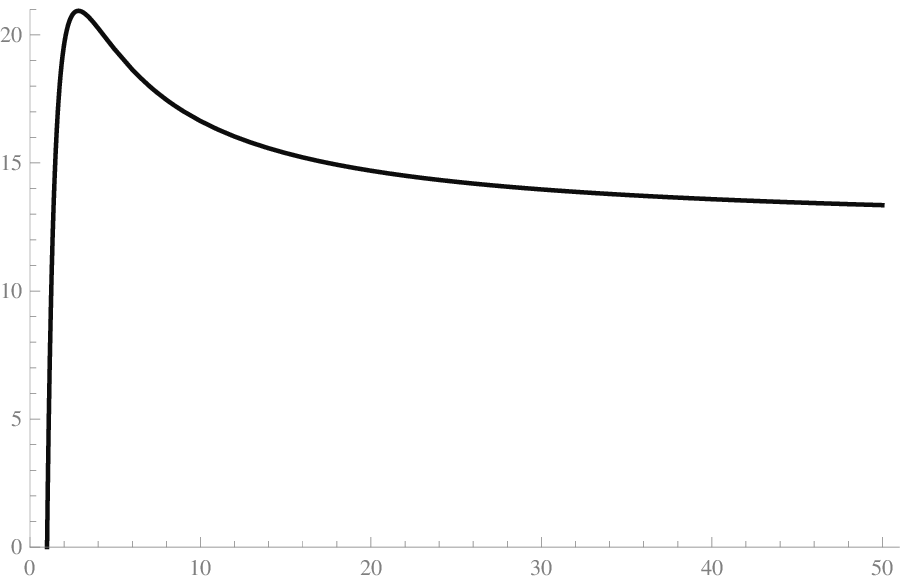}}
  \subfloat[$\beta_{12}$]{\label{beta_12}\includegraphics[scale=0.6]{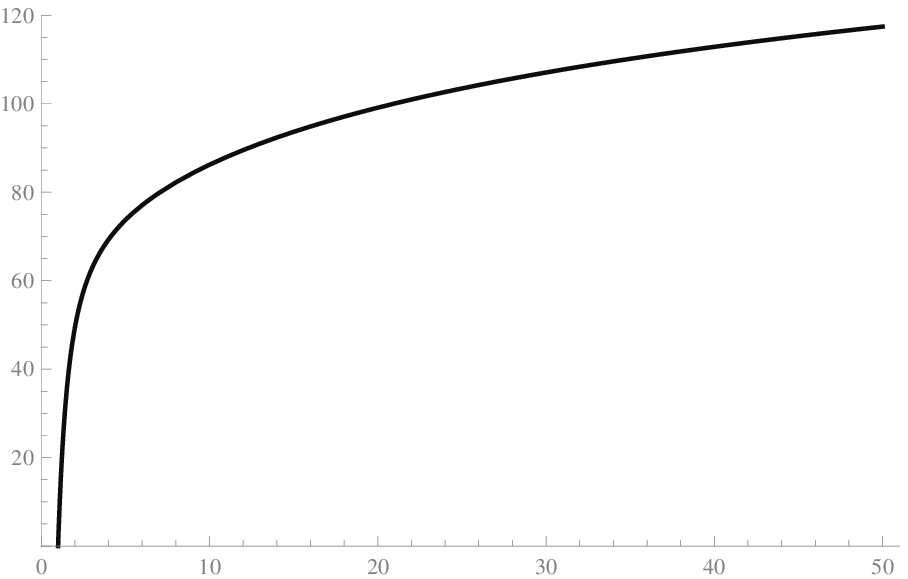}} \\    
  \subfloat[$\beta_{21}$]{\label{beta_21}\includegraphics[scale=0.6]{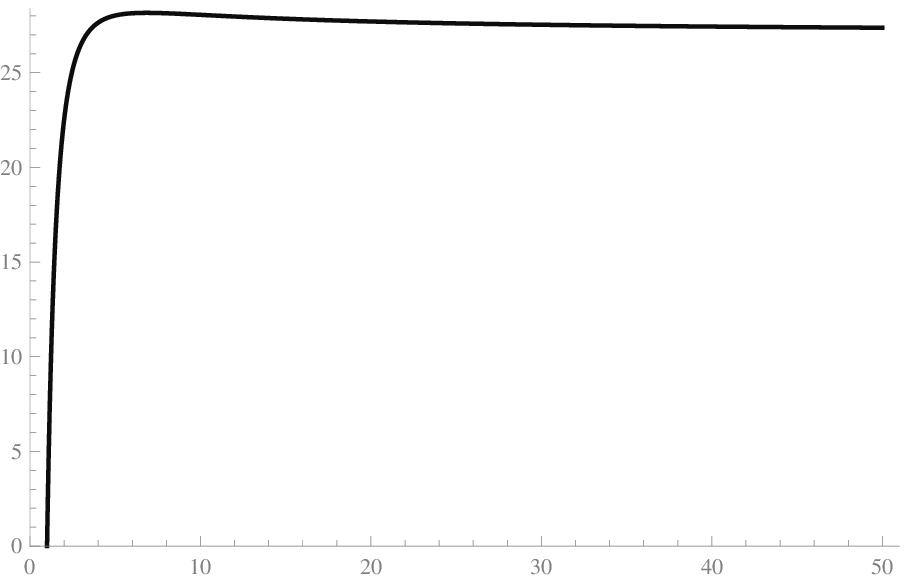}}
  \subfloat[$\beta_{22}$]{\label{beta_22}\includegraphics[scale=0.6]{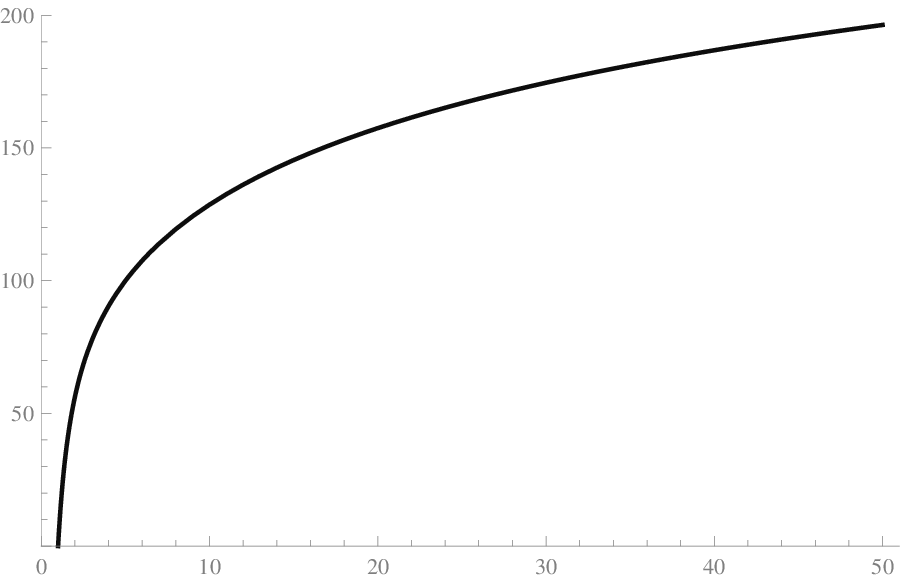}}
  \caption{From left to right: $\beta_{11}$, $\beta_{12}$, $\beta_{21}$ and $\beta_{22}$ as functions of $M>1$.}\label{beta_plots}
\end{figure}

\subsection{Phase portrait on the $u_{6,2} = 0$ invariant subspace}

Let us assume that $u_{6,2} = 0$ and look at the reduced autonomous system:
\begin{equation}
(S') \quad \left\{
\begin{aligned}
%\dot{u}_2 &= M^2 \left( \frac{M^2 - 1}{M^2} u_{2} + A_4 u_{4} + A_{6,1} u_{6,1} + A_{6,2} u_{6,2} \right) %\equiv f_2 (u_b) 
%\\
\dot{u}_4 &= (M-1) \left( u_4 + 2 \sqrt{2\pi} u_{6,1} \right)  %\equiv f_4 (u_b)
\\
\dot{u}_{6,1} &= u_{6,1} \left( C_4 u_{4} + C_{6,1} u_{6,1} \right) %\equiv f_{6,1} (u_b)
\end{aligned}
\right.
\end{equation}
According to Lemma \ref{positivebeta}, the relative positions of the lines $L_4 \equiv H_{4} \cap H_{6,2}^{(2)}$ and $L_{6,1}^{(1)} \equiv H_{6,1}^{(1)} \cap H_{6,2}^{(2)}$ in the plane $(u_4 , u_{6,1})$ are as represented in Figure \ref{phase_portrait2d}, no matter what the value of $M$ is. We define the open sets:
\begin{align}
 U_{I} &= \left\{ (u_4 , u_{6,1} ) \vert  u_{6,1} > 0, \, f_{6,1} (u_4 , u_{6,1}, 0)>0 \right\} \\
 U_{II} &= \left\{ (u_4 , u_{6,1} ) \vert f_{6,1} (u_4 , u_{6,1}, 0)<0 ,\,f_{4} (u_4 , u_{6,1}, 0)>0 \right\} \\
 U_{III} &= \left\{ (u_4 , u_{6,1} ) \vert u_{6,1} > 0 ,\, f_4 (u_4 , u_{6,1}, 0)<0 \right\} \\
 U_{IV} &= \left\{ (u_4 , u_{6,1} ) \vert u_{6,1} < 0 ,\, f_4 (u_4 , u_{6,1}, 0)>0  \right\} \\
 U_{V} &= \left\{ (u_4 , u_{6,1} ) \vert f_{6,1} (u_4 , u_{6,1}, 0)>0 ,\,f_{4} (u_4 , u_{6,1}, 0)<0 \right\} \\
 U_{VI} &= \left\{ (u_4 , u_{6,1} ) \vert u_{6,1} < 0,\, f_{6,1} (u_4 , u_{6,1}, 0)<0 \right\} 
\end{align}
in each of which the signs of $\dot{u}_4$ and $\dot{u}_{6,1}$ do not change. In Figure \ref{phase_portrait2d}, the small arrows in each region indicate the signs of $\dot{u}_4$ and $\dot{u}_{6,1}$. In particular, one notices that the only two regions in which $\vert u_4 \vert$ and $\vert u_{6,1} \vert$ are both strictly increasing are $U_I$ and $U_{V}$. We therefore expect that any non-trivial trajectory\footnote{That is, any trajectory not contained in the $u_{6,1} = 0$ axis.} converging to the Gaussian fixed point in the $t \to - \infty$ limit would approach it from one of these two regions. Let us now investigate in more details whether one of these scenarii actually occurs.

\

In the following, we denote by $\bu(t) = (u_{4} (t), u_{6,1} (t))$ a maximal solution of $(S')$, defined for any $t \in \mathbb{R}$. Let us first focus on the case $u_{6,1}>0$. One can first prove the following lemma.
\begin{lemma}\label{lem_2d1}
 Let $t_1 \in \mathbb{R}$. 
 If $\bu(t_1) \in U_{I} \cup U_{III}$, then one can find $t_0 < t_1$ such that $\bu(t_0) \in U_{II}$.
\end{lemma}
\begin{proof}
Let us first assume that $\bu(t_1) \in U_{III}$. 
If $\bu(t) \in \overline{U_{III}}$ for all $t \leq t_1$, then $u_{4} (t)$ and $u_{6,1} (t)$ are both decreasing with respect to $t$. Moreover, $\dot{u}_{6,1} (t) \leq \dot{u}_{6,1} (t_1) < 0$, and hence $u_{6,1}$ diverges in the limit $t \to - \infty$. Since $u_{4} (t) > u_{4} (t_1)$ for all $t \leq t_1$, $\bu (t)$ necessarily crosses $L_4$, which yields a contradiction. Hence there must exists $t_0 < t_1$ such that $\bu(t_0) \in U_{II}$. 

%If $\bu(t_1 - \tau) \in \overline{U_{III}}$ for all $\tau \geq 0$, then $u_{4} (t_1 - \tau)$ and $u_{6,1} (t_1 - \tau)$ are both increasing with respect to $\tau$. Moreover, $\frac{\extd u_{6,1} (t_1 - \tau)}{\extd \tau} \geq \restr{\frac{\extd u_{6,1} (t_1 - \tau)}{\extd \tau}}{\tau = 0} > 0$, and hence $u_{6,1}$ diverges in the limit $\tau \to + \infty$. Since $u_{4} (t_1 - \tau) > u_{4} (t_1)$ for all $\tau \geq 0$, $\bu (t_1 - \tau)$ necessarily crosses $L_4$ in a finite time, which yields a contradiction. Hence there must exists $t_0 < t_1$ such that $\bu(t_0) \in U_{II}$. 

The non-trivial part of this lemma concerns the situation in which $\bu(t_1) \in U_I$. If $\bu(t) \in \overline{U_{I}}$ for all $t \leq t_1$, then $u_{4}$ and $u_{6,1}$ are both monotonically increasing and bounded from below. Therefore they must converge in $- \infty$, and since there is no other fixed point available, $\bu(t)$ converges to $0$ in the limit $t \to - \infty$. Now, one also has:
\beq
\frac{\dot{u}_{6,1}}{u_{6,1}} =  C_4 u_{4} + C_{6,1} u_{6,1} \leq C_4 u_4 + 2 \sqrt{2 \pi} C_4 u_{6,1}\leq K \dot{u}_{4}\,,
\eeq
where $K= \frac{C_4}{M-1} >0$ and we made use of Lemma \ref{positivebeta}. This implies that
\begin{align}
\forall t \in \left]- \infty , t_1 \right]\,, \qquad u_{6,1} (t) &\geq u_{6,1} (t_1) \exp\left( K (u_{4}(t) - u_{4}(t_1) \right)\,, 
\end{align}
which is incompatible with $\underset{t \to - \infty}{\lim} \bu(t) =0$. Hence there must be a $t_0 < t_1$ such that $\bu(t_0) \in U_{II}$.
%
%The non-trivial part of this lemma concerns the situation in which $\bu(t_1) \in U_I$. If $\bu(t_1 - \tau) \in \overline{U_{I}}$ for all $\tau \geq 0$, then $u_{4} (t_1 - \tau)$ and $u_{6,1} (t_1 - \tau)$ are both decreasing with respect to $\tau$ and bounded from below. Therefore they must converge, and since there is no other fixed point available, $\bu(t)$ converges to $0$ in the limit $t \to - \infty$. Now, since $\dot{u}_4 > 0$ in $U_I$, we can parameterize the trajectory by $u_4$ itself. For $u_4 \in \left]0 , u_4(t_1)\right]$, one then has (with a slight abuse of notation):
%\beq
%\frac{\extd u_{6,1}}{\extd u_4} = \frac{u_{6,1}}{M-1} \frac{C_4 u_{4} + C_{6,1} u_{6,1}}{u_4 + 2 \sqrt{2\pi} u_{6,1}} \leq K u_{6,1}\,,
%\eeq
%where $K= \frac{C_4}{M-1} >0$ and we made use of Lemma \ref{positivebeta}. This implies that
%\begin{align}
%\forall u_4 \in \left]0 , u_4(t_1)\right]\,, \qquad u_{6,1} (u_4) &\geq u_{6,1} (t_1) \exp\left( K (u_{4} - u_{4}(t_1) \right) \\
%&\geq u_{6,1} (t_1) \exp\left( - K u_{4}(t_1) \right) > 0\,, 
%\end{align}
%which contradicts $\underset{u_4 \to 0}{\lim} u_{6,1} (u_4) =0$. Hence there must be a $t_0 < t_1$ such that $\bu(t_0) \in U_{II}$.
\end{proof}

This is sufficient to conclude with the following proposition.
\begin{proposition}
Let $t_0 \in \mathbb{R}$. There exists no solution $\bf{u}$ of $(S')$ such that:
\beq
u_{6,1}(t_0) > 0 \qquad {\rm{and}} \qquad \underset{t \to - \infty}{\lim} \bu(t) = 0 \,.
\eeq
\end{proposition}
\begin{proof}
The flow of $(S')$ is outward-pointing on the boundary of $U_{II}$. Therefore $U_{II}$ is stable with respect to the time-reversed flow of $(S')$. Combined with the fact that in this region both $\vert u_4 \vert$ and $\vert u_{6,1} \vert$ are increasing with respect to this time-reversed flow, the previous Lemma provides the desired result. 
\end{proof}

\noindent{\bf{Remark.}} We can easily go a bit further and show that the trajectories in the half-space $\left\{ u_{6,1} > 0 \right\}$ are of two types: those verifying  $\underset{t \to + \infty}{\lim} u_{4}(t) = - \infty$ and $\underset{t \to + \infty}{\lim} u_{6,1}(t) = 0$; and those such that $\underset{t \to + \infty}{\lim} u_{4}(t) = \underset{t \to + \infty}{\lim} u_{6,1}(t) = + \infty$. In both cases, $\underset{t \to + \infty}{\lim} u_{4}(t) = - \infty$ and $\underset{t \to + \infty}{\lim} u_{6,1}(t) = + \infty$. See Figure \ref{phase_portrait2d}.

 \begin{figure}[h]
  \centering
  \includegraphics[scale=0.8]{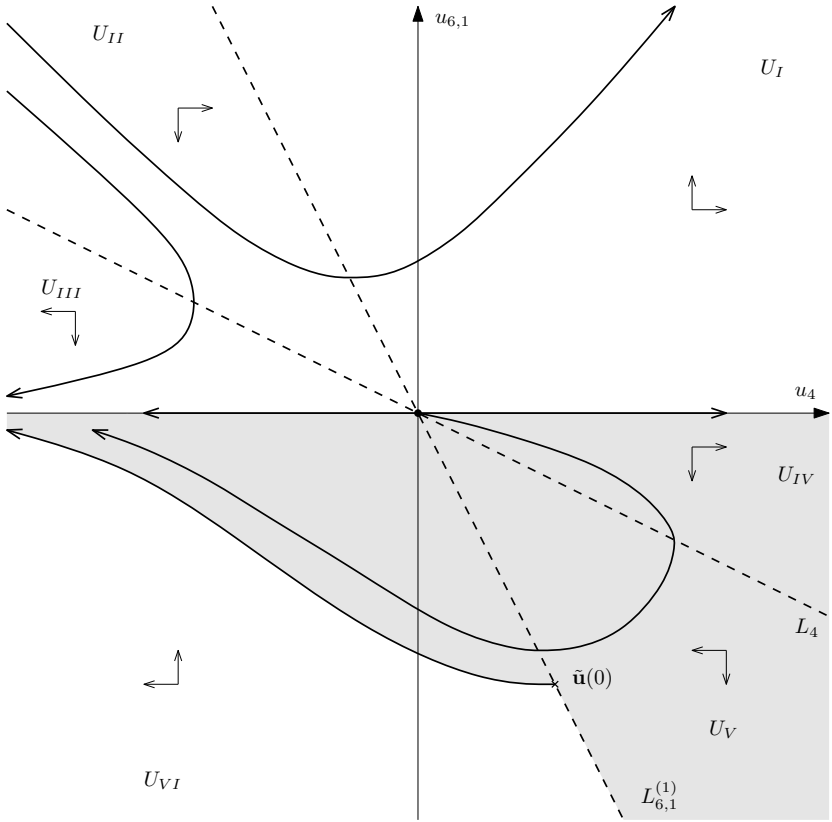}     
  \caption{Phase portrait in the $(u_{4}, u_{6,1})$ plane.}\label{phase_portrait2d}
\end{figure}

\

Let us now consider the other half of the $(u_4 , u_{6,1})$ plane. An analogue of Lemma \ref{lem_2d1} holds in this sector.
\begin{lemma}\label{lem_2d2}
 Let $t_1 \in \mathbb{R}$. 
 If $\bu(t_1) \in U_{V}$, then there exists $t_0 , t_2 \in \mathbb{R}$ such that:
 \beq
 \left\{
 \begin{aligned}
 &t_0 < t_1 < t_2\\ 
 &\bu(t_0) \in U_{IV} \\
 &\bu(t_2) \in U_{VI} 
 \end{aligned}
 \right.
 \eeq
Moreover:
\beq
\underset{t \to + \infty}{\lim} u_{4}(t) = - \infty\,.
\eeq
\end{lemma}
\begin{proof}
Assume $\bu(t_1) \in U_{V}$. 
 
 If $\bu(t) \in \overline{U_{V}}$ for all $t \leq t_1$, then $u_4$ and $u_{6,1}$ are monotonically decreasing on $\left]- \infty, t_1 \right]$, and bounded from above. $\bu(t)$ must therefore converge in the limit $t \to - \infty$, hence to the Gaussian fixed point. This is impossible since $u_4(t) \geq u_{4}(t_1)$ for all $t \leq t_1$, therefore there must be a $t_0 < t_1$ such that $\bu(t_0) \in U_{IV}$.

 %If $\bu(t) \in \overline{U_{V}}$ for all $t \leq t_1$, then $u_{4}(t) \geq u_{4}(t_1)$. Moreover, since $\dot{u}_{6,1} < 0$ in this region and $u_{6,1} < 0$, $u_{6,1}(t)$ must converge in the limit $t \to - \infty$. But the only possible limit is $u_{6,1} = 0$, which given that $u_{4}(t_1)>0$ cannot be reached within $U_{V}$. Hence there must be a $t_0 < t_1$ such that $\bu(t_0) \in U_{IV}$.
 
 Similarly, we can show that $\bu(t_2)$ has to be in region $U_{VI}$ for some $t_2 > t_1$. This region is stable under the flow of $(S)$. Therefore $u_{6,1}$ is increasing, bounded on $\left[t_2, + \infty \right[$, and must therefore converge in $+ \infty$. 
$u_4$ is decreasing, and similarly either converges or tends to $- \infty$ in $+ \infty$. In order to determine which of these two cases arises, notice that in the region $U_{VI}$ Lemma \ref{positivebeta} ensures that:
\begin{align}
\frac{\dot{u}_{6,1}}{u_{6,1}} = C_4 u_4 + C_{6,1} u_{6,1} \geq C_4 u_4 + 2 \sqrt{2 \pi} C_4 u_{6,1} \geq K \dot{u}_4\,,
\end{align}
where $K = \frac{C_4}{M-1} > 0$. This yields the bound
\beq
u_{6,1} (t) \leq u_{6,1} (t_2) \exp\left( K (u_4 (t) - u_4 (t_2) \right)\,,
\eeq
for any $t \geq t_2$. It therefore appears that if $u_4$ were to converge to a finite value in $+ \infty$, $u_{6,1}$ would converge to a non-zero value, and the limit value of $\bu(t)$ would not be a fixed point. Hence $u_4(t)$ must tend to $- \infty$ when $t \to + \infty$.

%Let us assume that the first case is realized. Then we can parameterize the trajectory between $t_2'$ and $+ \infty$  by $u_4$ itself. Lemma \ref{positivebeta} ensures that
% 
% {\bf{absurde pas nécessaire?}}
% Similarly, we can show that $\bu(t_2')$ has to be in region $U_{VI}$ for some $t_2' > t_1$. This region is stable under the flow of $(S)$. Therefore $u_{6,1}$ is increasing, bounded on $\left[t_2', + \infty \right[$, and must therefore converge to $0$ in $+ \infty$. 
%$u_4$ is decreasing, and similarly either converges to $0$ or becomes strictly negative for some $t_2 > t_2'$. Let us assume that the first case is realized. Then we can parameterize the trajectory between $t_2'$ and $+ \infty$  by $u_4$ itself. Lemma \ref{positivebeta} ensures that {\bf CHECK}:
% \beq
%\frac{\extd u_{6,1}}{\extd u_4} = \frac{u_{6,1}}{M-1} \frac{C_4 u_{4} + C_{6,1} u_{6,1}}{u_4 + 2 \sqrt{2\pi} u_{6,1}} \geq K u_{6,1}\,,
%\eeq
%with $K = \frac{C_{6,1}}{2 \sqrt{2 \pi} (M-1)}>0$. Hence
%\begin{align}
%\forall u_4 \in \left]0 , u_4(t_2')\right]\,, \qquad u_{6,1} (u_4) &\leq u_{6,1} (t_2') \exp\left( K (u_{4} - u_{4}(t_2') \right) \\
%&\leq u_{6,1} (t_2') \exp\left( - K u_{4}(t_2') \right) < 0\,, 
%\end{align}
%which contradicts $\underset{t \to + \infty}{\lim} u_{6,1}(t) =0$.
\end{proof}
This allows to conclude that the Gaussian fixed point is UV attractive relative to negative perturbations of $u_{6,1}$.
\begin{proposition}
Let $\bu$ be a solution of $(S')$ and $t_2 \in \mathbb{R}$. There exists a neighborhood $V$ of $0$ such that: 
\beq
\left\{\begin{aligned}
u_{6,1}(t_0) &\leq 0 \\
\bu(t_0) &\in V
\end{aligned}\right. \quad \Rightarrow \quad \underset{t \to - \infty}{\lim} \bu(t) =0\,.
\eeq
\end{proposition}
\begin{proof}
In order to construct an appropriate neighborhood $V$, let us consider a solution $\tilde{\bu}$ such that $\tilde{\bu}(0) \in L_{6,1}$ and $\tilde{u}_{6,1}(0) < 0$. By Lemma \ref{lem_2d2}, $\tilde{\bu}(t)$ enters in region $U_{V}$ and remains there at later times. Hence the distance between the origin and the curve $\tilde{\bu}(\left[ 0 , + \infty\right[)$ is non-zero. Reparameterizing the trajectory of $\tilde{\bu}$ by the function $\tilde{u}_{6,1}(\tilde{u}_{4})$, we define the open set:
\begin{align}
V &= \{ (u_4 , u_{6,1})\vert  u_4 \in \left] - \infty,  \tilde{u}_{4} (0) \right] , \, u_{6,1} > \tilde{u}_{6,1} (u_4) \} \\
&\qquad \cup \{ (u_4 , u_{6,1})\vert f_{6,1} ( u_4 ,  u_{6,1} , u_{6,2} ) > 0 \}\,, \nn 
\end{align}
which contains the origin. In the half-plane $u_{6,1} < 0$, $V$ coincides with the greyed region in Figure \ref{phase_portrait2d}. 

It is easy to check that $V \cap \{ u_{6,1}<0 \}$ is stable with respect to the time-reversed flow of $(S)$. Indeed, this flow evaluated on the boundary of $V$ is either tangential (on $\tilde{\bu}(\left[ 0 , + \infty\right[)$) or inward-pointing (on $L_{6,1}^{(1)}$). 

Consider now a solution $\bu$ such that $u_{6,1}(t_0) \leq 0$ and $\bu(t_0) \in V$. If $\bu(t_0) \in U_{VI}$, one can find $t_1 < t_0$ such that $\bu(t_1) \in U_{V}$. If not, $u_{4,1}$ and $u_{6,1}$ being monotonic on $\left[ - \infty, t_0 \right[$ and bounded, they would have to converge in $- \infty$, hence to $0$. But this cannot be since one would also have $u_{6,1}(t) \leq u_{6,1}(t_0) <0$ for any $t \leq t_0$.
%$u_{6,1}$, being monotonically increasing on $\left[ - \infty, t_0 \right[$ and bounded, would have to converge in $- \infty$, hence to $0$, which cannot be.  
This observation, together with Lemma \ref{lem_2d2} implies the existence of some $t_2 < t_0$ such that $\bu(t_2) \in U_{IV}$, and hence $\bu(t) \in U_{IV}$ for all $t \leq t_2$. On $\left] - \infty, t_2\right]$, $u_4$ is increasing and bounded from below by $0$, therefore converges in $- \infty$. Similarly, $u_{6,1}$ converges in  $- \infty$. Since the only fixed point available is the Gaussian one, $\bu(t)$ converges to $0$ when $t\to - \infty$.
\end{proof}

\noindent{\bf{Remark.}}  Like in the previous case, one could try to map the flow more precisely in this region. However this would require more involved computations. Since from a physical perspective we are only interested in the properties of the flow of $(S)$ in the vicinity of the trivial fixed point, we decided to ignore possible complications arising from trajectories reaching simultaneously large negative values of $u_4$ and $u_{6,1}$. That was the technical purpose of the introduction of the open set $V$. 
See Figure \ref{phase_portrait2d}.

\

Similar results can be proven with the same methods in the $u_{6,1} = 0$ plane: no trajectory with $u_{6,2} > 0$ is asymptotically free, and on the contrary all trajectories with $u_{6,1} \leq 0$ are.  

\subsection{General properties of trajectories with $u_{6,1} u_{6,2} \geq 0$}

We now come back to the full three-dimensional system $(S)$. It turns out that the case in which $u_{6,1}$ and $u_{6,2}$ have identical signs can be analyzed in a way very similar to what has been achieved in the previous section. We can first infer the relative positions of the planes $H_4$, $H_{6,1}^{(1)}$ and $H_{6,2}^{(1)}$ from Lemma \ref{positivebeta} (and our numerical checks).
\begin{lemma}
 If $u_{6,1}>0$ and $u_{6,2}>0$, then:
\beq
f_{4} (u_4 , u_{6,1} , u_{6,2}) \leq 0 \quad \Rightarrow \quad \left\{
\begin{aligned}  f_{6,1} (u_4 , u_{6,1} , u_{6,2}) < 0 \\
f_{6,2} (u_4 , u_{6,1} , u_{6,2}) < 0
\end{aligned}
\right.\eeq
 If $u_{6,1}<0$ and $u_{6,2}<0$, then:
\beq
f_{4} (u_4 , u_{6,1} , u_{6,2}) \geq 0 \quad \Rightarrow \quad \left\{
\begin{aligned}  f_{6,1} (u_4 , u_{6,1} , u_{6,2}) > 0 \\
f_{6,2} (u_4 , u_{6,1} , u_{6,2}) > 0
\end{aligned}
\right.\eeq
\end{lemma}
%This allows us to define the disjoint open sets
%\begin{align}
% U_{I} &= \left\{ (u_4 , u_{6,1}, u_{6,2} ) \vert  u_{6,1} > 0, \, u_{6,2} > 0, \, f_{6,1} (u_4 , u_{6,1}, u_{6,2})>0 \,, \right. \\
% & \left. \qquad f_{6,2} (u_4 , u_{6,1}, u_{6,2})>0 \right\} \nn \\
% U_{II} &= \left\{ (u_4 , u_{6,1} , u_{6,2} ) \vert f_{6,1} (u_4 , u_{6,1}, u_{6,2})<0 ,\,f_{4} (u_4 , u_{6,1}, u_{6,2})>0 \right\} \\
%& \qquad  \cup  \left\{ (u_4 , u_{6,1} , u_{6,2} ) \vert f_{6,2} (u_4 , u_{6,1}, u_{6,2})<0 ,\,f_{4} (u_4 , u_{6,1}, u_{6,2})>0 \right\} \nn \\
% U_{III} &= \left\{ (u_4 , u_{6,1} , u_{6,2} ) \vert u_{6,1} > 0 ,\, u_{6,2} > 0 ,\, f_4 (u_4 , u_{6,1}, u_{6,2})<0 \right\} \\
% U_{IV} &= \left\{ (u_4 , u_{6,1} , u_{6,2} ) \vert u_{6,1} < 0 ,\, u_{6,2} < 0 ,\, f_4 (u_4 , u_{6,1}, u_{6,2})>0  \right\} \\
% U_{V} &= \left\{ (u_4 , u_{6,1} , u_{6,2} ) \vert f_{6,1} (u_4 , u_{6,1}, u_{6,2})>0 ,\, \right.\\
% & \left. \qquad f_{6,2} (u_4 , u_{6,1}, u_{6,2})>0 ,\, f_{4} (u_4 , u_{6,1}, u_{6,2})<0 \right\} \nn \\
% U_{VI} &= \left\{ (u_4 , u_{6,1} , u_{6,2} ) \vert u_{6,1} < 0,\,u_{6,2} < 0,\, f_{6,1} (u_4 , u_{6,1}, u_{6,2})<0 \right\} \\
%& \qquad  \cup  \left\{ (u_4 , u_{6,1} , u_{6,2} ) \vert u_{6,1} < 0,\,u_{6,2} < 0,\,f_{6,2} (u_4 , u_{6,1}, u_{6,2})<0 \right\} \nn 
%\end{align}
%which moreover verify
%\begin{align}
% \overline{U_I \cup U_{II} \cup U_{III}} = \{ (u_4, u_{6,1},u_{6,2}) \vert u_{6,1} \geq 0 \,, u_{6,2} \geq 0 \} \,,\\
% \overline{U_{IV} \cup U_{V} \cup U_{VI}} = \{ (u_4, u_{6,1},u_{6,2}) \vert u_{6,1} \leq 0 \,, u_{6,2} \leq 0 \}\,. 
%\end{align}
This allows us to define the disjoint open sets
\begin{align}
 U_{I} &= \left\{ u_{6,1} > 0, \, u_{6,2} > 0, \, f_{6,1} >0 \,, f_{6,2} >0 \right\}  \\
 U_{II} &= \left\{ u_{6,1} > 0, \, u_{6,2} > 0, \, f_{6,1} <0 ,\,f_{4} >0 \right\} \\
& \qquad  \cup  \left\{ u_{6,1} > 0, \, u_{6,2} > 0, \, f_{6,2} <0 ,\,f_{4} > 0 \right\} \nn \\
 U_{III} &= \left\{ u_{6,1} > 0 ,\, u_{6,2} > 0 ,\, f_4 (u_4 , u_{6,1}, u_{6,2})<0 \right\} \\
 U_{IV} &= \left\{  u_{6,1} < 0 ,\, u_{6,2} < 0 ,\, f_4 (u_4 , u_{6,1}, u_{6,2})>0  \right\} \\
 U_{V} &= \left\{ u_{6,1} < 0 ,\, u_{6,2} < 0 ,\, f_4 (u_4 , u_{6,1}, u_{6,2})<0  \right\} 
\end{align}
which moreover verify
\begin{align}
 \overline{U_I \cup U_{II} \cup U_{III}} = \{ (u_4, u_{6,1},u_{6,2}) \vert u_{6,1} \geq 0 \,, u_{6,2} \geq 0 \} \,,\\
 \overline{U_{IV} \cup U_{V}} = \{ (u_4, u_{6,1},u_{6,2}) \vert u_{6,1} \leq 0 \,, u_{6,2} \leq 0 \}\,. 
\end{align}

\

In the following $\bu$ denotes a maximal solution of $(S)$. We can first show that in the $\{ u_{6,1}>0,\, u_{6,2} > 0 \}$ subspace, all trajectories must intersect $U_{II}$ at early times.
\begin{lemma}\label{lem_3d1}
 Let $t_1 \in \mathbb{R}$. 
 If $\bu(t_1) \in U_{I} \cup U_{III}$, then one can find $t_0 < t_1$ such that $\bu(t_0) \in U_{II}$.
\end{lemma}
\begin{proof}
Let us first assume that $\bu(t_0) \in U_{III}$. Since $\dot{u}_4$, $\dot{u}_{6,1}$ and $\dot{u}_{6,2}$ are all negative in this region, we can repeat the argument of Lemma \ref{lem_2d1} and conclude that $\bu(t)$ cannot remain within $\overline{U_{III}}$ at all times $t\leq t_1$, and hence must enter region $U_{II}$.

Consider now the non-trivial case in which $\bu(t_0) \in U_{I}$. In this region one has $\dot{u}_4 > 0$, $\dot{u}_{6,1} > 0$ and $\dot{u}_{6,2} > 0$. Moreover Lemma \ref{positivebeta} (supplemented with our numerical checks) ensures that:
\beq
\left\{
\begin{aligned}
 \frac{\dot{u}_{6,1}}{u_{6,1}} &\leq K_1 \dot{u}_{4}  \\
\frac{\dot{u}_{6,2}}{u_{6,2}} &\leq K_2 \dot{u}_{4}  
 \end{aligned}
 \right.
 \eeq
 where $K_1 = \frac{C_4}{M-1} > 0$ and $K_2 = \frac{D_4}{M-1} > 0$. Hence we can run the same argument as in Lemma \ref{lem_2d2} and conclude that $\bu(t)$ cannot remain in $U(I)$ at all times $t \leq t_1$ without yielding a contradiction. 
\end{proof}

It hence appears that there are no asymptotically free trajectories with both $u_{6,1} > 0$ and $u_{6,2} > 0$.
\begin{proposition}\label{propo_positif}
Let $t_0 \in \mathbb{R}$. There exists no solution $\bf{u}$ of $(S)$ such that:
\beq
\left\{ \begin{aligned} u_{6,1}(t_0) &> 0 \\
u_{6,2}(t_0) &> 0 \end{aligned} 
\right. \qquad {\rm{and}} \qquad \underset{t \to - \infty}{\lim} \bu(t) = 0 \,.
\eeq
\end{proposition}
\begin{proof}
One only needs to check that the region $U_{II}$, in which $\dot{u}_{4} > 0$, is stable with respect to the time-reversed flow of $(S)$. This is guaranteed by the following facts:
\beq
\begin{pmatrix}
 \dot{u}_{4} \\ \dot{u}_{6,1} \\ \dot{u}_{6,2}
\end{pmatrix}
\cdot \begin{pmatrix}
 C_{4} \\ C_{6,1} \\ C_{6,2}
\end{pmatrix}
= C_4 f_4  + C_{6,1} u_{6,1} f_{6,1}  + C_{6,2} u_{6,2} f_{6,2} \geq 0
\eeq
on $\{ u_{6,1} > 0 , \, u_{6,2} > 0 , \, f_{6,1} = 0 , \, f_{6,2} \geq 0 \}$;
\beq
\begin{pmatrix}
 \dot{u}_{4} \\ \dot{u}_{6,1} \\ \dot{u}_{6,2}
\end{pmatrix}
\cdot \begin{pmatrix}
 D_{4} \\ D_{6,1} \\ D_{6,2}
\end{pmatrix}
= D_4 f_4  + D_{6,1} u_{6,1} f_{6,1}  + D_{6,2} u_{6,2} f_{6,2} \geq 0
\eeq
on $\{ u_{6,1} > 0 , \, u_{6,2} > 0 , \, f_{6,1} \geq 0 , \, f_{6,2} = 0 \}$; and
\beq
\begin{pmatrix}
 \dot{u}_{4} \\ \dot{u}_{6,1} \\ \dot{u}_{6,2}
\end{pmatrix}
\cdot \begin{pmatrix}
 1 \\ 2 \sqrt{2 \pi} \\ 4 \sqrt{2 \pi}
\end{pmatrix}
= 2 \sqrt{2 \pi} u_{6,1} f_{6,1}  + 4 \sqrt{2 \pi} u_{6,2} f_{6,2} \leq 0
\eeq
on $\{ u_{6,1} > 0 , \, u_{6,2} > 0 , \, f_{4} = 0 \}$.
\end{proof}

\

The previous proposition is the \emph{main result} of this paper: the hypothesis that the Gaussian fixed point is UV stable against perturbations $u_{6,1} > 0$ and $u_{6,2} > 0$, as suggested by equations (\ref{r_61}) and (\ref{r_62}), is \emph{not} confirmed by a careful inspection of the equations.
 
\

We now consider the situation in which both $u_{6,1}$ and $u_{6,2}$ are negative. As compared to the previous section, we will avoid unnecessary complications by relying even more on a suitably chosen stable neighborhood of the origin.  
\begin{lemma}\label{lem_3d2} 
 Suppose that $\bu(t_0) \in U_{V}$ for some $t_0 \in \mathbb{R}$. Then:
\beq
\underset{t \to + \infty}{\lim} u_{4}(t) = - \infty\,.
\eeq
\end{lemma}
\begin{proof}
We can first remark that $U_{V}$ is stable with respect to the flow of $(S)$. Indeed one has
\beq
\begin{pmatrix}
 \dot{u}_{4} \\ \dot{u}_{6,1} \\ \dot{u}_{6,2}
\end{pmatrix}
\cdot \begin{pmatrix}
 1 \\ 2 \sqrt{2 \pi} \\ 4 \sqrt{2 \pi}
\end{pmatrix}
= 2 \sqrt{2 \pi} u_{6,1} f_{6,1}  + 4 \sqrt{2 \pi} u_{6,2} f_{6,2} \leq 0
\eeq
for any $(u_4 , u_{6,1} , u_{6,2}) \in H_4$ such that $u_{6,1} \leq 0$ and $u_{6,2} \leq 0$. Hence $u_4$ is monotonically decreasing on $\left[ t_0 , + \infty \right[$, and one of the two following cases occur: a) $\underset{t \to + \infty}{\lim} u_4 (t) = C \in \mathbb{R}$, or b) $\underset{t \to + \infty}{\lim} u_4 (t) = - \infty$. 

Assuming that a) is realized, then also $\underset{t \to + \infty}{\lim} \dot{u}_4 (t) = 0$. It results that $\bu(t)$ must asymptotically reach the set $\cS_C = \left\{ u_4 = C \right\} \cap \left\{ f_{4} = 0 \right\} \cap \{ u_{6,1} \geq 0 \, ,  u_{6,2} \geq 0 \}$. But $\cS_C = \emptyset$ unless $C=0$, in which case $\cS_C = \{ 0 \}$. This implies that $\bu$ must converge to the Gaussian fixed point in $+ \infty$.

Let us now show that this behaviour is not consistent with the flow equations. When $u_{6,1} < 0$ and $u_{6,2} < 0$, Lemma \ref{positivebeta} allows to prove that:
\beq
\left\{ \begin{aligned}
\frac{\dot{u}_{6,1}}{u_{6,1}} %&= C_4 u_4 + C_{6,1} u_{6,1} + C_{6,2} u_{6,2}\geq C_4 u_4 + 2 \sqrt{2 \pi} C_4 u_{6,1} + 4 \sqrt{2 \pi} C_4 u_{6,2} \\
&\geq K_1 \dot{u}_4 \\
\frac{\dot{u}_{6,1}}{u_{6,1}} %&= D_4 u_4 + D_{6,1} u_{6,1} + D_{6,2} u_{6,2}\geq D_4 u_4 + 2 \sqrt{2 \pi} D_4 u_{6,1} + 4 \sqrt{2 \pi} D_4 u_{6,2} \\
&\geq K_2 \dot{u}_4
\end{aligned}
\right.
\eeq
where $K_1 = \frac{C_4}{M-1} > 0$ and $K_2 = \frac{D_4}{M-1} > 0$. These in turn imply
\beq
\left\{
\begin{aligned}
u_{6,1} (t) &\leq u_{6,1} (t_0) \exp\left( K_1 (u_4 (t) - u_4 (t_2) \right) \\
u_{6,2} (t) &\leq u_{6,2} (t_0) \exp\left( K_2 (u_4 (t) - u_4 (t_2) \right)
\end{aligned}
\right.
\eeq
for any $t \geq t_0$. These two inequalities are incompatible with $\underset{t \to + \infty}{\lim} \bu (t) = 0$.

Hence the situation b) is the one which actually occurs.
% 
% - $U_{IV}$ is stable
% - $u_4$ is therefore monotonic and either converges to a finite value or diverges
% - assuming that $u_4$ converges $C$, $\dot{u}_4$ converges to $0$ and therefore $\bu$ is asymptotically arbitrarily close to the set $\left\{ u_4 = \lambda \right\} \cap \left\{ f_{4} (u_4, u_{6,1}, u_{6,2})= 0 \right\} \cap \{ u_{6,1} \geq 0 \, ,  u_{6,2} \geq 0 \}$. But this set reduces to $\{ 0 \}$ and therefore $\bu$ must converge to the Gaussian fixed point.
% - This type of convergence is however excluded by the bounds analogous to the ones derived before.
% - Hence $u_4$ diverges.
\end{proof}

\begin{proposition}\label{propo_negatif}
Let $\bu$ be a solution of $(S)$ and $t_0 \in \mathbb{R}$. There exists a neighborhood $V$ of $0$ such that: 
\beq
\left\{\begin{aligned}
u_{6,1}(t_0) &\leq 0 \\
u_{6,2}(t_0) &\leq 0 \\
\bu(t_0) &\in V
\end{aligned}\right. \quad \Rightarrow \quad \underset{t \to - \infty}{\lim} \bu(t) =0\,.
\eeq
\end{proposition}
\begin{proof}
Let us define the open disc 
\begin{align}
D &= \left\{ u_{6,1} < 0 , \, u_{6,2} < 0 , \, f_4 = 0 , \, {u_4}^2 + {u_{6,1}}^2 + {u_{6,2}}^2 <1 \right\} \,.%\\
%&\qquad \cup \left\{ u_{6,1} \leq 0 , \, u_{6,2} = 0 , \, f_4 = 0 , \, {u_4}^2 + {u_{6,1}}^2 \geq 1  \right\} \\
%&\qquad  \cup \left\{ u_{6,1} = 0 , \, u_{6,2} \leq 0 , \, f_4 = 0 , \, {u_4}^2 + {u_{6,2}}^2 \geq 1  \right\}
\end{align}
The flow $\Phi( t, \cdot )$ of $(S)$ generates the set
\beq
V_0 = \Phi( \left[ 0 , + \infty \right[ , D )\,, 
\eeq
which according to Lemma \ref{lem_3d2} is infinitely extended in the direction $u_4 \to - \infty$. See Figure \ref{openset} for a qualitative representation of $D$ and $V_0$. Then
\beq
V = V_0 \cup \left\{ f_4  > 0\right\} \cup \left\{ u_{6,1}   > 0\right\} \cup \left\{ u_{6,2}  > 0\right\}
\eeq
is an open set containing the origin. Moreover, $V$ is by construction stable under the time-reversed flow of $(S)$.

Assume that $\bu(t_0) \in V \cap U_V$. If $\bu(t) \in U_V$ for all $t \leq t_0$, then $u_4$ is a monotonically decreasing function. Moreover, since $u_4$ is bounded on $V_0 \cap \{ u_4 > u_{4}(t_0 )\}$, $u_4$ must converge in $- \infty$. Then one also has $\dot{u}_4 (t) \to 0$ when $t \to - \infty$, therefore by the same argument as already invoked in Lemma \ref{lem_3d2} one must have $\underset{t \to - \infty}{\lim} \bu(t) =0$. If on the other hand $\bu( t_1 ) \in U_{IV}$ for some $t_1 \in \mathbb{R}$, then $\bu(t) \in U_{IV}$ for any $t\leq t_1$. By monotonicity and boundedness of $u_4$, $u_{6,1}$ and $u_{6,2}$, one then again concludes that $\underset{t \to - \infty}{\lim} \bu(t) =0$.

Finally, the case in which $u_{6,1}(t_0)= 0$ has already been treated in the previous section, and similar results hold for $u_{6,2}(t_0)= 0$.
% 
% - $V$: constructed out of trajectories with initial condition given by points on $H_4$ at a unit distance from the origin and with $u_6 \leq 0$.
% - Assume $\bu(t_0) \in V \cap U_V$. Must go out of $U_V$ at early times, otherwise $u_4$ converges in $- \infty$ (here we use the fact that $V$ provides bounds). By monotonicity, $\dot{u}_4$ also converges. With the same argument as before, the only possibility is that $\bu$ converges to $0$. But in this case $u_{6,1}$ and $u_{6,2}$ are also monotonically increasing, hence the contradiction. 
\end{proof}

\begin{figure}[h]
  \centering
  \includegraphics[scale=0.8]{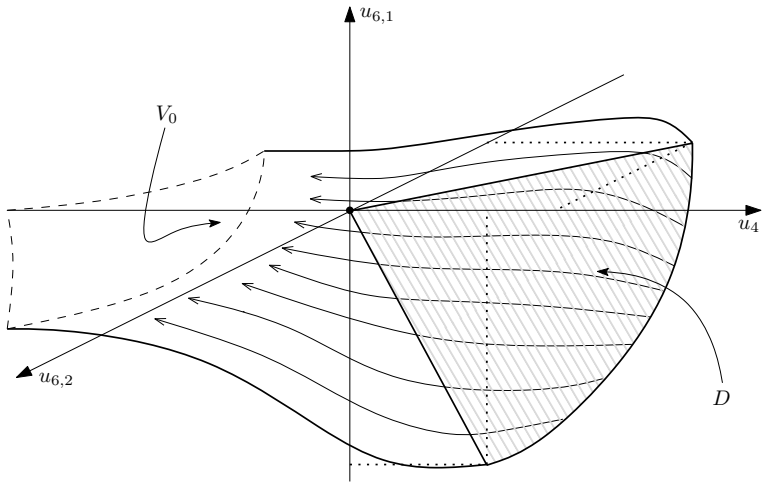}     
  \caption{The action of the flow of $S$ on $D$ generates the unbounded set $V_0$.}\label{openset}
\end{figure}

\subsection{Trajectories with $u_{6,1} u_{6,2} < 0$}

The phase portrait in regions in which $u_{6,1}$ and $u_{6,2}$ have opposite signs is more difficult to analyze. 

On the one hand, there is no obvious UV stable open set in which $\vert u_4 \vert$, $\vert u_{6,1} \vert$ and $\vert u_{6,2} \vert$ are simultaneously strictly increasing under the flow of $(S)$. Hence there seems to be no immediately identifiable argument proving that asymptotically free trajectories do exist in this region. On the other hand, none of the bounds we relied on in the previous sections (in particular to exclude asymptotic freedom when $u_{6,1} >0$ and $u_{6,2} >0$), are easily generalizable. Therefore the non-existence of asymptotically free trajectories in this sector seems equally hard to establish.  

In a sense the dynamical system $(S)$ becomes truly three dimensional when $u_{6,1} u_{6,2} < 0$, with all the complications this may lead to. A comprehensive study in this regime is therefore beyond the scope of the present work, and we leave the question open for future investigations.

% \
% 
% To begin with, $\vert u_4 \vert$, $\vert u_{6,1} \vert$ and $\vert u_{6,2} \vert$ are simultaneously strictly increasing under the flow of $(S)$, if and only if: $f_{6,1} > 0$, $f_{6,2} > 0$, and $f_4 u_4 > 0$. One can show that the maximal open sets in which these conditions hold are:
% \begin{align}
%  U_{I} &= \left\{ u_4 > 0 , \, u_{6,1} < 0 \, , u_{6,2} > 0 \, , f_4 > 0 \right\}\,, \\
%  U_{II} &= \left\{ u_4 < 0 , \, u_{6,1} > 0 \, , u_{6,2} < 0 \, , f_4 > 0 \right\}\,.
% \end{align}
% This is a simple consequence of the relative positions of the planes $H_4$, $H_{6,1}^{(1)}$ and $H_{6,2}^{(1)}$, which can be checked algebraically. If one of the two sets $U_I$ and $U_{II}$ were stable under the time-reversed flow of $(S)$, we could easily conclude that trajectories with $u_{6,1} u_{6,2}$ do exist. Unfortunately, both of them are actually unstable:
% \beq
% \begin{pmatrix}
%  \dot{u}_{4} \\ \dot{u}_{6,1} \\ \dot{u}_{6,2}
% \end{pmatrix}
% \cdot \begin{pmatrix}
%  C_4 \\ C_{6,1} \\ C_{6,2}
% \end{pmatrix}
% = 2 \sqrt{2 \pi} u_{6,1} f_{6,1}  + 4 \sqrt{2 \pi} u_{6,2} f_{6,2} \leq 0
% \eeq
% on $\left\{ u_4 > 0 , \, u_{6,1} < 0 \, , u_{6,2} > 0 \, , f_4 > 0 \right\}$
% 
% - natural sets for AS
% 
% - but, not stable
% 
% - conjecture: no AS trajectories

\section{Conclusion and outlook}\label{sec:conclusion}

Let us start with a summary of what has been achieved in this paper. The main goal was to develop a Wilsonian renormalization group picture of the renormalizable TGFT of \cite{cor_su2}. This required the introduction of dimensionless parameters, with canonical dimensions inferred from the power-counting. The perturbative expansion was then performed with respect to the dimensionless parameters rather than the dimensionful ones, which is responsible for subtle departures with respect to previous analyzes of such models \cite{josephaf, samary_beta, cor_su2}: the class of graphs which are divergent in this sense is enlarged. General (and formal) flow equations, involving melonic graphs only, have been derived.

We then applied these general equations in the vicinity of the Gaussian fixed point, where perturbation theory can be applied. We first recovered the splitting between renormalizable and non-renormalizable coupling constants, which allowed us to restrict ourselves to four independent parameters: the mass, a $4$-valent coupling constant, and two $6$-valent coupling constants. The eigendirections of the linearized system were then computed, reducing further the set of independent couplings in the deep UV region: the mass and $4$-valent coupling constant become linearly dependent of the two marginal $6$-valent coupling constants at high scales. In order to fully understand the flow in the vicinity of the Gaussian fixed point, the perturbative expansion was pushed to second order in the marginal parameters. Resorting to continuous methods, we eventually derived some general properties. Our main result is the following (see Proposition \ref{propo_positif}): 
\begin{center}
\emph{Renormalization group trajectories with non-zero and positive marginal parameters are not asymptotically free.}
\end{center}
For such trajectories, the existence of a Landau pole can therefore \emph{not} be excluded. On the contrary, we also proved that \emph{the Gaussian fixed point is UV attractive with respect to negative perturbations of the marginal parameters} (see Proposition \ref{propo_negatif}). However this corresponds to a regime in which one does not expect to be able to define the partition function rigorously.

%As a result the $\beta$-coefficients, which are analogous to the $\beta$-functions when the UV cut-off is varied in a continuous manner, are strictly positive. This shows that the trajectories such that the two marginal parameters remain strictly positive cannot be asymptotically free. 

\

A few comments are in order as far as asymptotic freedom is concerned. As was made clear in previous works \cite{josephaf, samary_beta}, asymptotic freedom is to be expected in such models, due to an enhanced wave-function renormalization as compared to ordinary scalar field theories. Indeed, tadpoles are not exactly local (i.e. tensorial) in TGFTs, and therefore wave-function counter-terms appear already at first order in the perturbative expansion. If big enough, these terms can balance the first-order coupling constant counter-terms and make the $\beta$-function negative. This is exactly what happens in \cite{josephaf, samary_beta}, and also what seems to happen in the present paper, see equations (\ref{r_61}) and (\ref{r_62}). However, our analysis shows that the non-linear terms appearing in equations (\ref{r_61}) and (\ref{r_62}) spoil these expectations. The non-trivial dynamics of the $4$-valent coupling constant $u_4$ and the two marginal coupling constants $u_{6,1}$ and $u_{6,2}$ yield a complicated phase portrait (already in the vicinity of the origin), as is well illustrated by Figure \ref{phase_portrait2d}. 

%However, according to our analysis, an additional relation between the marginal and $4$-valent coupling constants must be taken into account in the deep UV region, see equation (\ref{u4_UV}). Combined with the $u_4 \times u_6$ terms appearing in equations (\ref{r_61}) and (\ref{r_62}), this finally yields positive $\beta$-functions for the marginal coupling constants (see equations (\ref{beta_final1}) and (\ref{beta_final1}), as well as lemma \ref{positivebeta}). 

%It is important to note that the two crucial observations responsible for the positivity of the $\beta$-coefficients originate from the introduction of dimensionless coupling constants: a) first, equation (\ref{u4_UV}) is deduced from the diagonalization of the linearized flow, which can only be made sense of if the parameters are dimensionless\footnote{Indeed, the dimensionful parameters would receive large corrections after a single action of the renormalization group, and hence the perturbative expansion would be ill-defined.}; b) second, the cross-terms $u_4 \times u_6$ appearing in (\ref{r_61}) and (\ref{r_62}) are generated by graphs which are divergent in the sense of $\omegab$, but convergence in the sense of $\omega$. 

Note also that our construction relies heavily on the introduction of dimensionless parameters. While this is required and well-understood in ordinary quantum field theory, the same procedure in TGFT deserves more clarification. It is indeed based on a more abstract notion of dimensionality given by the power-counting, and not immediately interpretable in terms of physical dimensions such as lengths, energies etc. It depends in particular on the choice of propagator, of Laplace-type in our situation, which is one other aspect of TGFTs for which a clear physical interpretation is lacking. Note also that in renormalizable models with quartic interactions, a computation along the lines of \cite{josephaf, samary_beta} and the discrete renormalization group picture developed in this paper would match. Therefore we expect such models to be asymptotically free. Good candidates in TGFTs with gauge invariance condition are rank-$3$ models on dimension $4$ groups \cite{4-eps}, rank-$4$ models on dimension $2$ groups, and rank-$6$ models on dimension $1$ groups (see the full classification of renormalizable models in \cite{cor_su2, thesis}).

Finally, while our analysis excludes the existence of asymptotically free trajectories in this model when $u_{6,1} > 0$ and $u_{6,2} > 0$, it does not prove the existence of a Landau pole either. Moreover, we did not include a detailed investigation of the situation in which $u_{6,1}$ and $u_{6,2}$ have opposite signs, which as we have argued is probably involved, but might simultaneously support asymptotic freedom and convergence of the path-integral. Therefore the question of the existence of this model with no cut-off remains open, and could be investigated further. In particular, our calculations are consistent with the existence of a non-trivial UV fixed point with $u_2 > 0$, $u_4 < 0$, $u_{6,1} > 0$ and $u_{6,2} > 0$. More generally, the existence of non-trivial fixed points for this model will be first investigated by means of an $\varepsilon$-expansion \cite{4-eps}. Similarly to ordinary scalar field theories, one can indeed construct renormalizable $\vphi^4$ TGFT models on groups of dimension $4$ (see \cite{cor_su2} and \cite{thesis}), and then analytically continue the dimension to $4 - \varepsilon$. This procedure should give first hints about the $\SU(2)$ model. It will then be desirable to investigate the same questions with different and more adapted tools, such as the Functional Renormalization Group \cite{astrid_tim, tt4}.

\section*{Acknowledgments}

I would like to thank Dine Ousmane Samary and Joseph Ben Geloun for discussions, especially in the early stages of this work. I am also grateful to Daniele Oriti and Vincent Rivasseau for their comments on a draft version of this paper. The final version greatly benefited from insightful remarks I owe to Razvan Gurau and an anonymous referee, which I wish to acknowledge.

% appendices %

\appendix

\section*{Appendix: Laplace approximation}

Two properties are involved in the evaluation of the amplitudes at large scales. The first one is the short time asymptotics of the heat-kernel
\beq\label{hk_as}
K_{\alpha}(g) \underset{\alpha \to 0}{\sim} k \frac{e^{- \frac{\Psi(g)^2}{\alpha}}}{\alpha^{3/2}} \frac{\Psi(g)}{\sin \Psi(g)} \,,
\eeq
valid on any compact which does not contain $- \one$. $\Psi(g) \in \left[ 0 , \pi \right]$ is the class angle of $g$, defined by
\beq
g = \e^{X_g} \; ; \qquad \vert X_g \vert = \frac{\Psi(g)}{2}\,, 
\eeq
and $\vert X \vert = \sqrt{\rm{tr} (X^{\dagger} X)}$ is the norm of $X \in \su(2)$. We can use the orthonormal basis $\{ \tau_k \} = \{ \rm{i} \frac{\sigma_k}{2} \}$, with $\{ \sigma_k \}$ the Pauli matrices. We then have:
\beq
X_g \equiv \sum_k X_{g,k} \, \tau_k  \;, \qquad  \vert X_g \vert^2 = 4 \Psi(g)^2 = \sum_k {X_{g,k}}^2 \;.  
\eeq 
In what follows, we shall adopt a vectorial notation for scalar products between Lie algebra elements. 

\

The second property we use in the computations is that arbitrarily close to the identity, the normalized Haar measure $\extd g$ is equivalent to the Lebesgue measure on the Lie algebra $\su(2)$:
\beq\label{haar_vect}
\extd g \underset{g \to \one}{\approx} \frac{1}{16  \pi^2 } \, \extd X_g \equiv \frac{1}{16  \pi^2 }  \,\extd X_{g,1} \extd X_{g,2} \extd X_{g,3} \,.
\eeq

\

Consider now an integral of the form
\beq
\cI(\{ \alpha(f) \}) = \int_{\SU(2)} [\extd h_l ]^L \prod_{f \in F} K_{\alpha(f)}\left( \overrightarrow{\prod_{l \in f}} {h_l}^{\varepsilon_{lf}} \right) \,,
\eeq
where: $L$ and $F$ are finite sets such that, for any $f \in F$, $f \subset L$; $\varepsilon_{lf} = \pm 1$ or $0$; and $\alpha(f) > 0$ for any $f \in F$. Let us moreover assume that:
\beq
\left( \forall f \in F\,, \qquad  \overrightarrow{\prod_{l \in f}} {h_l}^{\varepsilon_{lf}} = \one \right) \; \Rightarrow \; \left( \forall l \in L\,, \qquad h_l = \one \right)\,.
\eeq
When the parameters $\alpha(f)$ are simultaneously sent to $0$, the integrand of $\cI(\{ \alpha(f) \})$ is more and more peaked around the configurations such that
\beq
\overrightarrow{\prod_{l \in f}} {h_l}^{\varepsilon_{lf}} = \one
\eeq
for any $f \in F$. By hypothesis, the unique such configuration is $\{ h_l = \one \}$. One can therefore linearize the variables $h_l$ around the identity to accurately estimate the integral. Using this idea in combination with (\ref{hk_as}) and (\ref{haar_vect}) one finds that:
\bes
\cI(\{ \alpha(f) \}) &\underset{\alpha(f) \to 0}{\sim}& (4\pi)^{\frac{F}{2} - 2 L} \left( \prod_{l \in L} \alpha_l \right)^{-3/2} \\ 
&& \times \int_{\mathbb{R}^3} [\extd X_l ]^L \prod_{f \in F} \exp\left( \frac{1}{\alpha(f)} \left( \sum_{l \in L} \varepsilon_{lf} X_l \right)^2 \right)\,. \nn
\ees
Hence determining the asymptotic behaviour of $\cI(\{ \alpha(f) \})$ reduces to the computation of a Gaussian integral. Such formula have already been used in the literature, for instance in \cite{Valentin_Joseph} and \cite{melonic_phase}.

% acknowledgements and bibliography %

\bibliographystyle{hunsrt}
\bibliography{biblio}

\end{document}